\PassOptionsToPackage{dvipsnames}{xcolour}
\documentclass[acmsmall,nonacm]{acmart} %
\setcopyright{cc}
\setcctype{by}
\makeatletter
\let\@ACM@checkaffil\relax
\makeatother
\usepackage{amsthm} %
\usepackage{proof}
\usepackage{thm-restate}
\usepackage{mathpartir}
\usepackage{booktabs}
\usepackage{amsmath}
\usepackage{balance}
\usepackage{nicefrac}
\usepackage{url}
\usepackage{xspace}
\usepackage{hyperref}
\usepackage[capitalise]{cleveref}
\usepackage{todonotes}
\usepackage[normalem]{ulem}  %
\usepackage{suffix}
\usepackage[inline]{enumitem}  %
\usepackage{afterpage}
\usepackage{stmaryrd} %
\usepackage{ifthen}
\usepackage{mathtools}
\usepackage{tcolorbox}
\mathtoolsset{showonlyrefs=true}

\usepackage[final]{fixme}
\fxsetup{theme=color,mode=multiuser}
\FXRegisterAuthor{JL}{aJL}{JL}
\usepackage{afterpage}

\usepackage[utf8]{inputenc}
\DeclareUnicodeCharacter{1438}{\pa}
\DeclareUnicodeCharacter{1433}{\po}
\DeclareUnicodeCharacter{1428}{\stroke}

\usepackage[final]{listings}
\usepackage{multicol}
\usepackage{golang}     %

\newcommand{\lstCodeSize}{\footnotesize\ttfamily} %

\lstdefinestyle{mygo}%
{%
language=Go,
morekeywords={any,change},
moredelim=[is][\color{gray}]{``}{``},
}

\usepackage{tikz}  %
\usetikzlibrary{arrows,shapes,calc,fit,patterns,backgrounds,trees,arrows,positioning,shapes.multipart}

\theoremstyle{acmplain}

\newtheorem{theorem}{Theorem}[section]
\newtheorem{lemma}[theorem]{Lemma}
\newtheorem{definition}[theorem]{Definition}

\newtheorem{example}[theorem]{Example}

\newenvironment{Figure}{
  \begin{figure}
    \small
}{
  \end{figure}
}
\newenvironment{ruled}{
  \hrule
  \vspace{1ex}
}{
  \vspace{1ex}
  \hrule
  \vspace{-2ex}
}

\newcommand{\ROUNDTWO}[1]{#1}

\newcommand{\ttilde}{{\thicksim}}

\newcommand{\Universe}{\mathbb{U}}

\newcommand{\types}{\meta{types}}

\DeclareMathOperator{\typeor}{\!\mid\!}

\newcommand{\st}{\, \mid \,}
\newcommand{\squeeze}{\hspace{-1.5em}}

\newcommand{\ov}{\overline}

\newcommand{\kw}[1]{\texttt{\bf #1}}
\newcommand{\id}[1]{\texttt{#1}}
\newcommand{\meta}{\mathit}
\newcommand{\ok}{~\meta{\checkmark}}
\newcommand{\nice}{\ok_{\!\text{b}}} %
\newcommand{\supernice}{\ok_{\!\text{m}}} %
\newcommand{\okelmt}{\ok_{\!\text{i}}}

\newcommand{\fields}{\meta{fields}}

\newcommand{\methods}{\meta{methods}}

\newcommand{\ubase}{\meta{base}}
\newcommand{\uinterface}{\meta{iface}}
\newcommand{\ustruct}{\meta{struct}}

\newcommand{\opdom}{\meta{dom}}

\newcommand{\typair}[2]{{#1} @ { #2}}
\newcommand{\ghostytypair}[2]{\ghost{#1} \ghost{@} { #2}}

\newcommand{\tybox}[2]{\langle #1 \rangle} %
\newcommand{\tycomp}[2]{\lfloor #1 \rfloor}  %
\newcommand{\revtycomp}[1]{\tycomp{#1}{}^{-1}}  %

\newcommand{\dom}{\mathit{dom}}

\newcommand{\unonconstraint}{\meta{noUnion}}
\newcommand{\utyvar}{\meta{tyvar}}
\newcommand{\underlying}{\meta{under}}

\newcommand{\vtype}{\meta{type}}
\newcommand{\ctype}{\meta{ctype}}

\newcommand{\mbody}{\meta{body}}
\newcommand{\mjump}{\meta{body_L}}
\newcommand{\tdecls}{\meta{tdecls}}
\newcommand{\mdecls}{\meta{mdecls}}
\newcommand{\sigdecls}{\meta{sigdecls}}
\newcommand{\bounds}{\meta{bounds}}

\newcommand{\unique}{\meta{unique}}
\newcommand{\distinct}{\meta{distinct}}
\newcommand{\type}{\kw{type}}
\newcommand{\struct}{\kw{struct}}
\newcommand{\interface}{\kw{interface}}
\newcommand{\func}{\kw{func}}
\newcommand{\return}{\kw{return}}
\newcommand{\package}{\kw{package}}
\newcommand{\main}{\kw{main}}

\newcommand{\maintype}[1]{#1{\color{black} \downharpoonright_1}}
\newcommand{\uppertype}[1]{#1{\color{black} \downharpoonright_2}}

\newcommand{\un}{\id{\textunderscore}}
\newcommand{\prog}{\rhd}
\newcommand{\br}[1]{\id{\{}#1\id{\}}}
\newcommand{\lst}[1]{[#1]}
\newcommand{\set}[1]{\{#1\}}

\newcommand{\assg}{\prec:}

\newcommand{\becomes}{\longrightarrow}
\newcommand{\becomesast}{\becomes^{\ast}}
\newcommand{\by}{\mathbin{:=}}

\newcommand{\black}[1]{{\color{black}#1}}
\newcommand{\comma}{,\,}
\newcommand{\stoup}{;\,}
\newcommand{\Hole}{\Box}

\newcommand{\goinl}[1]{\lstinline[style=mygo,basicstyle=\small\ttfamily]+#1+}
\newcommand{\goinlf}[1]{\lstinline[style=mygo]+#1+}
\lstnewenvironment{golst}{\lstset{style=mygo}}{}
\lstnewenvironment{golsttwo}{\lstset{style=mygo,multicols=2}}{}

\newcommand{\notmonomorphisable}{~{\notmonomorphisable*}}
\WithSuffix\newcommand\notmonomorphisable*{\meta{nomono}}  %

\newcommand{\imp}{\mathbin{\texttt{<} \mkern-1mu \raisebox{0.2ex}{\texttt{:}}}}
\newcommand{\notimp}{\mathbin{\ooalign{$\imp$\cr\hidewidth$/$\hidewidth}}}

\DeclareMathOperator{\IMPLOP}{\imp}

\newcommand{\MKTABLE}[3]{\mathsf{mkTable}_{#3}(#1,#2)}

\newcommand{\CHANGETYPE}[3]{#1.(\ghost{#3} \rightarrow #2)}
\newcommand{\STATICCHANGE}[3]{#1.(\ghost{#3} \twoheadrightarrow \ghost{#2})}
\newcommand{\MAKETYPE}[4]{#1.({#3} \rightarrow_{#4} \ghost{#2})}

\newcommand{\smsep}{{\#}}

\newcommand{\ghost}[1]{{\color{red!50}#1}}

\newcommand\eqdef{\stackrel{\text{def}}{=}}
\newcommand{\eqval}[2]{#1 \multimap #2}
\newcommand{\eqcast}[3][\Delta \stoup \Gamma]{#1 \vdash #2 \bumpeq #3}
\newcommand{\ceqcast}[2]{#1 \bumpeq #2}

\newif\iflong %
\longtrue %
\newcommand{\appendixref}[2]{\iflong {#2}~\ref{#1}\else the online appendix\fi}

\begin{document}

\pdfstringdefDisableCommands{\def\titlenote#1{}\let\footnotemark\relax}
\title{Welterweight Go: Boxing, Structural Subtyping, and Generics \\ (Extended Version)\texorpdfstring{\footnotemark[1]}{}\titlenote{This is an extended version of the paper published in the Proceedings of the ACM on Programming Languages, Volume~10, Issue~POPL (January 2026), Article~79. \url{https://doi.org/10.1145/3776721}.}}

\author{Raymond Hu}
\orcid{0000-0003-4361-6772}
\affiliation{%
  \institution{Queen Mary University of London}
  \city{London}
  \country{United Kingdom}
}
\email{r.hu@qmul.ac.uk}

\author{Julien Lange}
\orcid{0000-0001-9697-1378}
\affiliation{%
  \institution{Royal Holloway, University of London}
  \city{Egham}
  \country{United Kingdom}
}
\email{Julien.Lange@rhul.ac.uk}

\author{Bernardo Toninho}
\orcid{0000-0002-0746-7514}
\affiliation{%
  \institution{Instituto Superior Técnico, University of Lisbon}
  \city{Lisbon}
  \country{Portugal}
}
\affiliation{%
  \institution{INESC-ID}
  \city{Lisbon}
}
\email{bernardo.toninho@tecnico.ulisboa.pt}

\author{Philip Wadler}
\orcid{0000-0001-7619-6378}
\affiliation{%
  \institution{University of Edinburgh}
  \city{Edinburgh}
  \country{United Kingdom}
}
\email{wadler@inf.ed.ac.uk}

\author{Robert Griesemer}
\orcid{0009-0001-1887-043X}
\affiliation{%
  \institution{Google}
  \city{Mountain View}
  \country{USA}
}
\email{gri@google.com}

\author{Keith Randall}
\orcid{0009-0000-9075-1737}
\affiliation{%
  \institution{Google}
  \city{Mountain View}
  \country{USA}
}
\email{keithr@alum.mit.edu}

\keywords{Go, Generics, Boxing}

\begin{abstract}
  Go's unique combination of structural subtyping between generics and types
  with non-uniform runtime representations presents significant challenges for
  formalising the language.

  We introduce WG (Welterweight Go), a core model of Go that
  captures key features excluded by prior work, including underlying types,
  type unions and type sets, and proposed new features, such as generic methods.
  We also develop LWG, a lower-level language that models Go's runtime
  mechanisms, notably the distinction between raw struct values and
  interface values that carry runtime type information (RTTI).

  We give a type-directed compilation from WG to LWG that demonstrates
  how the proposed features can be implemented while observing important
  design and implementation goals for Go: compatibility with separate compilation, and no runtime code generation.
  Unlike existing approaches based on static monomorphisation, our compilation
  strategy uses runtime type conversions and adaptor methods to handle the
  complex interactions between structural subtyping, generics, and Go's
  runtime infrastructure.
\end{abstract}

\ccsdesc[500]{Theory of computation~Program semantics}
\ccsdesc[500]{Theory of computation~Type structures}
\ccsdesc[500]{Software and its engineering~Polymorphism}

\maketitle

\renewcommand{\MathparLineskip}{\lineskiplimit=0.8em\lineskip=0.75em plus 0.1em}

\section{Introduction}
\label{sec:intro}

Go is a popular programming language that is widely used in industry.
A key characteristic of Go is \emph{structural} subtyping of interfaces, unlike
other memory managed OO languages such as Java and C\# that are based on
nominal typing.

The first work on formalising Go
by~\citet{DBLP:journals/pacmpl/GriesemerHKLTTW20} tackled the extension of Go
with \emph{generics} (bounded parametric polymorphism).
It presented a system that integrates Go's structural subtyping and generics
in a manner that is compatible with static monomorphisation (as opposed to,
e.g., erasure in Java, and dynamic monomorphisation in C\#).
This shaped the release of generics in Go~1.18 in 2022, described as
the biggest change in its history~\cite{griesemer2022introduction}.

This paper tackles several problems in the ongoing design and formal investigation
of Go.
\begin{itemize}[leftmargin=*, itemindent=1em]
\item
  \ROUNDTWO{
    We introduce WG (Welterweight Go). Its purpose is to provide a
    minimal formal model that captures essential type system features
    of Go needed to study sound compilation strategies for generics,
    and to investigate new features proposed for future versions of
    Go, notably generic methods and method set intersections for
    generic interfaces.
    Structural typing in Go implies, for example, that methods can
    accept differently-named structs with the same underlying
    structure: several subtleties related to various kinds of type coercion are needed,
    which are absent from previous work.  We focus on generic type unions
    and type sets, and model the key areas that interact with Go's use of
    structural typing: anonymous types, underlying types, and
    primitive data and operators.
    The combination of these features is absent from models found in
    the literature.  Our design of WG %
    distils the least set of features that captures all the relevant
    interactions.
    We prove that WG satisfies preservation and progress properties
    (Theorems~\ref{thm:wg-preservation} and~\ref{thm:wg-progress}).
}
\item
We formalise a lower level language, LWG,
to model key mechanisms of the Go \emph{runtime}.
LWG captures essential elements of the design and implementation of Go---notably
the runtime distinction between struct types and interface types---that WG
(or any prior
formalism~\cite{DBLP:journals/pacmpl/GriesemerHKLTTW20,DBLP:journals/pacmpl/EllisZYS22,DBLP:journals/jfp/SulzmannW23})
on its own does not.
We develop LWG as a typed language and establish its own preservation and
progress properties (Theorems~\ref{thm:lwg-preservation} and~\ref{thm:lwg-progress}).
\item
We formalise compilation as a type-directed translation from WG to LWG and
prove its correctness by a behavioural equivalence (Theorem~\ref{thm:compbisim}).
We have designed LWG and compilation to support the proposed new
features while respecting the key criteria set by the Go Team:
compatibility with separate compilation and no code generation during
runtime.
Our compilation strategy is based on boxing and runtime type
conversions.
\end{itemize}

\noindent
\ROUNDTWO{
  Overall, this paper presents the first full account of structural
  typing for a language with generic structs, methods and
  interfaces. We establish type safety results and the design of a
  type and behaviour preserving compilation to a low-level model
  that is feasible in practice. It provides a framework for
  investigating the design, implementation and correctness of future
  extensions to the Go language and runtime.
}
Method set intersections have been proposed but not yet
implemented~\cite{go2024method,dominikh2022method}; these would allow
interfaces to contain methods common to all listed types.
Generic methods have been under consideration in a long-standing
proposal~\cite{mariomac2021parametersmethod}, but have yet to be
included in Go due to uncertainty about the best
implementation approach~\cite{taylor2021no}.
\ROUNDTWO{
  It is important for language designers to investigate if there are
  bad interactions between features. For instance, Amin and
  Tate~\cite{Amin-and-Tate-2016} have shown that Java generics are
  unsound. The issue they found could not be expressed in
  Featherweight Generic Java~\cite{Igarashi-et-al-2001}, due to its limited feature set.
}

Formalising LWG and compilation is important for ensuring that the proposed
features are indeed compatible with the requirements of Go and
demonstrating how they can be implemented in practice.
A key distinction between base or struct types and interface types in the Go
runtime is that interface types carry runtime type information (RTTI) for
dynamic operations (e.g., type assertions) whereas constants and structs do
not.
The design of Go crucially depends on the compiler being able to
\emph{statically} insert runtime operations for safely converting between the
distinct cases for constants or structs and interfaces as related by structural
subtyping.
The prior high level formalisms lack this perspective and fail to faithfully
model the Go runtime.
For example, they use terms \goinl{$S$\{...\}} to denote \emph{runtime} struct
values, and \goinl{$S$\{...\}.($T$)} for type assertions on structs; whereas
actual Go structs, as mentioned, do not carry any RTTI corresponding to the
\goinl{$S$}, and expressions that would dynamically depend on such RTTI are
thus not valid in practice.

Unlike the original core formalism
of~\citet{DBLP:journals/pacmpl/GriesemerHKLTTW20} that is based \emph{solely}
on static monomorphisation, the actual Go runtime records RTTI for generic
type arguments dynamically using dictionaries that are generated by the
compiler and passed around at
runtime~\cite{godictionaries}.
LWG models Go interfaces by pairing RTTI with method tables, capturing
a key aspect of the actual Go implementation.
We show that our approach in fact allows WG, including the proposed new
features, to be directly compiled to LWG without relying on monomorphisation
to eliminate genericity or incurring its associated restrictions (e.g., on
polymorphic recursion).

A key challenge for the compilation is to deal with structural
subtyping and generic types. The interplay between these means that
generic code must, in general, be statically compiled to handle two
distinct kinds of arguments that can be passed at runtime: base or
struct types on the one hand and interface types on the other.
Our solution builds on Go's framework for implicit runtime type conversions.
It involves generating adaptor code that we can quantify as one extra wrapper
method per method defined in the user program.
The benefits are that it supports generic methods with separate compilation,
as opposed to monomorphisation, which requires a whole program analysis, while as mentioned lifting other
present limitations of Go such as the restriction on polymorphic recursion.
In general, our compilation is more efficient and results in less code bloat
than monomorphisation, but involves more type conversions at runtime.

This work was motivated by discussions with the Go Team who
requested feedback on the design of generic type unions and type
sets. Two members of the Go Team are coauthors of this paper.

\subsubsection*{Roadmap}
Sec.~\ref{sec:overview} gives an overview of key concepts in WG, LWG and
compilation.
Sec.~\ref{sec:wg} defines WG and establishes preservation and
progress properties.
Sec.~\ref{sec:lwg} defines LWG and establishes preservation and
progress properties (independently of WG).
Sec.~\ref{sec:compilation} defines the compilation of WG into LWG and shows
that a well-typed WG program and its LWG compilation are behaviourally
equivalent.
Sec.~\ref{sec:related} discusses related work.
We have implemented a minimal prototype of WG and our compilation approach in
an accompanying artifact~\cite{wg-artifact}.
\section{Overview}
\label{sec:overview}

\subsection{WG by Example}
\label{sub:WGexample}

Figure~\ref{fig:WGexample} gives a WG example.
For readability, we use functions (as opposed to methods),
anonymous functions and function types; in our core formalism these can be
represented using interfaces, structs and methods.
\ROUNDTWO{
Type \goinl{MyInt} is defined with \emph{underlying} basic type
\goinl{int}.
\goinl{MyInt} is nominally distinct from \goinl{int}; it can be thought of as
a wrapper with no runtime overhead (akin to \goinl{newtype} in Haskell).
Type \goinl{MyFloat} is similar.
}
\goinl{MyInt} and \goinl{MyFloat} have a method \goinl{String} with
implementations that involve \emph{static type conversions} based on the
underlying types.
Interface \goinl{MyNum} contains a \emph{type union} specifying that the
only members of this type are \goinl{MyInt} and \goinl{MyFloat}.

Generic type \goinl{List} has a type parameter \goinl{a} with upper bound \goinl{any}, an
alias for the empty interface (i.e., \goinl{interface\{\}}), meaning that \goinl{a} can
be instantiated by any type. \goinl{List} is defined as an interface
supporting a single method \goinl{FoldL[b]}, a fold operation that takes a
second type parameter \goinl{b}. By structural subtyping, any type \goinl{T} that
implements \goinl{FoldL[b]} is considered a subtype of \goinl{List[T]}. Structs \goinl{Nil[a]}
and \goinl{Cons[a]} implement \goinl{List[a]}. The two fields of \goinl{Cons[a]} depend on its
type parameter.
Generic function \goinl{join} takes a \goinl{List} of any type bounded by
\goinl{MyNum}, and folds the list by calling the \goinl{String} method on each
item and concatenating them.
The idea is that the formals of the (anonymous) combining function safely
support \goinl{String} because they are bounded by \goinl{MyNum}, all members
of which implement \goinl{String}.
The \goinl{main} method calls \goinl{join} to fold a list of \goinl{MyFloat}.

The code above the line in Figure~\ref{fig:WGexample} is supported by Go, but
not by prior formalisms due to lacking underlying types, type unions and
static conversions.
The code below is WG only because Go does not yet support:
(1) generic methods (see
lines~\ref{lst:gen-meth-a},\ref{lst:gen-meth-b}-\ref{lst:gen-meth-c});
nor (2) \emph{method set intersections} for type unions which is
required to support the \goinl{String} call at line~\ref{lst:string-call}.

\begin{figure}[t]
{\lstset{numbers=left, xleftmargin=2em}
\begin{golst}
package main
type MyInt int
func (x MyInt) String() string { return strconv.Itoa(int(x)) }
type MyFloat float64
func (x MyFloat) String() string { return strconv.FormatFloat(float64(x), 'E', -1, 64) }
type MyNum interface { MyInt | MyFloat }
\end{golst}
}

\hrule

{\lstset{numbers=left, xleftmargin=2em,firstnumber=7}
\begin{golst}
type List[a any] interface { FoldL[b any](f func(v b, w a) b, z b) b } (*\label{lst:gen-meth-a}*)
type Nil[a any] struct {}
type Cons[a any] struct { head a; tail List[a] }
func (x Nil[a]) FoldL[b any](f func(v b, w a) b, z b) b {return z} (*\label{lst:gen-meth-b}*)
func (x Cons[a]) FoldL[b any](f func(v b, w a) b, z b) b {
  return x.tail.FoldL[b](f, f(z, x.head)) (*\label{lst:gen-meth-c}*)
} 
func join[a MyNum](x List[a]) String {
  return x.FoldL[string](func(v string, w a) string { return v + ", " + w.String() }, "") (*\label{lst:string-call}*)
}
func main() {
  var xs List[MyFloat] = Cons[MyFloat]{MyFloat(1), Cons[MyFloat]{MyFloat(2), Nil[MyFloat]{}}}
  fmt.Println(join[MyFloat](xs))  // Prints: ", 1E+00, 2E+00"
}
\end{golst}}
  \vspace{-0.5cm}
\caption{(top) Code supported by Go and WG; (bottom) code supported by WG only.}
\label{fig:WGexample}
\end{figure}

\subsection{LWG: A First Mini Example}

This mini example demonstrates how structural subtyping between
base or struct types and interface types in WG is handled by the Go compiler
and runtime.

\begin{golst}
type C struct { f any } $\quad$ type ${\color{blue}\mathtt{INum}}$ interface { String() string } // MyInt (etc) implements INum
func bar(x ${\color{blue}\mathtt{INum}}$) MyInt { return C{x}.f.(MyInt) } $\quad$ func main() { return bar(MyInt(42)).String() }
\end{golst}

\noindent
\ROUNDTWO{The function \goinl{bar} takes an \goinl{INum}} and wraps it in a
\goinl{C} struct with field type \goinl{any}.
It then accesses the field, and performs a type assertion to \goinl{MyInt}.
It returns the \goinl{MyInt}.
The \goinl{main} body expression first calls \goinl{bar} and
passes a \goinl{MyInt}, which is well typed because the \goinl{MyInt} base
type structurally implements all the methods of the \goinl{INum} interface.
It then calls \goinl{String} on the result.

The LWG output of compiling the above code is as follows.
For readability, we simplify slightly and use an abridged notation compared to
the full formal definitions later.
\begin{golst}
func bar(x INum) MyInt { return {x.(change any)}.0.(MyInt) }
func main() { return bar(42.(make MyInt $\rho$))#MyInt.String() }
\end{golst}
LWG models that in the Go runtime, interfaces are implemented as fat
pointers that box a value and carry its RTTI and a method table,
whereas base types and structs are unboxed raw values with no RTTI or
method table.
The compiler generates operations for dynamically
converting between the two kinds, via the \goinl{make} and \goinl{change}
operations.

\begin{itemize}[leftmargin=*, itemindent=1em]
\item
The expression \goinl{42.(make MyInt $\phantom{a}\rho$)} will convert (i.e.,
box) the raw value \goinl{42} %
to an \emph{interface value} of interface type \goinl{INum}.
While the raw value itself incurs no runtime overhead, when boxed as
an interface value it will explicitly carry the RTTI denoted by
\goinl{MyInt} and the method table denoted by $\rho$. Here $\rho$ is
the \emph{method table} of \goinl{MyInt} projected to \goinl{INum},
i.e., the \goinl{String} method.
Both elements are statically computed by the compiler and embedded into the
output code.

\item
The expression \goinl{x.(change any)} will convert the interface value
held by \goinl{x} at runtime to an interface value for the \goinl{any}
interface.
In general such conversions must be done dynamically as the exact
base or struct type (i.e., the RTTI) boxed by the interface value is
statically unknown.

\item
The expression \goinl{$e$\#MyInt.String()} is a method call that has been
statically resolved by the compiler to be a \goinl{String} call on the
type \goinl{MyInt}.
Since the receiver sub-expression $e$, i.e.,
\goinl{bar(42.(make MyInt $\phantom{a}\rho$))} is of type \goinl{MyInt}, it
will evaluate to a raw value at runtime; all such calls on a receiver of
base or struct type must be statically resolved because the raw values at
runtime do not carry any RTTI or method tables for dynamic dispatch.
\end{itemize}

The LWG reduction proceeds as follows.

\noindent
\smallskip
\begin{tabular}{lll}
& \goinlf{bar(42.(make MyInt $\phantom{a}\rho$))\#MyInt.String()} \quad &
Make the interface value \goinlf{(MyInt,$\rho$,42)}
\\
$\rightarrow$ &
\goinlf{bar((MyInt,$\rho$,42))\#MyInt.String()} &
Call function \goinlf{bar} (substitute the method body)
\\
$\rightarrow$ &
\multicolumn{2}{l}{\goinlf{\{(MyInt,$\rho$,42).(change any)\}.0.(MyInt)\#MyInt.String()}
\quad Change the interface value to \goinlf{any}}
\\
$\rightarrow$ &
\goinlf{\{(MyInt,$\epsilon$,42)\}.0.(MyInt)\#MyInt.String()} &
Select the first field of the struct (i.e., \goinl{f} in \goinl{C})
\\
$\rightarrow$ &
\goinlf{(MyInt,$\epsilon$,42).(MyInt)\#MyInt.String()} &
Assert the boxed type is \goinlf{MyInt} (and unbox)
\\
$\rightarrow$ &
\goinlf{42\#MyInt.String()} &
Call the statically resolved method \goinlf{MyInt.String}
\\
$\rightarrow^*$ &
\goinlf{"42"} &
Final result value
\\
\end{tabular}
\smallskip

\noindent
Note the \goinl{change} to \goinl{any} results in an empty method table $\epsilon$ since
\goinl{any} has no methods.
Overall, using runtime conversions to keep method tables aligned with the
expected interface type is key to Go's implementation of structural subtyping.
Since an interface can be openly implemented by any struct with an arbitrary
superset of the required methods, aligning the method table with the interface
allows calls to be dynamically dispatched based on the method offset for that
interface.

The type assertion to \goinl{MyInt} dynamically checks the RTTI and, in this
case, \emph{unboxes} the value.

\subsection{LWG and Generics}

The current implementation of Go uses a combination of static monomorphisation
and runtime management of RTTI for generic type
arguments~\cite{godictionaries}.
Prior work~\cite{DBLP:journals/pacmpl/GriesemerHKLTTW20} has formalised
monomorphisation as a translation from a core subset of generic Go
to non-generic Go.
By contrast, this paper focuses on formalising the lower level mechanisms of
the Go runtime in LWG, and the compilation from WG to LWG.
We model Go's RTTI for generic types by reusing the existing
infrastructure for interface values: in an interface value
\goinl{($R$,$\rho$,$v$)}, we allow the RTTI element $R$ to be a ground
generic type, e.g., $S[\ov{T}]$ -- $\ov{T}$ is a sequence of ground
types that corresponds to the dictionaries used by the Go runtime.
In our system, we could consider monomorphisation as an optimisation
to eliminate some of the runtime boxing and conversion operations: with
appropriate restrictions, we could monomorphise generic WG to
non-generic WG and then compile to LWG.
However, this paper shows that WG can be compiled directly to LWG and,
unlike monomorphisation, requires no restrictions.

We briefly illustrate some of the challenges and subtleties of LWG and our
compilation arising from the interplay between structural subtyping, generics
and Go's runtime infrastructure.

One key aspect is the type system for LWG.
Consider again a runtime interface value \goinlf{(MyInt,$\epsilon$,42)} that boxes
a \goinl{MyInt} under the \goinl{any} interface.
How should we type this interface value?
Should it be typed as an \goinl{any} according to its interface?
Unfortunately, this would break type safety of the compiled code.
Below on the left is a well-typed WG expression (recall the generic bound of
\goinl{Cons} is \goinl{any}).
On the right is the LWG compilation and its reduction step that converts the
constant into the interface value mentioned above: neither expression before
nor after the step would be well-typed because the struct requires its first
field expression to be of type \goinl{MyInt}, not \goinl{any}.

\smallskip
\centerline{%
\goinlf{Cons[MyInt]\{MyInt(42), ...\}}
\quad
\goinlf{Cons[MyInt]\{42.(make MyInt $\phantom{a}\epsilon$), ...\}}
$\rightarrow$
\goinlf{Cons[MyInt]\{(MyInt,$\epsilon$,42), ...\}}
}
\smallskip

\noindent
Should the interface value be typed as a \goinl{MyInt} according to its RTTI?
Unfortunately, this would break progress for LWG: bogus LWG expressions such
as \goinlf{(MyInt,$\epsilon$,42).String()} would be well-typed but stuck, as
the \goinl{String} call cannot be dispatched by an empty method table
$\epsilon$.
Neither option is adequate on its own.
Section~\ref{sec:lwg} presents the LWG type system which
tracks \emph{both}.

As a final example for this overview, consider the WG code in
Figure~\ref{fig:WGvariant} where we add a second method to %
\goinl{List} and define a version of \goinl{Cons} specifically for
\goinl{MyInt}.
\begin{figure}[t]
\begin{golst}
type List[a any] interface { ...; Fold(f func(v a, w a) a, z a) a }
func (x Cons[a]) Fold(f func(v a, w a) a, z a) a {return x.FoldL[a](f, z)}
type MyIntCons struct { head MyInt; tail List[MyInt] }
func (x MyIntCons) FoldL[b any](f func(v b, w MyInt) b, z b) b {return x.tail.FoldL[b](...)}
func (x MyIntCons) Fold(f func(v MyInt, w MyInt) MyInt, z MyInt) MyInt {return x.FoldL[MyInt](...)}
\end{golst}
  \vspace{-0.2cm}
\caption{Extending the example from Fig.~\ref{fig:WGexample}.}
\label{fig:WGvariant}
\end{figure}
Note that both \goinl{Cons[MyInt]} and \goinl{MyIntCons} are structural
subtypes of \goinl{List[MyInt]}.
Now consider this WG function: %

\smallskip
\centerline{\goinlf{func foo(x List[MyInt]) MyInt \{ return x.Fold(..., 0) \} }}
\smallskip

\noindent
How should we compile this function?
How should the argument \goinl{0} be passed in the call to
\goinl{Fold}?
Statically, we do not  know if the runtime value of \goinl{x} will
be (i) a \goinl{Cons[MyInt]}, noting that the compilation of
\goinl{Fold} for \goinl{Cons[a]} expects an interface value for its generic
formal bounded by \goinl{any}, or (ii) a \goinl{MyIntCons}, where \goinl{Fold}
expects a raw \goinl{MyInt}.
Sec.~\ref{sec:compilation} presents our compilation that addresses this
problem by building on Go's approach to statically generating %
runtime type conversions. %
\section{Welterweight Go}
\label{sec:fgg}
\label{sec:wg}

In this section we introduce Welterweight Go (WG), an extension of
\ROUNDTWO{Featherweight Generic Go
(FGG)~\cite{DBLP:journals/pacmpl/GriesemerHKLTTW20}, a core
calculus of Go with generics.
WG extends FGG} with a selection of core and recent type-level
additions to Go, notably anonymous (interface and struct) types,
type unions (an idiomatic form of untagged sum types) and
static type conversions.

Type unions bridge Go's interface bounds with basic types
and operators. Go overloads operators like $<$ on basic types
(integers, floats, strings) but requires arguments of the same
type.
Since operators are not methods, interface bounds alone cannot permit
their use on generic types, creating a gap that type unions
fill. For example, it is not possible to define a generic \goinl{Min}
function by appealing to the overloaded $<$ operator via interface
bounds alone.
To account for this lack of
expressiveness, the Go team introduced type
unions, which enable the kind of genericity mentioned
above. A generic \goinl{Min} function can be written as:

\begin{golst}
type  Ordered interface { int | float64 | ~string }
func Min[T Ordered](x T, y T) T { if x < y { return x } else { return y } }
\end{golst}

The \goinl{Ordered} interface denotes a \emph{type union},
satisfied by any \goinl{int}, \goinl{float64} or any type whose
\emph{underlying type} is \goinl{string} ($\ttilde U$ refers to the
set of all types whose underlying type is $U$).
The use of \goinl{Ordered} as a bound for type parameter \goinl{T}
warrants the use of any built-in operator that is common to all types
in the union, making the definition correct.
\ROUNDTWO{
The mix of $\ttilde$ and non-$\ttilde$ elements offers fine-grained control
over the constraint: \goinl{T} may be instantiated by \goinl{int} specifically
but not other named types with \goinl{int} as the underlying representation,
whereas \goinl{T} may be instantiated by any type with \goinl{string} as its
underlying representation.}

In their full generality, type unions can include lists of both
types and methods, and can be used as named or anonymous type parameter bounds in
the language (although they cannot be used directly as types).
Type unions in the language generalises
interfaces to be understood not only as sets of methods but \emph{also} as sets of
\emph{types}. %
\paragraph{From FGG to WG}
Whereas FGG provides a faithful model of Go's subtyping
relation for \emph{named} interfaces, WG fully models both
implementability of anonymous and named interface types, relying on
method and type sets to faithfully account for type unions, as
well as Go's concept of \emph{assignability}, which generalises
implementability to account for anonymous types.
Moreover, to accurately model the full range of features provided by
type unions, we consider base types and operators in WG.
Finally, in WG we can also naturally account for Go's type
\emph{conversion} operator. This is necessarily absent from FGG, which
includes only type \emph{assertion}.

Assignability in Go relies on the notion of \emph{underlying type}
which in turn is a cornerstone of the language formalism.
In the type union \goinl{Ordered}, the clause that includes string types
is written $\sim$\goinl{string}.  Go allows \ROUNDTWO{a new type with the same
underlying memory representation as another to be defined as a nominally
distinct type and to which, for instance, additional methods can be attached}.
By using the $\sim$ type operator we can include not just the type (e.g.)
\goinl{string} itself, but also all types whose \emph{underlying type} is
\goinl{string} (i.e.,~all
\ROUNDTWO{types whose underlying representation is a \goinl{string}}).

\subsection{WG Syntax}
\label{sub:wg-syntax}

\begin{Figure}
  \begin{ruled}
    \begin{minipage}[t]{\textwidth}
      \begin{tabular}[t]{ll}
        Field name                      & $f$ \\
        Method name                     & $m$ \\
        Variable name                   & $x$ \\
        \black{Type parameter}          & \black{$\alpha$} \\
        Type name                       & $t, u$  \qquad ($t,u \neq B$) \\
        Base types                 & $B$ ::= $\mathtt{int} \mid \mathtt{bool} \mid \ldots$ \\
        Type                            & $R, S, T,U$ ::= $B \mid L \mid t \black{[\ov{S}]} \mid \black{\alpha}$\\
        Type/bound pair                 & \black{$\Phi, \Psi$ ::= $\typair{T}{V} $} \\
        Type literal                    & $L$ ::= \\
        \quad Structure                 & \quad $\struct~\br{\ov{f~{\Phi}}}$ \\
        \quad Interface                 & \quad $\interface~\br{\ov{F}}$ \\
        Method signature                & $M$ ::= $\black{[\ov{\alpha~S}]}(\ov{x~{\Phi}})~{\Psi}$ \\
        Declaration                     & $D$ ::= \\
        \quad Type declaration          & \quad $\type~t\black{[\ov{\alpha~S}]}~T$ \\
        \quad Method declaration        & \quad $\func~(x~t\black{[\ov{\alpha~S}]})~mM~\br{\return~e}$ \\
        Program                         & $P$ ::= $\package~\main;~\ov{D}~\func~\main()~\br{\un=e}$\\
      \end{tabular}
    \end{minipage}
    \hspace{-0.5\textwidth}
    \qquad
    \begin{minipage}[t]{0.4\textwidth}
      \begin{tabular}[t]{ll}
        Type term                    & \black{$C$ ::= $T \mid \ttilde  T $} \\
        Type union                   & \black{$E$ ::= $C \mid C \typeor E $} \\
        Interface element            & \black{$F$ ::= $ mM \mid E $} \\
        Constants                       & $ c$ ::= $0 \mid 1 \mid  \ldots$ \\
        Operators                       & $ \circ $ ::= $ + \mid * \mid \land \mid  \ldots$ \\
        Expression                      & $e$ ::= \\
        \quad Variable                  & \quad $x$ \\
        \quad Method call               & \quad $e.m\black{[\ov{S}]}(\ov{e})$ \\
        \quad Structure literal         & \quad ${T}\br{\ov{e}}$ \\
        \quad Select                    & \quad $e.f$ \\
        \quad Type assertion            & \quad $e.({T})$ \\
        \quad Type conversion           & \quad $T(e)$ \\
        \quad Constant                  & \quad $T(c)$ \\
        \quad Operation                 & \quad $\circ(\ov{e})$ \\
      \end{tabular}
    \end{minipage}
  \end{ruled}
  \caption{WG syntax}
  \label{fig:wg-syntax}
  \vspace{3ex}
\end{Figure}

The syntax of WG is given in Figure~\ref{fig:wg-syntax}.
Field names $f$, method names $m$, variable names $x$, type parameters
$\alpha$, and type names $t$ form the basic identifiers.
Types $T$ include base types $B$ (such as \texttt{int} and
\texttt{bool}), type literals $L$, (parameterised) named types
$t[\ov{S}]$, type parameters $\alpha$.
We write $\ov{S}$ for a (possibly empty) list of $S_i$.
We use the notation $\typair{T}{V}$ to annotate types with their
bounds, where $T$ is the actual type and $V$ is its bound. In WG, when
$T$ is a type parameter $\alpha$, $V$ represents its declared bound;
when $T$ is a concrete instantiation of a type parameter, $V$
represents the bound of that parameter; otherwise $V = T$.
For instance, in a struct type literal $\struct~\br{\ov{f~{\Phi}}}$,
we assume that every field type is annotated with its bound. Such
annotations can be straightforwardly inferred during type checking.
Our formalisation of WG would be identical if we replaced annotations
$\typair{T}{V}$ with simple types $T$. However, the type-directed
compilation to LWG described in Section~\ref{sec:compilation} would
become significantly more complex to present without these
annotations.

An interface type literal $\interface~\br{\ov{F}}$ contains a sequence
of elements, i.e., method signatures ($M$) or \emph{unions} of
type terms ($E$).
Each union specifies that the interface is only satisfied by
its members.
Type terms $C$ can be either types $T$, or their approximation $\ttilde T$
which specifies all types whose underlying type is $T$.
\ROUNDTWO{
Note, a type declaration $\type~U~T$ introduces a new named type $U$ with
underlying type given by $T$:
if $T$ is also named, Go's notion of assignability (Sec.~\ref{sub:wg-subtyping})
distinguishes $T$ and $U$ nominally, but $\ttilde T$ allows to express
compatibility of (nominally-distinct) types with the same underlying type.
}
Method signatures $M$ take the form
$[\ov{\alpha~S}](\ov{x~\Phi})\ \Psi$, where type parameters
$\ov{\alpha}$ are bounded by types $\ov{S}$, parameters $\ov{x}$ have
types $\ov{\Phi}$, and the return type is given by $\Psi$.
Type literals can be used in place of named types, in which case we
call them \emph{anonymous types}.

Expressions include variables, method calls, structure literals, field
selection, type assertions, type conversions, typed constants, and
operations.
We model only explicitly typed constants $T(c)$ such as \goinl{int(42)} or
\goinl{float64(42.0)}, abstracting away Go's untyped constants as
their type inference is orthogonal to our work.
A program $P$ consists of a sequence of declarations $\ov{D}$ and a
top-level expression $e$, written in the stylised form shown in the
figure to make it legal Go.  We often abbreviate it as
$\ov{D} \prog e$.

\subsection{Typing and Subtyping in WG}
\label{sub:wg-subtyping}

\begin{Figure}
  \begin{ruled}
    \begin{mathpar}
      \inferrule
      {}
      {\underlying_\Delta(B) = B}

      \inferrule
      {}
      {\underlying_\Delta(L) = L}

      \inferrule
      {
        (\type~t[\ov{\alpha~S}]~T) \in \ov{D}
        \and
        \eta = (\ov{\alpha \by R})
      }
      {\underlying_\Delta(t[\ov{R}]) = \underlying_\Delta(T[\eta])}

      \inferrule
      {(\alpha : T) \in \Delta}
      {\underlying_\Delta(\alpha) = \underlying_\Delta(T)}

      \ROUNDTWO{
        \inferrule
        {}
        {\underlying_\Delta(\ttilde V) = V}
      }

      \inferrule
      {(\alpha: T) \in \Delta}
      {\utyvar_\Delta(\alpha)}

      \inferrule
      {(\alpha: T) \in \Delta}
      {\unonconstraint_\Delta(\alpha)}

      \inferrule
      {\ustruct(T)}
      {\unonconstraint_\Delta(T)}

      \inferrule
      {\ubase(T)}
      {\unonconstraint_\Delta(T)}

      \inferrule
      {\underlying_\Delta(T) =\interface~\br{\ov{mM}}}
      {\unonconstraint_\Delta(T)}

      \inferrule
      {\,}
      {\maintype{\typair{T}{V}}  =  \maintype{T}}

      \inferrule
      {\,}
      {\uppertype{\typair{T}{V}}  =  V}

      \inferrule
      {(\alpha : T) \in \Delta}
      {\bounds_{\Delta}(\alpha) = T}

      \inferrule
      {\neg \utyvar_\Delta(T)}
      {\bounds_{\Delta}(T) = T}

      \inferrule
      {
        \underlying_\Delta(T) = \struct~\br{\ov{f~\Phi}}
      }
      {\fields_\Delta(T) = \ov{f~\Phi}}

      \inferrule
      {
      }
      {\maintype{(m[\ov{\alpha~S}](\ov{x~\Phi})~\Psi)}  =
        m[\ov{\alpha~S}](\ov{x~\maintype{\Phi}})~\maintype{\Psi}
      }

      \inferrule
      {
        (\func~(x~t[\ov{\alpha~S}])~m[\ov{\beta~T}](\ov{y~\Phi})~\Psi~\br{\return~e}) \in \ov{D} \\
        \theta = (\ov{\alpha \by S'} \cup \ov{ \beta \by T'})
      }
      {\mbody(t[\ov{S'}].m[\ov{T'}]) = (x ,\ov{y:\maintype{\Phi}}[\theta]).e : \maintype{\Psi}[\theta]}
    \end{mathpar}
  \end{ruled}
  \caption{WG: auxiliary definitions}\label{fig:wg-aux-defs}
\end{Figure}

Go's type system is built around interface satisfaction through
structural subtyping (a.k.a.\ ``duck typing'').
A key ingredient of subtyping in Go is the notion of \emph{underlying
  type} (\ROUNDTWO{written $\underlying_\Delta(T)$, where $\Delta$ is
a type environment}) as formalised in
Figure~\ref{fig:wg-aux-defs} along with other auxiliary
definitions.
The underlying type of a base type (resp.\ type literal) is itself.
The underlying type of a type variable is its bound, recorded
in $\Delta$, and the underlying type of a named type
is the underlying type of the RHS of its declaration.
We use the following notation: $\uinterface(T)$ holds when
$T$ is an interface or a type variable, $\ustruct(T)$ holds when $T$
is a struct, $\ubase(T)$ when $T$'s underlying type is a base type
(see Figure~\ref{fig:wg-aux-defs-app}).
We define a few additional predicates:
$\utyvar$ is used to distinguish type variables from other types,
$\unonconstraint$ is used to discriminate interfaces that specify type
unions.
Figure~\ref{fig:wg-aux-defs} also introduces functions to retrieve the
bound of a type, the fields of a struct, and the body of a method.
Function $\maintype{\typair{T}{V}}$ returns the first element of a
pair type, recursively applying the transformation;
$\uppertype{\typair{T}{V}}$ returns the second such element.

The Go specification defines interface satisfaction with two
flavours: implementability and assignability. We formalise both
relations in Figure~\ref{fig:wg-impl-assig}.
\begin{Figure}
  \begin{ruled}
    Types of interfaces \hfill  \fbox{$\types_\Delta(T)$} \fbox{$\types_\Delta(F)$} \fbox{$\types_\Delta(C)$}  \fbox{$\types_\Delta(E)$}
    \begin{mathpar}
      \inferrule
      {\neg \uinterface(T)}
      {\types_\Delta(T) = \{ \maintype{T} \} }

      \inferrule
      {\underlying_\Delta(T) = \interface~\br{\ov{F}}
      }
      {\types_\Delta(T) = \bigcap \{ \types_\Delta(F_i) \mid F_i \in  \ov{F} \} }

      \inferrule
      {\, }
      {\types_\Delta(mM) = { \Universe}}

      \ROUNDTWO{
      \inferrule
      {
        \,
      }
      {\types_\Delta(\ttilde T) = \{ \ttilde \maintype{T} \}
      }
      }

      \inferrule
      {
      }
      {\types_\Delta(E) = \bigcup \{ \types_\Delta(C) \mid C \in  E \}}
    \end{mathpar}

    Methods of interface \hfill \fbox{$\methods_\Delta(T)$} \fbox{$\methods_\Delta(F)$}  \fbox{$\methods_\Delta(C)$}
    \begin{mathpar}
      \inferrule[]
      {
      }
      {\methods_\Delta(B) = \emptyset}

      \inferrule[]
      { \ustruct(L)}
      {\methods_\Delta(L) = \emptyset}

      \inferrule
      {      
        \underlying_\Delta(T) = \interface~\br{\ov{F}}
      }
      {
        \methods_\Delta(T) = \bigcup \{ \methods_\Delta(F_i) \mid \, F_i \in \ov{F}  \}
      }

      \inferrule
      { \neg \uinterface(t[\ov{S'}]) }
      {\methods_\Delta(t[\ov{S'}]) =
        \set{mM[\eta]
          \mid (\func~(x~t[\ov{\alpha~S}])~mM~\br{\return~e}) \ { \in \ov{D}} \comma
          \eta = (\ov{\alpha \by S'})
        }
      }

      \inferrule
      {
      }
      {\methods_\Delta(mM) = \{ mM \} }
      \quad
      \inferrule[]
      {        
      }
      {\methods_\Delta(\ttilde T) = \methods_\Delta( T)  }
      \quad
      \inferrule
      { |E| > 1 }
      {\methods_\Delta(E) =  \bigcap \{ \methods_\Delta(C) \mid C \in  E \} }   
    \end{mathpar}
    
    Implements
    \hfill \fbox{$T \imp_\Delta U$}
    \begin{mathpar}
      \inferrule[<:$_{\alpha}$]
      {\,
      }
      {\alpha \imp_\Delta \alpha}

      \inferrule[<:$_{T}$]
      { \neg\uinterface(T)}
      {T \imp_\Delta T}

      \ROUNDTWO{
      \inferrule[<:$_I$]
      {
        \uinterface(U)
        \and
        \neg\utyvar_\Delta(U)
        \and
        \maintype{\methods_\Delta(T)} \supseteq \maintype{\methods_\Delta(U)}
        \\\\
        \forall C \in \types_\Delta(T) : \, C \in \types_\Delta(U) \lor (\ttilde V \in \types_\Delta(U) \land \underlying_\Delta(C) = V)
      }
      { T \imp_\Delta U }
      }
    \end{mathpar}

    Assignable
    \hfill \fbox{$T \prec:_\Delta U$}
    \begin{mathpar}
      \inferrule[$\assg_{\imp}$]
      {T \imp_\Delta U}
      {T \assg_\Delta U }

      \inferrule[$\assg_{tL}$]
      {\underlying_\Delta(t[\ov{S}]) = L}
      {t[\ov{S}] \assg_\Delta L }

      \inferrule[$\assg_{Lt}$]
      { L = \underlying_\Delta(t[\ov{S}])}
      {L \assg_\Delta t[\ov{S}] }
    \end{mathpar}
  \end{ruled}
  \caption{Implements and Assignability relations.
    $\Universe$
    is defined s.t.\ $T \in \Universe$ for all $T$.
  }\label{fig:wg-impl-assig}
\end{Figure}

Implementability determines whether a type satisfies an interface.
Assignability, on the other hand, governs when a value of one type can
be assigned to a variable of another type, extending implementability.
Before Go 1.18, an interface was solely defined by its methods and so
any type implementing those methods would satisfy the interface,
regardless of its declared type.
To handle operators on generic types, interfaces were extended in Go
1.18, to also include explicit type unions, allowing developers to
restrict which types can satisfy an interface beyond just method
requirements.
Figure~\ref{fig:wg-impl-assig} formalises this through two auxiliary
functions that return type sets ($\types_\Delta(T)$) and method sets
(${\methods_\Delta(T)}$).

The $\types_\Delta(T)$ function returns the type set of a type, and
particularly the set of types explicitly listed in an interface,
possibly flagged with $\ttilde T$ notation.
The type set of a non-interface type, is the singleton set of this
type.
The type set of elements of an interface correspond to the
intersection of their respective type sets.
The type set of a method signature is the universe of all types
($\Universe$), i.e., methods alone do not constrain types.
Let $\Universe$ be the universe of all types, axiomatically
defined such that for any type $T$ we have that $T \in
\Universe$. %
The type set of an approximation is itself ($\ttilde T$).
The type set of a union of types is the union of their type sets.

The ${\methods_\Delta(T)}$ function extracts method signatures
attached to a type $T$.
The method set of base types and anonymous structs is empty.
The method set of an interface is the union of the method set of its
elements.
The method set of a named, non-interface type is the (instantiated)
signatures of its declared methods.
The method set of a method signature is the singleton of itself.
The method set of $\ttilde T$ is the method set of $T$.
The method set of a union of types is the intersection of the method
sets of its elements (note that when $\lvert E \rvert = 1$, case $T$
or case $\ttilde T$ applies).

\begin{Figure}
  \begin{ruled}
    Well-formed types
    \hfill \fbox{$\Delta \vdash T \nice\vphantom{f}$}  \fbox{$\Delta \vdash \Phi \supernice$}  \fbox{$\Delta \vdash T \ok\vphantom{f}$}
    \begin{mathpar}
      \inferrule[t-nc]
      {
        \Delta \vdash T \ok
        \\
        \unonconstraint_\Delta(T)
      }
      { \Delta \vdash T \nice }

      \inferrule[t@v-bc]
      {
        \Delta \vdash T \nice
        \\
        V = {\bounds_{\Delta}(T)}
      }
      { \Delta \vdash \typair{T}{V} \supernice }

      \inferrule[t-named]
      {
        \Delta \vdash \ov{R'} \ok
        \\
        (\type~t[\ov{\alpha~R}]~T) \in \ov{D}
        \and
        \eta = (\ov{\alpha \by R'})
        \\
        \ov{\alpha[\eta] \imp_\Delta R[\eta]}
      }      
      { \Delta \vdash t[\ov{R'}] \ok }  
    \end{mathpar}
    Well-formed declarations \hfill \fbox{$D \ok$}
    \begin{mathpar}
      \inferrule[t-type]
      {
        \Delta = \ov{\alpha: R}
        \\
        \neg\utyvar_\Delta(T)
        \\
        \emptyset \vdash    \ov{\alpha~R} \ok \\
        \Delta \vdash T \ok
      }
      { \type~t[\ov{\alpha~R}]~T \ok }

      \inferrule[t-func]
      {
        \neg\uinterface(t[\ov{\alpha}]) \\
        \distinct(x, \ov{y}) \\
        (\type~t[\ov{\alpha~R}]~\_) \in \ov{D} \\
        \emptyset \vdash  \ov{\alpha~R}  \ok \\
        \ov{\alpha: R} \vdash \ov{\beta~S} \ok\\
        \Delta =  \ov{\alpha: R} , \ov{\beta: S}
        \\
        \Delta \vdash \ov{\Phi} \supernice \\
        \\
        \Delta \vdash \Psi \supernice \\
        \\
        \Delta \stoup
        x : t(\ov{\alpha}) \comma \ov{y : \maintype{\Phi}} \vdash e : U \\
        \\
        U \assg_\Delta \maintype{\Psi}
      }
      { \func~(x~t[\ov{\alpha~R}])~m[\ov{\beta~S}](\ov{y~\Phi})~\Psi~\br{\return~e} \ok }
    \end{mathpar}
  \end{ruled}
  \caption{Well-formed types and declarations}\label{fig:wg-wf-types-decl}
\end{Figure}

We can now define subtyping precisely.
Figure~\ref{fig:wg-impl-assig} (bottom) defines interface
implementation ($\imp_\Delta$) and assignability ($\assg_\Delta$)
relations.
A type $T$ implements an interface $U$ if
(1) $T$'s method set is a superset of $U$'s method set and
(2) $T$'s type set is a subset of $U$'s type set.
\ROUNDTWO{
The function $\maintype{\_}$ essentially erases (recursively) our type
annotations in method signatures and within types, e.g., $\maintype{\typair{T}{V}}=
\maintype{T}$.
Aspect (2) involves checking that each element in the type set of $T$ is
directly in the type set of $U$ or is otherwise covered by an appropriate
$\ttilde{}V$.
}

\ROUNDTWO{We note that the official Go Language
Specification~\cite{The-Go-Team-2025} defines interface satisfaction purely in
terms of type sets.
We give an equivalent definition that explicitly tracks both type sets and method sets as
  separate components of our subtyping judgements. This design choice
  faithfully models the actual implementation
  of the Go compiler, which maintains analogous internal representations
  tracking both aspects independently.
}

Assignability extends implementation with rules for
dealing with named types and their underlying types, capturing Go's
implicit type conversion rules.
Note that $\assg$ is reflexive via $\imp$, but it is not transitive, as
Rules $\assg_{tL}$ and $\assg_{Lt}$ require one of the types to be
anonymous.

The WG type system relies on type and method declarations to be well-formed.
We show key rules in Figure~\ref{fig:wg-wf-types-decl}.
The judgement $\mathcal{J} \ok$ ensures types are well-formed,
possibly under variable and type parameter contexts.
We specify two variants of this judgement:
$\Delta \vdash T \nice$ holds when $T$ is a basic interface, i.e., not
a union as defined in Figure~\ref{fig:wg-aux-defs}.
In WG and in Go, variables and parameters must be assigned a type $T$
such that $T \nice$ holds.
Judgement $\Delta \vdash \Phi \supernice$ simply enforces that our
meta-theoretic annotations are well-formed.

Type declarations in WG require their bound types and body to be
well-formed.
Method declarations enforce additional constraints: the receiver must
be a non-interface type, parameters must be distinct, and the method
body must be well-typed under the appropriate context, see
Figure~\ref{fig:wg-wf-types-decl}.
We relegate some wellformedness rules for interface and struct
definition to \appendixref{fig:wg-iface-types-wf}{Figure} as they are similar
to those of FGG.
A wellformedness restriction is that interfaces are restricted so that
unions cannot include type variables nor interfaces with non-empty
method sets.

We show key typing rules in
Figure~\ref{fig:fgg-typing}. For simplicity, we enforce that operators
take as argument and return type values of the same
type. Generalising the typing of operators is straightforward but
orthogonal to our design.
The expression typing judgement $\Delta; \Gamma \vdash e : T$ (where
$\Delta$ maps type parameters to their bounds and $\Gamma$ maps
variables to their types) follows typical patterns. For instance, Rule
\textsc{t-call} handles generic type parameter instantiation through
substitution $\eta$, it also ensures that the instantiated type
parameters satisfy their bounds.
Observe how implementability ($\imp_\Delta$) is used to check that type
instances implements their bounds, while assignability
($\assg_\Delta$) is used to check that struct fields and method
arguments are compatible with their declared types.
Rule \textsc{t-op} deals with operators when the underlying type of the
operands is a base type. We write $\opdom(\circ)$ for the set of
types on which $\circ$ is defined.
Rule \textsc{t-op-$\alpha$} allows the operands to be typed with
a type variable $\alpha$, in which case all types in the bound of
$\alpha$ must support $\circ$.

For the sake of space, typing rules for type assertions (as in FGG),
constants (standard), and conversions are in
\appendixref{fig:wg-typing-expr-app}{Figure}.
Essentially a conversion $T(e)$ is well-typed if $e$ has type $U$ and
either $U \assg_\Delta T$ or
$\underlying_\Delta(U) = \underlying_\Delta(T)$, for a suitable
$\Delta$.

\begin{Figure}
  \begin{ruled}
    Expressions \hfill \fbox{$\Delta \stoup \Gamma \vdash e : T$}
    \begin{mathpar}   
      \inferrule[t-call]
      {
        \Delta \vdash \ov{S'} \nice
        \and
        (m[\ov{\alpha~S}](\ov{x~\Phi})~\Psi) \in \methods_\Delta(R)  \\
        \Delta \stoup \Gamma \vdash e : R \\
        \Delta \stoup \Gamma \vdash \ov{e : T} \\
        \eta = (\ov{\alpha \by S'})
        \\
        \ov{\alpha[\eta] \imp_\Delta S[\eta]}
        \\
        \ov{T[\eta] \assg_\Delta \maintype{\Phi}[\eta]}
      }
      { \Delta \stoup \Gamma \vdash e.m[\ov{S'}](\ov{e}) : \maintype{\Psi}[\eta] }

      \inferrule[t-literal]
      {
        \ustruct(T) \\
        \Delta \vdash T \nice
        \and
        \Delta \stoup \Gamma \vdash \ov{e : S}
        \and
        (\ov{f~\Phi}) = \fields_\Delta(T)
        \and
        \ov{S \assg_\Delta \maintype{\Phi}}
      }
      { \Delta \stoup \Gamma \vdash T\br{\ov{e}} : T }

      \inferrule[t-var]
      {
        (x : T) \in \Gamma
      }
      { \Delta \stoup \Gamma \vdash x : T }

      \inferrule[t-field]
      {
        \Delta \stoup \Gamma \vdash e : T
        \\
        \ustruct(T)
        \\
        (\ov{f~\Phi}) = \fields_\Delta(T)
      }
      { \Delta \stoup \Gamma \vdash e.f_i : \maintype{\Phi_i} }
      
      \inferrule[t-op]
      {
        \Delta  \stoup \Gamma\vdash    \ov{ e} : T
        \and
        \underlying_\Delta(T) = B
        \and
        B \in \opdom(\circ)
      }
      {\Delta \stoup  \Gamma \vdash \circ(\ov{e}) : T}

      \inferrule[t-op-$\alpha$]
      {
        \Delta  \stoup \Gamma \vdash \ov{e} : \alpha
        \and
        \{ \underlying_\Delta(V) \mid V \in \types(\Delta(\alpha)) \}
        \subseteq
        \opdom(\circ)
      }
      {\Delta \stoup  \Gamma \vdash \circ(\ov{e}) : \alpha}
    \end{mathpar}
    Programs  \hfill \fbox{$P \ok$}
    \begin{mathpar}
      \inferrule[t-prog]
      {
        \distinct(\tdecls(\ov{D})) \\
        \distinct(\mdecls(\ov{D})) \\
        \ov{D} \ok \\
        \emptyset \stoup \emptyset \vdash e : T
      }
      { \package~\main;~\ov{D}~\func~\main()~\br{\un=e} \ok }
    \end{mathpar}
  \end{ruled}
  \caption{WG typing}
  \label{fig:fgg-typing}
\end{Figure}

\begin{example}[Assignability]\label{ex:assign}
  Consider the following
  well-typed WG (and Go) code:
\begin{golst}
type Point struct { x int; y int }              func (p Point) move() {}
type Coord struct { x int; y int }              type Converter struct { }
func (c Converter) toPoint(coord struct{ x int, y int}) Point { return coord }
func (c Converter) toAnonymous(p Point) struct{ x int; y int} { return p }
\end{golst}
\end{example}
The code above defines two identical named struct types \goinl{Point} and
\goinl{Coord}, with \goinl{Point} having a method \goinl{move}, and a struct \goinl{Converter}
with methods \goinl{toPoint} and \goinl{toAnonymous}, both with the same body.
Method \goinl{toPoint} essentially converts its (anonymously typed)
argument into a \goinl{Point} struct, while \goinl{toAnonymous} does
the reverse.
Intuitively,
the body of \goinl{toPoint} is well-typed because an anonymous struct of
the specified shape matches exactly with the shape of type \goinl{Point}
and so the \goinl{return} is justified (\goinl{struct\{x
  int; y int\}} $\assg$ \goinl{Point}).
Dually, we can always treat a
named type as its (anonymous) underlying type and so the body of
\goinl{toAnonymous} is also well-typed (\goinl{Point} $\assg$ \goinl{struct\{x
  int; y int\}}).
Notably, the calls
\goinl{Converter\{\}.toPoint(struct\{x int; y int\}\{1,1\})},
\goinl{Converter\{\}.toPoint(Point\{1,1\})} and
\goinl{Converter\{\}.toPoint(Coord\{1,1\})}
are all
well-typed in WG (and Go).
Similarly,
\goinl{Converter\{\}.toAnonymous(Point\{1,1\})} is
well-typed.
However, the invocation
\goinl{Converter\{\}.toAnonymous(Coord\{1,1\})} is ill-typed since
\goinl{Coord} $\assg$ \goinl{Point} does not hold.

A program is well-typed if its declarations are well-formed and
pairwise
distinct. %

\begin{example}[Implementability and Unions]
  Consider the following well-typed WG code:

\begin{golst}
type Addable interface { int | float64 | ~string }
type Printable interface { ~int | ~string ; customPrint() string }
type MyString string
func (x MyString) customPrint() string { ... }
type MyInt int
func (x MyInt) customPrint() string { ... }
type C struct { }
func (t C) Combine[T Addable](x T, y T) T { return x+y }
func (t C) Print[T Printable](x T) string { return x.customPrint() }
\end{golst}
\end{example}
The \goinl{Addable} type union includes types \goinl{int}, \goinl{float} and all types whose
underlying type is \goinl{string}.  \goinl{MyString} and
\goinl{MyInt} have both a custom printing method.
It is \emph{not} the
case that both types satisfy the constraint specified by
\goinl{Addable}: \goinl{C\{\}.Combine[MyString](MyString("a"),MyString("b")))} is well-typed, but
\goinl{C\{\}.Combine[MyInt](MyInt(2),MyInt(3))} is not. Since
\goinl{Addable} lists \goinl{int} explicitly but refers to $\ttilde$\goinl{string}, we have
$\types($\goinl{Addable}$) = \{$\goinl{int}, \goinl{float},
\goinl{string}, \goinl{MyString}$\}$ and so \goinl{MyString} $\imp$
\goinl{Addable} but \goinl{MyInt} $\not\imp$ \goinl{Addable}.

The example also defines interface \goinl{Printable},
consisting of all types whose underlying types are either \goinl{int}
or \goinl{string} that also have a \goinl{customPrint} method
defined on them.
Thus, method \goinl{Print} and
 the calls \goinl{C\{\}.Print[MyInt](MyInt(2))} and
\goinl{C\{\}.Print[MyString](MyString("a"))} are all well-typed, since
$\types($\goinl{Printable}$) = \{$\goinl{int}, \goinl{string},
\goinl{MyInt}, \goinl{MyString}$\} $ and $\methods($\goinl{Printable}$)
= \{$~\goinl{customPrint() string} $\}$ and thus \goinl{MyString} $\imp$
\goinl{Printable} and \goinl{MyInt} $\imp$
\goinl{Printable}.

\subsection{Operational Semantics of WG}
\label{sub:wg-reduction}

\begin{Figure}
  \begin{ruled}
    \begin{center}
      \begin{minipage}{0.5\textwidth}
        \begin{tabular}{ll}
          Values & $v$ ::= $T\br{\ov{v}} \; \mid \; T(  c)$
          \\
          Evaluation context	        & $E$ ::= \\
          \quad Hole                    & \quad $\Hole$ \\
          \quad Operator                & \quad $\circ(\ov{v}, E, \ov{e})$ \\
          \quad Method call receiver    & \quad $E.m\black{[\ov{T}]}(\ov{e})$ \\
          \quad Method call arguments	& \quad $v'.m\black{[\ov{T}]}(\ov{v},E,\ov{e})$
        \end{tabular}
      \end{minipage}
      \begin{minipage}{0.4\textwidth}
        \begin{tabular}{ll}
          ~\\
          \quad Structure	        & \quad $T\br{\ov{v},E,\ov{e}}$ \\
          \quad Select		        & \quad $E.f$ \\
          \quad Type assertion	        & \quad $E.(T)$ \\
          \quad Type conversion	        & \quad $T(E)$
        \end{tabular}
      \end{minipage}
    \end{center}

    Reduction \hfill \fbox{$d \becomes e$}
    \begin{mathpar}
      \inferrule[r-field]
      { (\ov{f~\Phi}) = \fields_\black{\emptyset}(T) }
      { T\br{\ov{v}}.f_i \becomes \maintype{\Phi_i}(v_i) }

      \inferrule[r-call]
      { (x  ,\ov{y: T}).e : U = \mbody(\vtype(v').m\black{[\ov{S}]}) }
      { v'.m\black{[\ov{S}]}(\ov{v}) \becomes U( e[x \by v', \ov{y \by T(v)}] ) }

      \inferrule[r-assert]
      { \vtype(v) \imp_{\black{\emptyset}} T }
      { v.(T) \becomes v }

      \inferrule[r-op]
      { \delta(\circ, \ov{v}) = v' }
      { \circ(\ov{v}) \becomes v'}
      \quad
      \inferrule[r-convert-b]
      {\ubase(T)}
      { T( U ( c)) \becomes  T(  c ) }
      \quad
      \inferrule[r-convert-s]
      {  \ustruct(T) }
      { T(U\br{\ov{v}}) \becomes T\br{ \ov v }}
      \quad
      \inferrule[r-convert-i]
      {  \uinterface(T) }
      { T(v) \becomes v}
      \quad
      \inferrule[r-context]
      { d \becomes e }
      { E[d] \becomes E[e] }
    \end{mathpar}
  \end{ruled}
  \caption{WG reduction}
  \label{fig:wg-reduction}
\end{Figure}

The operational semantics of WG are presented in
Figure~\ref{fig:wg-reduction}, defining a reduction
relation $d \longrightarrow e$ defined over a
standard evaluation context.
Values in WG consist of structure literals $T\{v\}$ and typed
constants $T(c)$. Given such a value we define
$\vtype(T{\br{\ov{v}}}) = \vtype(T(c)) = T$.

Most rules are standard: Rule {\sc r-op} evaluates
primitive operations via a semantic function $\delta$;
the type conversion rules {\sc r-convert-b}, {\sc -s} and {\sc -i}
specify the behaviour of type conversions. A conversion to a
base type or a struct type collapses all other (nested) type
conversions. A  conversion to an interface type is silently
erased, due to the nature of values in WG.
The type environment is empty in $\fields_\emptyset$ in
\textsc{r-field} (resp.\ $\imp_\emptyset$ in \textsc{r-assert}) since
reductions consider closed, ground terms.

While the essence of the operational semantics is standard,
we note that field selection and method calls are augmented  with type
conversions. The field selection rule {\sc r-field} introduces a type
conversion to the type of the projected field. Similarly, the method invocation
rule {\sc r-call} adds type conversions of the method arguments and
its eventual return value.

These type conversions are key in bridging the gap between the type
system's (structural) subtyping relations and the runtime representation of
values.
Recall Example~\ref{ex:assign}.  If we consider the dynamics of calls to
methods \goinl{toPoint} and \goinl{toAnonymous} from the example, the
need for type conversions becomes clear: the call
\goinl{Converter\{\}.toPoint(struct\{x int; y int\}\{1,1\})} requires
a conversion to type \goinl{Point} since we must be able to treat the
return value of \goinl{toPoint} as if it were a \goinl{Point}
struct. For instance, we may use it as a receiver of a call to
\goinl{move}, which is only defined on \goinl{Point}.
Similarly, the call \goinl{Converter\{\}.toAnonymous(struct\{x int; y
  int\}\{1,1\})} also requires a conversion of its argument to
\goinl{Point}---the body of the method can in principle call upon any
method of \goinl{Point}.
Since method calls look up the appropriate method body by inspecting
the runtime type of the receiver, these type conversions ensure that
the runtime objects line up accordingly.
The call \goinl{Converter\{\}.toPoint(struct\{x int; y int\}\{1,1\})}
thus evaluates to \goinl{Point\{1,1\}} as needed.

Note that although it is not the case that types \goinl{Coord} and
\goinl{Point} are assignable, we may use type conversions between the two
types. In a call such as
\goinl{Converter\{\}.toPoint(Coord\{1,1\})}, our semantics will first type convert \goinl{Coord} to
its anonymous underlying struct type (the type conversion of the
method argument) and from that type to \goinl{Point} (the conversion of
the return value). The method call results in the expression
\goinl{Point(struct\{x int; y int\}(Coord\{1,1\}))}, which reduces in one step
to \goinl{Point(struct\{x int; y int\}\{1,1\})} and then to the value \goinl{Point\{1,1\}}.

While the description above focuses on method calls, struct field
accesses must behave similarly, and so those too require explicit type
conversions in our semantics.

\subsection{Metatheory of WG}

In this section we report the type safety of WG, following mostly
standard arguments of type preservation and progress. As usual, type
preservation relies on type and expression-level substitution lemmas
(and several standard related lemmas that relate substitutions,
subtyping and type well-formedness which we omit from the main text
for the sake of conciseness).
\newcommand{\case}[1]{\textbf{#1}}
\newcommand{\subcase}[1]{\textbf{#1}}

Notably, due to the type conversions that are generated by our
operational semantics, the substitution property need not appeal to subtyping.

\begin{lemma}[Term Substitution]
  If $\Delta ; \Gamma , \ov{x:T} \vdash e : U$ and $\Delta ; \Gamma \vdash
  \ov{e' : T}$ then $\Delta ; \Gamma \vdash e[\ov{x := e'}] : U$.
\end{lemma}

As usual in languages with subtyping, type preservation is defined by
appealing to the implements relation, essentially due to the dynamics
of rules {\sc r-convert-i} and {\sc r-assert}. However, there is a
subtlety due to the fact that most typing rules rely on \emph{assignability}.
As this relation is not transitive, the contextual cases of the type
preservation proof are less standard. However, we can derive a
quasi-transitivity property that relates implementable and assignable
types accordingly, which allows us to conclude type preservation for
the contextual cases as needed.

\begin{lemma}[Quasi-transitivity]
  Let $\Delta \vdash T, T',S\ok$. If $T \imp_\Delta T'$ and $T' \assg_\Delta S$ then $T
  \assg_\Delta S$.
\end{lemma}

\ROUNDTWO{
\begin{theorem}[Type Preservation]\label{thm:wg-preservation}
  If ${\Delta;\Gamma} \vdash e : T$ and $e \longrightarrow e'$
  then ${\Delta;\Gamma}\vdash e' : U$ for some $U$ such that
  $ \vdash U \imp_\ROUNDTWO{\Delta} T$.
\end{theorem}
}

We say expression $e$ panics if there exists an evaluation context
$E$, value $v$ and type $T$ such that $e = E[v.(T)]$ and $\vtype(v)
\notimp T$.

\begin{theorem}[Progress]\label{thm:wg-progress}
  If $\emptyset ; \emptyset \vdash e : T$ then either $e$ is a value,
  $e \longrightarrow e'$ for some $e'$, or $e$ panics.

\end{theorem}

\section{Low-Level Welterweight Go}
\label{sec:lwg}

\newcommand{\lwgExInterface}{\texttt{Shape}}
\newcommand{\lwgExBound}{\texttt{Drawable}}
\newcommand{\lwgExStruct}{\texttt{Circle}}
\newcommand{\lwgMethName}{\texttt{Draw}}
\newcommand{\lwgMethNameBar}{\texttt{Render}}
\newcommand{\lwgMethNameBaz}{\texttt{asStruct}}

\begin{Figure}
  \begin{ruled}
    \begin{minipage}[t]{\textwidth}
      \begin{tabular}[t]{ll@{\hspace{3cm}}ll}
        Field index                     & $i$ \\
        Method table                    & \multicolumn{3}{l}{$\rho$ ::= $\ov{m \mapsto T.m} \cup \ov{\circ \mapsto \circ_T}$}
    \\
        Expression                      & $e$ ::= \\
        \quad Variable                  & \quad $x$
                                        & %
        \\
        \quad Dynamic Method call       & \quad $e.m\black{[\ov{S}]}(\ov{e})$
                                        &      \quad Type assertion            & \quad $e.(\ghostytypair{T}{V})$
        \\
        \quad Static Method call        & \quad $e \smsep  T .  m\black{[\ov{S}]}(\ov{e})$
                                        &      \quad Make Type                 & \quad $\MAKETYPE{e}{\Phi}{S}{\rho}$
        \\
        \quad Structure literal         & \quad $\ghost{T} \br{\ov{e}}$
                                        &         \quad Change Type               & \quad $\CHANGETYPE{e}{\ghostytypair{T}{V}}{S}$
        \\
        \quad Interface value      & \quad $\ghost{\Phi}(S,\rho, e)$
                                        & \quad Static Change             & \quad $\STATICCHANGE{e}{T}{S}$
        \\
        \quad Select                    & \quad $e.i$
                                        &         \quad Constant                  & \quad $\ghost{T}(c)$
        \\
        \quad Dynamic operation         & \quad $\circ(\ov{e})$
                                        &        \quad Static operation          & \quad $\circ\smsep T(\ov{e}) $ \\
      \end{tabular}
    \end{minipage}
  \end{ruled}
  \caption{LWG syntax}
  \label{fig:lwgg-syntax}
\end{Figure}

We propose a \emph{low-level} model of Go that, compared to WG, more
closely reflects the runtime mechanisms of actual Go.
It introduces an additional kind of value, the \emph{interface value}
which we illustrate with the following WG example:
\begin{lstlisting}[language=go,basicstyle=\footnotesize\ttfamily]
type Drawable interface { ... }               type Shape interface { ... } // Assume $\texttt{Shape} \imp_\emptyset \texttt{Drawable}$
type C struct {}                              func (t C) f[T Drawable](x T) T { return x }
type $\lwgExStruct$ struct {p Point; r int} // Assume  $\texttt{Circle} \assg_\emptyset \texttt{Shape}$
func main() { _ = C{}.f[Shape](Circle{}) }
\end{lstlisting}
where we assume that $\texttt{Circle} \assg_\emptyset \texttt{Shape} \imp_\emptyset \texttt{Drawable}$.
In LWG when the value $\ghost{\text{\lwgExStruct}}\br{\ldots}$ is
passed to function \goinl{f} at runtime, it is boxed in an interface
value such as:

\smallskip\centerline{\small%
$\ghost{\typair{\lwgExInterface}{\lwgExBound}}(\lwgExStruct, \rho, \ghost{\text{\lwgExStruct}}\br{\texttt{\ghost{Point}}\br{1,1}, \texttt{\ghost{int}}(2)})$
}\smallskip

\noindent
In this interface value, the ghost $\ghost{\text{\lwgExInterface}}$
specifies the interface logical type (which may be an instantiated type
parameter), $\ghost{\lwgExBound}$ (another interface) represents the
type bound that $\lwgExInterface$ must satisfy, and
${\text{\lwgExStruct}}$ is the concrete implementation type
(i.e., a struct).
The \ghost{pale red} shade indicates ghost types, i.e., information
that has no operational runtime significance. It is used in the formal
type system and compilation only, not in the operational semantics.
In the interface value above, the ghost
$\ghost{\typair{\lwgExInterface}{\lwgExBound}}$ specifies the type of
the interface, i.e., it is a value of type $\lwgExInterface$ which has
only access to the methods of $\lwgExBound$.

We also introduce a notion of \emph{boxy-ness}.
We say a struct value by itself is an \emph{unboxed} value, while
an interface value is a box that contains a
struct value (above, the nested $\ghost{\text{\lwgExStruct}}\br{\ldots}$).
The other elements of an interface value represent RTTI (the
non-ghost ${\text{\lwgExStruct}}$) and the \emph{method table} $\rho$
(mapping abstract methods of $\lwgExBound$ to their
implementations in $\lwgExStruct$).
For any interface value $ \ghost{\typair{T}{V}}(S, \rho, v)$ we always
require that $S \assg_\emptyset T \imp_\emptyset V$.

In accordance with actual Go, our low-level model
distinguishes two kinds of method calls.

\smallskip
\centerline{\small
  \hfill
  $\ghost{\text{\lwgExStruct}}\br{\ldots}\#{\text{\lwgExStruct}}.\lwgMethName()
  \qquad\qquad \ghost{\typair{\lwgExInterface}{\lwgExBound}}
  ({\text{\lwgExStruct}}, \rho,
  \ghost{\text{\lwgExStruct}}\br{\ldots}).\lwgMethName()$
  \hfill
}
\smallskip

\noindent
Above on the left is a \emph{static} method call.
The Go compiler generates these when the type of the receiver expression is a
\emph{struct} (or named base) type, which above is ${\text{\lwgExStruct}}$.
The ${\text{\lwgExStruct}}.\lwgMethName$ part of the static call can
be interpreted as a \emph{statically}-compiled jump address for the target method.
Above on the right is a \emph{dynamic} method call.
Operationally, the receiver must be an \emph{interface} value.
The low-level operational semantics dynamically looks up the target method
$\lwgMethName$ in the table $\rho$.

A key point is that structural subtyping in WG's type system permits,
e.g., a \emph{struct} value to be passed as a method call argument for
a compatible \emph{interface}-typed parameter---yet the two kinds are
\emph{not} operationally interchangeable at the lower-level.
Besides the distinct low-level mechanisms for method calls, the
two kinds have different representations in memory: the
layout of a struct value depends on its field types, while the layout
of an interface value is a fixed ``fat pointer'' structure.

Mediation between the two kinds of values is resolved by the Go compiler.
Based on static typing information, it inserts operations for ``boxing''
of struct values and ``box conversion'' between interface types.
We add these operations as expressions to our low-level language,
matching what can be concretely witnessed in the actual Go compiler.

The syntax of LWG (Figure~\ref{fig:lwgg-syntax}) reflects the
characteristics discussed above.
For simplicity LWG re-uses the syntax of declarations and types from
WG (Figure~\ref{fig:wg-syntax}) but the syntax of expressions is
revised to include: interface values, two types of method calls, and
operations for boxing and conversions. Additionally, field names in
field accesses are substituted by the index of the field, type
assertion uses $\typair{T}{V}$ targets, and type conversions are
no longer present.
This changes reflects the implementation of actual Go and the
assumptions that only interfaces carry RTTI at runtime.

The semantics of LWG is given in Figure~\ref{fig:lwg-reduction}.
We discuss LWG's semantics before its type system because they
directly correspond to Go's runtime behaviour, whereas LWG's type
system (presented in Section~\ref{sub:lwg-typing}) exists solely for
our meta-theory. %
We start by giving a high-level overview of the operations for value
conversions. Consider the following Go
declarations:

\begin{lstlisting}[language=go,basicstyle=\footnotesize\ttfamily]
type $\lwgExInterface$ interface { $\lwgMethName$() $\lwgExInterface$ } $\qquad\qquad\qquad$ func (x $\lwgExStruct$) $\lwgMethName$() $\lwgExInterface$ { return x }
\end{lstlisting}

\begin{Figure}
  \begin{ruled}
    \begin{center}
      \begin{minipage}{0.55\textwidth}
        \begin{tabular}{ll}
          Values & $w$ ::=  $\ghost{T}\{ \ov{v} \}
                   \;|\;
                   \ghost{T}(c)$
                   \qquad \,
                   $ v ::= $  $w
                   \;|\;
                   \ghost{\Phi}(S, \rho, w)$
          \\
          Evaluation context	        & $E$ ::= \\
          \quad Hole                    & \quad $\Hole$ \\
          \quad Dyn Operator                & \quad $\circ(\ov{v}, E, \ov{e})$ \\
          \quad Operator                & \quad  $\circ\smsep T(\ov{v}, E, \ov{e}) $ \\
          \quad Dyn Method call receiver    & \quad $E.m\black{[\ov{T}]}(\ov{e})$ \\
          \quad Dyn Method call arguments	& \quad $v'.m\black{[\ov{T}]}(\ov{v},E,\ov{e})$ \\
          \quad Method call receiver    & \quad $E \smsep  T.m\black{[\ov{T}]}(\ov{e})$ \\
          \quad Method call arguments	& \quad $v' \smsep  T.m\black{[\ov{T}]}(\ov{v},E,\ov{e})$ \\
        \end{tabular}
      \end{minipage}
      \begin{minipage}{0.4\textwidth}
        \begin{tabular}{ll}
          ~\\
          \quad Structure	        & \quad $T\br{\ov{v},E,\ov{e}}$ \\
        \quad Select		        & \quad $E.i$ \\
        \quad Type assertion	        & \quad $E.(\typair{T}{V})$ \\
        \quad Make type  	        & \quad $\MAKETYPE{E}{\Phi}{S}{\rho}$ \\
        \quad Change type  	        & \quad $\CHANGETYPE{E}{\ghostytypair{T}{V}}{S}$ \\
        \quad Static change  	        & \quad $\STATICCHANGE{E}{T}{S}$ \\
      \end{tabular}
    \end{minipage}
  \end{center}
Auxiliary functions \hfill \; %
\begin{mathpar}
      \inferrule
    {
      (\func~(x~t[\ov{\alpha~S}])~m[\ov{\beta~T}](\ov{y~\Phi})~\Psi~\br{\return~e}) \in \ov{D}
    }
    {\mjump(t.m) =      [\ov\alpha \stoup \ov\beta](x,  \ov{y}).e }
    
    \inferrule
    {
      L \imp_\Delta T 
    }
    {
      \MKTABLE{L}{T}{\Delta} = \emptyset
    }
    
    \inferrule
    {
      t[\ov{R}] \imp_\Delta T \and
      \mathcal{O} =  \{   \circ \mid
      \{ \underlying_\Delta(V)  \mid V \in \types_\Delta(T) \}
      \subseteq 
      \opdom(\circ)  
      \} }
    {
      \MKTABLE{t[\ov{R}]}{T}{\Delta} =
      [ m \mapsto t.m \mid mM \in \methods_\Delta(T) ]
      \cup
      [ \circ \mapsto \circ_B \mid  \circ \in \mathcal{O} ,  B = {\underlying(t[\ov{R}])}  ]
    }
  \end{mathpar}

  Reductions \hfill \; \fbox{$e \becomes e'$}
  \begin{mathpar}
    \inferrule[select]{
     }{
     \ghost{S}\{\ov{v} \}.i
     \rightarrow
     v_i
     }

     \inferrule[make]
     { %
     }
     {\MAKETYPE{w}{T@U}{S}{\rho}
     \rightarrow
     \ghost{T@U}(S,\rho, w )
     }

     \inferrule[change-type]
     {
     \rho' = \MKTABLE{S}{U}{\emptyset}}
     {
     \CHANGETYPE{\ghost\Phi(S,\rho , w)}{\ghostytypair{T'}{U}}{T}
     \rightarrow
     \ghost{{\typair{T'}{U}}}( S, \rho' , w)
     }

     \inferrule[static-change-s]
     {
     }
     {
     \STATICCHANGE{ \ghost{S}\{ \ov{v} \}}{T}{S}
     \rightarrow
     \ghost{T}\{ \ov{v} \}
     }

     \inferrule[static-change-c]
     {
     }
     {
     \STATICCHANGE{\ghost{S}(c)}{T}{S}
     \rightarrow
     \ghost{T}(c)
     }

     \inferrule[assert-ok${}_S$]{
     }{
     \ghost\Phi(S, \rho, w)
     .(\typair{S}{S})
     \rightarrow
     w
     }

     \inferrule[assert-ok${}_I$]{
       \uinterface(V)
       \quad
       S \imp_\emptyset U
       \quad
       \rho = \MKTABLE{S}{V}{\emptyset}
     }{
       \ghost\Phi(S, \rho', w)
       .(\typair{U}{V})
       \rightarrow
       \ghost{{\typair{U}{V}}}(S, \rho, w)
     }

     \inferrule[call${}$-dyn]{
     U = t[\ov{R}]
     \quad
     [\ov\alpha \stoup \ov\beta](x,  \ov{y}).e = \mjump(\rho(m))
     }{
     \ghost{\Phi}(U, \rho, v). m[\ov{S}](\ov{v})
     \rightarrow
     e[\ov{\alpha := R}][\ov{\beta := S}][x := v][\ov{y := v}]
     }

     \inferrule[call${}$-static]{
     [\ov\alpha \stoup \ov\beta](x,  \ov{y}).e = \mjump(t.m)
     }{
     w \smsep t[\ov{R}] . m[\ov{S}](\ov{v})
     \rightarrow
     e[\ov{\alpha := R}][\ov{\beta := S}][x :=w][\ov{y := v}]
     }

     \inferrule[op-dyn]
     {
     \forall i . \,
     v_i = \ghost{\Phi}(S, \rho, v_i')
     \\
     \delta(\rho(\circ), \ov{v'}) = v'' }
     { \circ(\ov{v}) \becomes \ghost{\Phi}(S, \rho, v'')}

     \inferrule[op-static]
     {
       \underlying_\emptyset(T) = B
       \and
       \delta(\circ_B, \ov{v}) = v' }
     {   \circ\smsep T(\ov{v}) \becomes v'}
     
     \inferrule[context]{
     e \rightarrow d
     }{
     E[e] \rightarrow E[d]
     }
   \end{mathpar}
 \end{ruled}
   \caption{LWG reduction}\label{fig:lwg-reduction}
\end{Figure}

\noindent
\textbf{Boxing -- a.k.a.\ ``Make Type''.}
  Consider the $\lwgMethName$ method above which returns \texttt{x} of
  type $\lwgExStruct$ while its signature declares $\lwgExInterface$ as return type.
  Since $\lwgExStruct$ is a structural subtype of $\lwgExInterface$
  ($\lwgExStruct\ \IMPLOP \lwgExInterface$), the method is well-typed
  in WG.
  At the low-level, however, values of type $\lwgExStruct$ are of struct
  kind, while the method return type $\lwgExInterface$ is of interface kind, which
  have incompatible memory formats.
  To resolve such cases, the Go compiler inserts a \texttt{make} operation.
  In the compiled program, the return expression of $\lwgMethName$ becomes:
  {\small$\MAKETYPE{\mathtt{x}}{\typair{\lwgExInterface}{\lwgExInterface}}{\lwgExStruct}{\rho}$}.
 The compiler knows $\lwgExInterface$ and $\lwgExStruct$ statically, hence
$\rho$ can be produced at compile-time.
At runtime, this \texttt{make} operation will box the $\lwgExStruct$ struct value  in a $\lwgExInterface$ interface value.
Notice that the RTTI $\lwgExStruct$ comes from the $\lwgExStruct$ statically embedded by the
compiler.

\smallskip
\centerline{\small
  $\MAKETYPE{\ghost\lwgExStruct\br{\ldots}}{\typair{\lwgExInterface}{\lwgExInterface}}{\lwgExStruct}{\rho}
  \quad \becomes \quad
  \ghost{{\typair{\lwgExInterface}{\lwgExInterface}}}(\lwgExStruct, \rho, \ghost\lwgExStruct\br{\ldots})$
}
\smallskip

\noindent
\textbf{Box conversion -- a.k.a.\ ``Change Type''.}
Assume we add another method:
\begin{lstlisting}[language=go,basicstyle=\footnotesize\ttfamily]
func (x $\lwgExStruct$) $\lwgMethNameBar$() any { return x.$\lwgMethName$() }
\end{lstlisting}
The return types of $\lwgMethName$ and $\lwgMethNameBar$ are both of
interface kind.
However, interface values contain method tables that must \emph{align} with
the type of the interface.
The method table of a $\lwgExInterface$ interface value contains $\lwgMethName$,
whereas the method table of $\mathtt{any}$ is empty, regardless of the
boxed struct value (i.e., $\ghost{\lwgExStruct}\br{\ldots}$ for both).
To resolve such cases, the Go compiler inserts a \emph{Change-Type} operation.
The return expression of $\lwgMethNameBar$ becomes
{\small
$\CHANGETYPE{\mathtt{x}.\lwgMethName()}{\ghostytypair{\mathtt{any}}{\mathtt{any}}}{\lwgExInterface}$
}.
At runtime, this operation will \emph{re}-box the $\lwgExStruct$ struct value (that was
boxed by the $\lwgExInterface$ interface value returned by $\lwgMethName$) in an
$\mathtt{any}$ interface value.
To do so, it \emph{dynamically} creates the new method table (cf.\ Make) based
on the \emph{RTTI} $\lwgExStruct$ carried by the old interface value and the \emph{new
interface type} $\mathtt{any}$ (embedded by the
compiler)---in this case, yielding the empty method table.

\smallskip
\centerline{\small
    $\CHANGETYPE{\ghost{{\typair{\lwgExInterface}{\lwgExInterface}}}(\lwgExStruct, \rho, \ghost\lwgExStruct\br{\ldots})}
    {\ghostytypair{\mathtt{any}}{\mathtt{any}}}{\lwgExInterface}
    \ \ \becomes\ \
    {\ghost{{\typair{\mathtt{any}}{\mathtt{any}}}}(\lwgExStruct, \emptyset, \ghost\lwgExStruct\br{\ldots})}$
}
\smallskip

\noindent
\textbf{Static cast -- a.k.a.\ ``Static change''.}
Assume we add another method:
\begin{lstlisting}[language=go,basicstyle=\footnotesize\ttfamily]
func (x $\lwgExStruct$) $\lwgMethNameBaz$() struct{p Point; r int} { return x }
\end{lstlisting}
The ``Static-Change'' operation performs a ``conversion'' between
assignable non-interface types.
This operation does not perform any actions related to boxing such as
RTTI or method tables (in Go it supports reflection which we are not
modelling here).
The return expression of $\lwgMethNameBaz$ becomes:
{\small
  $\STATICCHANGE{\mathtt{x}}{ \struct~\br{p~\texttt{Point}, r~\texttt{int}} }{\lwgExStruct}$
}.

Since it is converting between struct types, the compiler statically knows
both the old and the new types anyway.
Static-change is necessary to obtain a type
safety result for LWG.

\smallskip
\centerline{\footnotesize
    $\STATICCHANGE{\ghost\lwgExStruct\br{\texttt{\ghost{Point}}\br{1,1}, \texttt{\ghost{int}}(2)}}{ \struct~\br{p~\texttt{Point}, r~\texttt{int}} }{\lwgExStruct}
    \becomes
    \ghost{\struct~\br{p~\texttt{Point}, r~\texttt{int}}}\br{\texttt{\ghost{Point}}\br{1,1}, \texttt{\ghost{int}}(2)}$
}
\smallskip

We review Figure~\ref{fig:lwg-reduction} in more detail.
Values in LWG consist of structure literals $\ghost{T}\{ \ov{v} \}$,
typed constants $\ghost{T}(c)$, and interface value
$\ghost{\Phi}(S, \rho, w)$.
The reduction rules rely on context $E$, similar to
WG.
Function $\MKTABLE{S}{T}{\Delta}$ returns a table mapping method names
specified in interface $T$ to those implemented in the type named $t$
when $S = t[\ov{R}]$, as well as operators defined on all types listed
in $T$ to those defined on $t$ ($\circ_t$ denote the definition of
$\circ$ for $t$).
When $S$ is an anonymous type, then the table is empty.
Function $\mjump(t.m)$ returns the body of method $m$ defined for type
$t$ as well relevant bindings for parameters.
Note that the type arguments $\ov{T}$ passed to generic types
($t[\ov{T}]$) and methods ($m[\ov{T}](\ov{e})$) correspond to the
dictionaries used by the Go runtime (cf.\ Section~\ref{sec:overview}).
For readability, we syntactically separate type parameters
$\ov\alpha \stoup \ov\beta$ from value parameters $x, \ov{y}$, but
both are treated uniformly as substitutable variables in the
operational semantics.

Following the discussion above, the reduction rules are
straightforward. Rule \textsc{select} extract a field from a struct,
using the field index $i$.
Rules \textsc{make}, \textsc{change-type}, and \textsc{static-change-s},
\textsc{static-change-c} formalise the behaviour of boxing and type
conversion operations discussed above.
Observe %
that \textsc{assert-ok$_S$} allows an interface value to be
``unboxed''.
Rule \textsc{call-dyn} and \textsc{call-static} are for dynamic and
static method calls, respectively.
In the dynamic case, we first have to look-up the ``address'' of
method $m$ in the method table $\rho$; while we can jump directly to
$t.m$ in the static case. In both cases, formal (type) parameters are
substituted by the (type) arguments.
The rules for operators work similarly. In the dynamic case, we
look-up the implementation of $\circ$ for type $S$ via the table
$\rho$; in the static case we can perform the concrete operation directly.

\subsection{Typing in LWG}
\label{sub:lwg-typing}

\begin{Figure}
  \begin{ruled}
    Well-formed types and declarations \hfill \fbox{$ \Delta \vdash  \Phi \nice$}
    \begin{mathpar}
      \inferrule[t@v-b]
      {
        \Delta \vdash T \nice
        \\
        \Delta \vdash V \ok
        \\
        \neg\utyvar_\Delta(V)
      }
      { \Delta \vdash \typair{T}{V} \nice }
      
      \inferrule[t-func]
      {
        \neg\uinterface(t[\ov{\alpha}]) \\
        \distinct(x, \ov{y}) \\
        (\type~t[\ov{\alpha~R}]~T) \in \ov{D} \\
        \emptyset \vdash  \ov{\alpha~R}  \ok \\
        \ov{\alpha: R} \vdash \ov{\beta~S} \ok\\
        \Delta =  \ov{\alpha: R} , \ov{\beta: S}
        \\
        \Delta \vdash \ov{\Phi} \nice \\
        \Delta \vdash \Psi \nice \\
        \Delta \stoup
        x : \typair{t[\ov{\alpha}]}{t[\ov{\alpha}]} \comma \ov{y : \Phi} \vdash e : \Psi
      }
      { \func~(x~t[\ov{\alpha~R}])~m[\ov{\beta~S}](\ov{y~\Phi})~\Psi~\br{\return~e} \ok }
    \end{mathpar}
    Expressions \hfill \fbox{$\Delta \stoup \Gamma \vdash e : \Phi$}
    \begin{mathpar}
      \inferrule[t-dyn-call]
      {
        \Delta \vdash \ov{S} \nice
        \and
        \Delta \stoup \Gamma \vdash e : \typair{R}{V} \\
        \uinterface(V)\\
        (m[\ov{\beta~S'}](\ov{y~\Phi})~\Psi) \in \methods_\Delta(V) \\
        \eta = (\ov{\beta \by S})    \\
        \Delta \stoup \Gamma \vdash \ov{e : \Phi[\eta]} \\
        \ov{\beta[\eta] \imp_\Delta S'[\eta]}
      }
      { \Delta \stoup \Gamma \vdash e.m[\ov{S}](\ov{e}) : \Psi [\eta] }

      \inferrule[t-static-call]
      {
        \Delta \vdash \ov{S} \nice
        \and
        \Delta \stoup \Gamma \vdash e : \typair{R}{R} \\
        \neg\uinterface(R) \\
        (m[\ov{\beta~S'}](\ov{y~\Phi})~\Psi) \in \methods_\Delta(R) \\
        \eta = (\ov{\beta \by S})    \\
        \Delta \stoup \Gamma \vdash \ov{e : \Phi[\eta]} \\
        \ov{\beta[\eta] \imp_\Delta S'[\eta]}
      }
      { \Delta \stoup \Gamma \vdash e \smsep R .m[\ov{S}](\ov{e}) : \Psi [\eta]}      

      \inferrule[t-interface]{
        \emptyset \vdash \typair{T}{V} \nice \\ 
        \neg\uinterface(S) \\
        \uinterface(V)\\
        \neg\uinterface(T) \Rightarrow S=T\\
        \emptyset \stoup \Gamma \vdash e \mathbin{:}  \typair{S}{S} \\
        S \assg_\emptyset T  \imp_\emptyset V \\
        \rho = \MKTABLE{S}{V}{\emptyset}
      }{
        \emptyset \stoup \Gamma \vdash
        \ghost{\typair{T}{V}}(S, \rho, e)
        \mathbin{:}
        \typair{T}{V}
      }

      \inferrule[t-var]
      {
        (x : \Phi) \in \Gamma
      }
      { \Delta \stoup \Gamma \vdash x : \Phi }
      
      \inferrule[t-literal]
      {
        \Delta \vdash T \nice
        \and
        \ustruct(T)
        \and
        \Delta \stoup \Gamma \vdash \ov{e : \Phi}
        \and
        (\ov{\_~\Phi}) = \fields_\Delta(T)
      }
      { \Delta \stoup \Gamma \vdash  \ghost{T}\br{\ov{e}} : \typair{T}{T} }
      
      \inferrule[t-field]
      {
        \ustruct(T)
        \\
        \Delta \stoup \Gamma \vdash e : \typair{T}{T}
        \\
        (\ov{\_~\typair{S}{V}}) = \fields_\Delta(T)
      }
      { \Delta \stoup \Gamma \vdash e.i : \typair{S_i}{V_i} }    
    \end{mathpar}
  \end{ruled}
  \caption{LWG: Typing for declarations and expressions}\label{fig:lwg:type-expr1}\label{fig:lwgg-type-expr}
\end{Figure}

\begin{Figure}
  \begin{ruled}
    Expressions \hfill \fbox{$\Delta \stoup \Gamma \vdash e : \Phi$}        
    \begin{mathpar}
      \inferrule[t-make]
      {
        \Delta \vdash T \nice
        \\
        \Delta \vdash \typair{U}{V} \nice
        \\
        \neg\uinterface(T)
        \\
        \neg\uinterface(U) \Rightarrow (T=U)
        \\
        \uinterface(V)
        \\
        \Delta\stoup\Gamma \vdash e :  \typair{T}{T} \\
        T \assg_\Delta U \imp_\Delta V \\
        \rho = \MKTABLE{T}{V}{\Delta}\\
      }
      {\Delta\stoup\Gamma \vdash  \MAKETYPE{e}{\typair{U}{V}}{T}{\rho} : \typair{U}{V}}
      
      \inferrule[t-change]
      {
        \Delta \vdash \typair{T}{V} \nice        
        \\
        \Delta \vdash \typair{U}{V'} \nice
        \\
        \uinterface(V)
        \\
        \neg\uinterface(U)\Rightarrow T=U
        \\
        \uinterface(V')
        \\
        \Delta \stoup \Gamma \vdash e : \typair{T}{V} \and
        T \assg_\Delta U  \imp_\Delta V'
      }
      {\Delta \stoup \Gamma \vdash \CHANGETYPE{e}{\ghostytypair{U}{V'}}{T} : {\typair{U}{V'}} }

      \inferrule[t-static-change]
      {
        \Delta \vdash T \nice
        \and
        \Delta \vdash U \nice
        \and
        \neg\uinterface(T)
        \and
        \neg\uinterface(U)
        \and
        \Delta \stoup \Gamma \vdash e : \typair{T}{T} \and
        \underlying_\Delta(U) = \underlying_\Delta(T)
      }
      {\Delta \stoup \Gamma \vdash \STATICCHANGE{e}{U}{T} : \typair{U}{U}}

      \inferrule[t-assert$_I$]
      {
        \uinterface(T)
        \\
        \Delta \vdash \typair{T}{U} \nice
        \\
        \Delta \stoup \Gamma \vdash e : \typair{S}{V}
        \\
        \uinterface(V) \\
        \uinterface(U)\\
        T\assg_\Delta U \\ 
      }
      { \Delta \stoup \Gamma \vdash e.(\typair{T}{U}) : \typair{T}{U} }

      \squeeze
      \inferrule[t-assert$_S$]
      {
        \neg\uinterface(T)
        \\
        \Delta \vdash T \nice
        \\
        \Delta \stoup \Gamma \vdash e : \typair{S}{V}
        \\
        \uinterface(V)
        \\
        T \imp_\Delta V %
      }
      { \Delta \stoup \Gamma \vdash e.(\typair{T}{T}) : {\typair{T}{T}} }

      \fbox{
      \inferrule[t-stupid]
      {
        \neg\uinterface(T)
        \\
        \Delta \vdash \typair{T}{U} \nice
        \\
        \Delta \stoup \Gamma \vdash e : \typair{S}{V}
        \\
        \uinterface(V) \\
        T\assg_\Delta U \\ 
      }
      { \Delta \stoup \Gamma \vdash e.(\typair{T}{U}) : \typair{T}{U}
      }
    }\\

      \inferrule[t-const]
      {  \Delta \vdash T \nice
        \and
        \underlying_\Delta(T) = 
        \ctype(c)
      }       
      {\Delta \stoup  \Gamma \vdash T(c) : {\typair{T}{T}}}    

      \inferrule[t-op-static]
      {        
        \Delta  \stoup \Gamma \vdash \ov{e} : \typair{T}{T}
        \and
        \underlying_\Delta(T) = B
        \and
        B \in \textit{dom}(\circ)
      }
      {\Delta \stoup  \Gamma \vdash  \circ\smsep T(\ov{e}) : \typair{T}{T}}

      \inferrule[t-op-dyn]
      {
        \Delta  \stoup \Gamma \vdash \ov{e} : \typair{\alpha}{T}
        \and
        (\alpha : T) \in \Delta
        \and
        \{ \underlying_\Delta(V) \mid V \in \types(T) \}
        \subseteq
        \opdom(\circ)
      }
      {\Delta \stoup  \Gamma \vdash \circ(\ov{e}) :  \typair{\alpha}{T}}
      
      \inferrule[t-op-runtime]
      {        
        \Delta  \stoup \Gamma \vdash \ov{e} : \typair{T}{U}
        \and
        \underlying_\Delta(T) = B
        \and
        B \in 
        \{ \underlying_\Delta(V) \mid V \in \types(T) \}
        \subseteq
        \opdom(\circ)                
      }
      {\Delta \stoup  \Gamma \vdash  \circ (\ov{e}) : \typair{T}{U}}
    \end{mathpar}
  \end{ruled}
  \caption{LWG: additional typing rules}\label{fig:lwgg-type-expr-app}
\end{Figure}

The typing rules for LWG are presented in
Figures~\ref{fig:lwg:type-expr1} and~\ref{fig:lwgg-type-expr-app}.
A key feature of the LWG type system are types of the form
$\typair{T}{V}$, where $T$ represents the logical type and $V$ tracks
the available methods and the runtime representation constraint.
Our compilation (Section~\ref{sec:compilation}) uses these annotated
types to determine boxing status. Indeed type variables denote boxed
values, but after instantiation this information is lost---a type
variable instantiated with a base type loses its boxing
annotation. Without this information, the compiler cannot determine
when (un)boxing is required.

When $T = V$, we have standard Go types that are either consistently
boxed (interfaces) or unboxed (structs/base types). When $T \neq V$,
then $T$ is a type parameter $\alpha$ or a concrete instantiation
that must be boxed and satisfy interface $V$.
For example, $\typair{\texttt{int}}{\texttt{any}}$ represents an
integer value that has been boxed into an {\texttt{any}} interface.
As a general rule, if $\neg\uinterface(V)$ then $T = V$.
The $V$ component determines both the runtime representation
("boxy-ness") and the available method set for dynamic dispatch, while
$T$ preserves the precise information needed for static typing.
This dual tracking is essential because Go's structural subtyping
allows the same logical operation to require different runtime
representations depending on context: a method call on a struct
receiver expects unboxed arguments, while a matching method signature
in a generic interface may expect boxed arguments, yet we cannot
always determine statically which will be needed.

The well-formedness checks for LWG are similar to those of WG. We do
not repeat the rules for type declarations as they are the same as for
WG.
Judgement $\Delta \vdash \typair{T}{V} \nice$ holds when $T$ is a
basic interface (i.e., not a union) and $V$ is a well-formed
type.
An LWG method declaration is well-formed if it is defined on a
concrete type ($t$) and its body $e$ has type $\Psi$, matching the
method signature, with the appropriate context. Note how $x$ has type
$\typair{t[\ov{\alpha}]}{t[\ov{\alpha}]}$.

Method calls are treated similarly to WG but deal with the two forms
of method calls and subtyping is not used between actual arguments and
declared types.
Rule \textsc{t-static-call} applies when the receiver expression has
type $\typair{R}{R}$ with $R$ not an interface.
Rule \textsc{t-dyn-call} applies when the receiver has type
$\typair{R}{V}$ with $V$ an interface.
The method must be in $V$'s method table.

Rule \textsc{t-interface} ensures the boxed
expression $e$ is a concrete struct or base type, that $S$ is
assignable to the interface type $T$, and that $T$ implements the
bound $V$.
The method table $\rho$ must map the methods of the bound $V$ to their
concrete implementations in $S$.
If the value is boxing into a concrete type, then the RTTI must match the
logical type ($S=T$).

Rules \textsc{t-var}, \textsc{t-literal}, and \textsc{t-field} are
straightforward.
Note that, a struct type $T$ is typed with $\typair{T}{T}$ which
implies an unboxed value since $\neg\uinterface(T)$.

We comment on the rules for boxing and conversion, see
Figure~\ref{fig:lwgg-type-expr-app}.
Make, $\MAKETYPE{e}{\ghost{\typair{U}{V}}}{T}{\rho}$, creates an
interface value by boxing a concrete value. Hence Rule \textsc{t-make}
applies only when converting from a non-interface type to an interface
type.
Change-type, $\CHANGETYPE{e}{\ghostytypair{U}{V'}}{T}$, converts
between different interface types by updating the method table while
preserving the boxed value, hence Rule \textsc{t-change} applies only
when both source and target are interface types.
Note that the premises of these rules match those of Rule
\textsc{t-interface} as such terms reduce to interface values.
Static-change, $\STATICCHANGE{e}{U}{T}$, performs conversions between
struct or base types without any boxing operations, hence rule
\textsc{t-static-change} applies only when both $T$ and $U$ are
concrete type.
Type assertions and constants (Rules \textsc{t-assert$_I$} and
\textsc{t-assert$_S$}, and \textsc{t-const}) are straightforward. We
write $\ctype(c)$ for the type of constant $c$, e.g., $\ctype(42) =$
\goinl{int}.
Rule \textsc{t-stupid} is a standard
rule~\cite{Igarashi-et-al-2001,DBLP:journals/pacmpl/GriesemerHKLTTW20}
to avoid that expressions become ill-typed during reduction. It is
used for runtime terms only.

The last three rules deal with static and dynamic form of operators.
Rule \textsc{t-op-static} handles operators on concrete types that are
resolved statically. In this case, the operator must be available on
the underlying type of $T$.
Rule \textsc{t-op-dyn} handles operators for type parameters, in which
case the operator must be available for all underlying types of types
listed in $\alpha$'s type bound ($T$).
Rule \textsc{t-op-runtime} is necessary to handle runtime terms where
a generic parameter bounded by $V$ has been instantiated by a type $T$
whose underlying type is a base type.

\subsection{Metatheory of LWG}

We prove LWG's type safety through progress and preservation
arguments. Unlike WG, type preservation does not require the
implements relation because interface and struct values are distinct
at runtime, and typing requires exact type matches rather than
assignability.
Though this seems restrictive, our compilation procedure
(Section~\ref{sec:compilation}) shows that all well-typed WG programs
can be compiled to LWG using suitable make and change type primitives.

\ROUNDTWO{
\begin{theorem}[Type Preservation]\label{thm:lwg-preservation}
If $\Delta;\Gamma\vdash e : T$ and $e \longrightarrow e'$
then $\Delta;\Gamma \vdash e' : T$.
\end{theorem}
}

We define panics as in WG, noting that progress implies that only type
assertions may effectively fail at runtime, not make or change type.

\begin{theorem}[Progress]\label{thm:lwg-progress}
  If $\emptyset ; \emptyset \vdash e : T$ then either $e$ is a value,
  $e \longrightarrow e'$ for some $e'$, or $e$ panics.
\end{theorem}

\section{Compilation}
\label{sec:compilation}

\begin{Figure}
  \begin{ruled}
  Compilation: types \hfill \fbox{$\tybox{T}{}$}  \fbox{$mM \mapsto mM'$} \fbox{$\Phi \mapsto \Phi'$}   \fbox{$C \mapsto C$} \fbox{$T \mapsto T'$}
  \fbox{$\tycomp{T}{}$}
  \fbox{$\tycomp{\Phi}{}$} 
  \begin{mathpar}
    \inferrule
    {\uinterface(\uppertype{\Phi})}
    {\tybox{\Phi}{\Delta} = \Phi}
    
    \inferrule
    {\neg \uinterface(T)}
    {\tybox{\typair{T}{T}}{\Delta} = \typair{T}{\tt Any}}

    \inferrule[d-sig]
    {
      \ov{S \mapsto S'}
      \and
      \ov{\Phi \mapsto \Phi'}
      \and
      \Psi \mapsto \Psi'
    }
    {m[\ov{\beta~S}](\ov{y~\Phi})~\Psi
      \mapsto
      m_D[\ov{\beta~S'}](\ov{y~\tybox{ \Phi' }{\Delta}})~\tybox{ \Psi' }{\Delta}
    }

    \inferrule[d-type-b]
    { \,
    }
    {
      B \mapsto B
    }

    \inferrule[d-tvar]
    { \,
    }
    {
     \alpha \mapsto \alpha
    }
    
    \inferrule[d-type-named]
    {
      \ov{ S \mapsto S'}
    }
    { t{[\ov S]}\mapsto  t{[\ov S']}
    }

    \inferrule[d-type-struct]
    {
      \ov{\Phi \mapsto \Phi'}
    }
    {
      \struct~\br{\ov{f~{\Phi}}} \mapsto  \struct~\br{\ov{f~\Phi'}}
    }

    \inferrule[d-type-iface]
    {
      \ov{F \mapsto F'}
    }
    {
      \interface~\br{\ov{F}} \mapsto \interface~\br{\ov{F'}}
    }

    \inferrule[d-tilde]
    {
      T \mapsto T'
    }
    {
      \ttilde T \mapsto \ttilde T'
    }

    \inferrule[d-or]
    {
      C \mapsto C'
      \and
      E \mapsto E'
    }
    {
      C \typeor E \mapsto  C' \typeor E'
    }

    \inferrule[d-pair]
    {
      T \mapsto  S
      \and
      V \mapsto U
    }
    {
      \typair{T}{V} \mapsto \typair{S}{U}
    }
        
    \inferrule
    {  T \mapsto S}
    {\tycomp{T}{\Delta} = S}
    
    \inferrule
    {  T \mapsto S
      \and
      V \mapsto U
    }
    {\tycomp{\typair{T}{V}}{\Delta} = \typair{S}{U}}
  \end{mathpar}
\end{ruled}
\caption{\ROUNDTWO{Auxiliary notations and} compilation rules for types}\label{fig:compil-types}
\end{Figure}

Monomorphisation tackles the problem of data layouts for generic code by
generating all possibly needed specialisations, which typically requires a
whole program analysis.
An alternative approach more compatible with separate compilation is a uniform
(or \emph{boxed}) representation using pointer indirection, which requires
runtime boxing/unboxing actions.

Our compilation strategy is based on the latter approach.
The key points and challenges are how we (i) reuse Go's pre-existing runtime
infrastructure for interfaces to perform the type-directed static compilation
and the runtime boxing, and (ii) achieve a uniform compilation strategy given
Go's structural typing, as we discuss in Section~\ref{sec:boxesstructtyp}.

\subsection{Repurposing Go's Runtime Infrastructure}
\label{sec:repurpose}

We repurpose interface values as boxes that represent values of \emph{generic}
type, with boxing performed by LWG's make and change type operations.
This is natural since interface values are already boxed in the Go runtime,
and generic values have interface-like status with bounds that constrain their
operations.
Unboxing is performed by type assertions.
We present our compilation strategy as a translation from well-typed WG
programs to LWG programs using a $\mapsto$ arrow.

Figure~\ref{fig:compil-types} gives the compilation of types to prepare the program with uniform handling of dynamic methods. Each method $m$ occurring in interfaces is renamed to
$m_D$ ($D$ for dynamic) and its signature is adapted so
that all argument and return types are boxed (using
$\tybox{T}{}$). The rules ensure these changes
are made consistently, including in literal interface types in struct definitions.
We introduce $\tycomp{T}{}$ (resp.\ $\tycomp{\Phi}{}$) as a
shorthand for a compiled type.

\begin{Figure}
  \begin{ruled}
    Compilation: synthetic casts \hfill \fbox{$\Delta \stoup \Gamma \vdash e : U \mapsto_\Phi e'$}
  \begin{mathpar}
        \inferrule[make-iface]
    {
      \uinterface(V)
      \and
      \neg\uinterface(S)
      \and
      \uinterface(T)
      \and
      \Delta\stoup\Gamma \vdash e \mapsto e' \and \rho =\MKTABLE{\tycomp{S}{}}{\tycomp{V}{}}{\Delta}
    }
    {\Delta\stoup\Gamma \vdash e \mathbin{:} S \mapsto_{\typair{T}{V}} \MAKETYPE{e'}{{\tycomp{\typair{T}{V}}{\Delta}}}{\tycomp{S}{\Delta}}{\rho}}

    \inferrule[make-bs]
    {
      \uinterface(V)
      \and
      \neg\uinterface(S)
      \and
      \neg\uinterface(T)
      \and
      \Delta\stoup\Gamma \vdash e \mapsto e' \and \rho =\MKTABLE{\tycomp{T}{}}{\tycomp{V}{}}{\Delta}
    }
    {\Delta\stoup\Gamma \vdash e \mathbin{:}S \mapsto_{\typair{T}{V}}
      \MAKETYPE{
        \STATICCHANGE{e'}{\tycomp{T}{}}{\tycomp{S}{}}
      }{{\tycomp{\typair{T}{V}}{\Delta}}}{\tycomp{T}{\Delta}}{\rho}}

    \inferrule[change]
    {
      \uinterface(U)
      \and
      \uinterface(V)
      \and
      \Delta\stoup\Gamma \vdash e \mapsto e'}
    {\Delta\stoup\Gamma \vdash e  \mathbin{:} U \mapsto_{\typair{T}{V}} \CHANGETYPE{e'}{\tycomp{\ghostytypair{T}{V}}{\Delta}}{\tycomp{U}{\Delta}}}

    \inferrule[static-change]
    {
      \neg\uinterface(U)
      \and
      \neg\uinterface(T)
      \and \Delta\stoup\Gamma \vdash e \mapsto e'}
    {\Delta\stoup\Gamma \vdash e  \mathbin{:} U \mapsto_{\typair{T}{T}} \STATICCHANGE{e'}{\tycomp{T}{\Delta}}{\tycomp{U}{\Delta}}}
  \end{mathpar}
\end{ruled}
\caption{Compilation rules: generate synthetic casts}\label{fig:synthcast}
\end{Figure}

To compile expressions and method declarations, we rely on an auxiliary
judgement given in Figure~\ref{fig:synthcast}. In
$\Delta ; \Gamma \vdash e : U \mapsto_{\Phi} e'$, $e$ is a
well-typed WG expression, $e'$ is an LWG expression that $e$ compiles into, and
$\Phi$ is the target LWG type of $e'$.
This judgement determines which synthetic casts (make,
change or static-change) are needed in the resulting expression.
Interface values are created in two situations.
Rule {\sc make-iface} handles non-interface expressions used at
interface types (standard Go interface creation), computing method
table $\rho$ with methods of $S$ listed in interface $V$.
Rule {\sc make-bs} handles generic value boxing: non-interface source
type $S$ with non-interface target type $T$ but interface bound $V$,
creating an interface value with method table from $T$ for 
$V$.

For example, consider the WG declaration \goinl{type Cell[T any]
  struct\{ val T@any \}}.
Initialising \goinl{Cell[int]\{int(42)\}} requires boxing the integer
for uniform representation.
The integer field compiles to
$\MAKETYPE{\mathtt{int}(42)}{{\tycomp{\typair{\mathtt{int}}{\mathtt{any}}}{}}}{\mathtt{int}}{\rho}$
via Rule {\sc make-bs}. Field access generates
unboxing via Rule {\sc select-boxed} and the bound $\ghost{\texttt{any}}$ allows us to identify a boxed \goinl{int}.

Rule {\sc change} handles interface type $U$ to interface type $V$
conversion via Rule \textsc{change-type} (Figure~\ref{fig:lwg-reduction}).
Unlike {\sc make}, the method table must be computed at runtime when the
runtime type becomes known.
This showcases reuse of Go's interface infrastructure: existing
interface values instantiating type parameters are adapted to the
interface specified by the type parameter bound.

Rule {\sc static-change} applies when neither source type $U$ nor target
bound $T$ are interfaces. Such transformations convert between assignable
WG types, similar to type conversions in WG's operational semantics.
Consider \goinl{type StringerCell struct\{ val Stringer@Stringer \}}.
Assume \goinl{int}~$\imp_\emptyset$~\goinl{Stringer}, in \goinl{StringerCell\{int(42)\}}, the field \goinl{int(42)} compiles to
$\MAKETYPE{\mathtt{int}(42)}{{\tycomp{\typair{\mathtt{Stringer}}{\mathtt{Stringer}}}{}}}{\mathtt{int}}{\rho}$
by rule {\sc make-iface}.
Rules \textsc{box} and \textsc{nbox} are variants to ensure all values
are boxed, this is used in dynamic method calls and matches the
$\tybox{\Phi}{}$ function introduced earlier.

The compilation of expressions is formalised in
Figures~\ref{fig:compexp1} and~\ref{fig:compmcalls}.
In judgement $\Delta ; \Gamma \vdash e \mapsto e'$, $e$ is a
well-typed WG expression (under suitable contexts $\Delta$ and
$\Gamma$) and $e'$ is an LWG expression that $e$ compiles into.
Rules \textsc{Var}, \textsc{Const}, and \textsc{Assert} are
straightforward.
Rule \textsc{Conversion} essentially replaces a WG static conversion
with the appropriate synthetic cast.

The compilation of field selection is given by Rules
\textsc{Select-boxed} and \textsc{Select-unboxed}, the selected rule
depends on whether the struct field needs to be unboxed or not.
For Rule \textsc{Select-boxed}, the type of field $f_i$ is
$\typair{S_i}{V_i}$, where $S_i$ is not a type variable and $V_i$ is
an interface. This means that either the field denotes a generic type
that has been instantiated or it is a value of interface type (in
which case $S_i = V_i$).
In the former
case the value will be boxed at runtime but is treated in the
remaining source WG program as if it were a plain value of type
$S_i$.
In both cases we insert a type assertion to $\tycomp{S_i}{}$ which
will unbox the value in the former case and leave it unchanged (as
needed) in the latter.
Rule \textsc{Select-unboxed} applies when $S_i$ is a type variable or
$V_i$ is not an interface type.
No unboxing is needed: if $S_i$ is a type variable, the field must
remain boxed to use operations from its bound; if $V_i$ is not an
interface, the field is already unboxed.

Rules \textsc{op-static} and \textsc{op-dyn} deal with operators.
Recall that the WG type system enforces that the operands have all the
same type, which is either a type whose underlying type is a base
type, or a type variable whose bound supports $\circ$.
The first rule applies when the type of the operand is not a type
variable, hence it must be a static operator invocation. If the type
of the operand is a type variable, then it is a dynamic invocation and
the second rule applies.

\begin{Figure}
  \begin{ruled}
    Compilation: expressions \hfill \fbox{$\Delta \stoup \Gamma \vdash e \mapsto e'$}
    \begin{mathpar}
      \inferrule[Var]{ \,
      }{
        \Delta\stoup\Gamma \vdash x
        \mapsto x
      }

      \inferrule[Const]{ \,
      }{
        \Delta\stoup\Gamma \vdash T(c)
        \mapsto T(c)
      }

    \inferrule[Assert]{
      \Delta\stoup\Gamma \vdash e \mapsto e'
    }{
      \Delta\stoup\Gamma \vdash e . (T)
      \mapsto
      e' . (\typair{\tycomp{T}{\Delta}}{\tycomp{\bounds_\Delta(T)}{\Delta}})
    }
    
    \inferrule[Conversion]{
      \Delta\stoup\Gamma \vdash e \mathbin{:} S
      \and
      \Delta\stoup\Gamma \vdash e :S \mapsto_{\typair{T}{\bounds_\Delta(T)}} e'
    }{
      \Gamma \vdash T(e)
      \mapsto
      e'
    }
    \qquad
    \inferrule[Struct]{
      \ov{\_~\Phi} = \fields_\Delta(T)
      \and
      \Delta\stoup\Gamma \vdash \ov{e : U}
      \and
      \ov{\Delta\stoup\Gamma \vdash e  \mathbin{:} U \mapsto_\Phi e'}
    }{
      \Delta\stoup\Gamma \vdash T \{ \ov{e} \}
      \mapsto\ghost{ {\tycomp{T}{\Delta}}} \{ \ov{e'} \}
    }

    \inferrule[Select-boxed]{
      \Delta\stoup\Gamma \vdash e \mapsto e'
      \and
      \Delta\stoup\Gamma \vdash e \mathbin{:} T
      \and
      \ov{f\; \typair{S}{V}} = \fields_\Delta(T)
      \and
      \neg\utyvar_\Delta(S_i)
      \and
      \uinterface(V_i)
      \and
      f_i = f
    }{
      \Delta\stoup\Gamma \vdash e . f
      \mapsto
      e' . i . {\color{black} (\typair{\tycomp{S_i}{\Delta}}{\tycomp{S_i}{\Delta}}) }
    }

      \inferrule[Select-unboxed]{
        \Delta\stoup\Gamma \vdash e \mapsto e'
        \and
        \Delta\stoup\Gamma \vdash e \mathbin{:} T
        \and
        \ov{f\; \typair{S}{V}} = \fields_\Delta(T)
        \and
        ( \utyvar_\Delta(S_i)  \lor  \neg \uinterface(V_i))
        \and
        f_i = f
      }{
        \Delta\stoup\Gamma \vdash e . f
        \mapsto
        e' . i
      }

      \inferrule[op-static]{
        \Delta  \stoup \Gamma \vdash \ov{e} : T
        \and
        \neg \utyvar_\Delta(T)
        \and
        \ov{\Delta\stoup\Gamma \vdash e \mapsto e'}
      }
      {
        \Delta\stoup\Gamma \vdash
        \circ(\ov{e})
        \mapsto
        \circ\smsep {\tycomp{T}{\Delta}}(\ov{e'})
      }
      
      \inferrule[op-dyn]{
        \Delta  \stoup \Gamma \vdash \ov{e} : \alpha
        \and
        \ov{\Delta\stoup\Gamma \vdash e \mapsto e'}
      }
      {
        \Delta\stoup\Gamma \vdash
        \circ(\ov{e})
        \mapsto
        \circ(\ov{e'})
      }
    \end{mathpar}
  \end{ruled}
  \caption{Compilation rules: expressions (part 1)}\label{fig:compexp1}
\end{Figure}

\subsection{Boxing and Structural Typing}
\label{sec:boxesstructtyp}

We now move on to the compilation of methods and method calls. It
turns out that Go's structural typing makes the treatment of methods
especially subtle as we illustrate below.

\begin{example} \label{ex:duck-typing-adapters}
Consider the following WG code:
\begin{golst}
type Processor[T any] interface { process(input T) T }
type IntProcessor struct {} // IntProcessor $\imp_\emptyset$ Processor[int]
func (p IntProcessor) process(input int) int { return input }
type GenericProcessor[T any] struct {} // GenericProcessor[T] $\imp_{\{\texttt{T:any}\}}$ Processor[T]
func (p GenericProcessor[T any]) process(input T) T { return input }
type Client struct {}
func (c Client) useProcessor(processor Processor[int]) int { return processor.process(int(42)) }
\end{golst}
\end{example}

The code above defines a generic interface \goinl{Processor[T]} with a
method \goinl{process}.
We then define a non-generic struct
type \goinl{IntProcessor} with a \goinl{process} method from
\goinl{int} to \goinl{int}.
Notably, \goinl{IntProcessor} implements
\goinl{Processor[int]}.
Struct type \goinl{GenericProcessor} is a generic version
\goinl{IntProcessor}, it implements \goinl{Processor[T]} for any
\goinl{T}.
We further define a \goinl{useProcessor}
method on a client struct, which takes an argument of type
\goinl{Processor[int]} and invokes the method \goinl{process} on it.

Following our uniform compilation strategy and considering interface
\goinl{Processor}, it would be expected that its \goinl{process}
method would take a boxed argument and return \ROUNDTWO{a} result.
Thus, when compiling the body of \goinl{useProcessor} the compiler
would box the argument to \goinl{process} and then unbox the return
value, in order for it to match the signature of \goinl{useProcessor}
which has a non-generic return.
The attentive reader might have already identified the problem with
this naive approach: while \goinl{IntProcessor} implements
\goinl{Processor[int]}, its  \goinl{process} method is a plain
\goinl{int} to \goinl{int} method with no reason for it to receive or
return boxed values.
Thus if we pass a value of type \goinl{IntProcessor} as an argument to
\goinl{useProcessor}, the box and unbox operations would be
incorrect. If instead we were to pass a value of type
\goinl{GenericProcessor[int]}, which also implements
\goinl{Processor[int]}, the boxing strategy described above would be
correct.

\begin{Figure}
  \begin{ruled}
  Compilation: boxing casts \hfill
  \fbox{$\Delta \stoup \Gamma \vdash d : U \Mapsto_\Phi d'$}
  \begin{mathpar}
    \inferrule[box]{
      \neg\uinterface(V)
      \and
      \Delta \stoup\Gamma \vdash d  \mathbin{:} U  \mapsto_{\typair{T}{\texttt{Any}}}  d'
    }
    {
      \Delta \stoup\Gamma \vdash d  \mathbin{:} U  \Mapsto_{\typair{T}{V}}  d'
    }

    \inferrule[nbox]{
      \uinterface(V)
      \and
      \Delta \stoup\Gamma \vdash d  \mathbin{:} U  \mapsto_{\typair{T}{V}}  d'
    }
    {
      \Delta \stoup\Gamma \vdash d  \mathbin{:} U  \Mapsto_{\typair{T}{V}}  d'
    }
  \end{mathpar}
    
  Compilation: method calls \hfill \fbox{$\Delta \stoup \Gamma \vdash e \mapsto e'$}
  \begin{mathpar}

    \inferrule[Call${}_S$-boxed]{
      \Delta\stoup\Gamma \vdash e \mathbin{:}  R
      \and
      \neg\uinterface(R)
      \and
      m[\ov{\beta~S''}](\ov{y\:\Phi})\: {\color{black} \typair{T}{V} } \mathbin{\in} \methods_\Delta(R)
      \and
      \eta = (\ov{\beta \by S})
      \and
      \uinterface(V[\eta])
      \\
      \neg\utyvar_\Delta(T[\eta])
      \and
      \Delta\stoup\Gamma \vdash e \mapsto e'
      \and
      \Delta\stoup\Gamma \vdash \ov{d : U}
      \and
      \Delta \stoup\Gamma \vdash \ov{ d  \mathbin{:} U  \mapsto_{\Phi[\eta]}  d'}
    }{
      \Delta\stoup\Gamma \vdash e . m[\ov{S}](\ov{d})
      \mapsto
      e'\smsep \tycomp{R}{\Delta} .  m[\ov{\tycomp{S}{\Delta}}]( \ov{d'}) .
      {\color{black} (\typair{\tycomp{T[\eta]}{\Delta}}{\tycomp{T[\eta]}{\Delta}}) }
    }
    \vspace{-3pt}

      \inferrule[Call${}_S$-unboxed]{
        \Delta\stoup\Gamma \vdash e \mathbin{:}  R
        \and
        \neg\uinterface(R)
        \and
        m[\ov{\beta~S''}](\ov{y\:\Phi})\: {\color{black} \typair{T}{V} } \mathbin{\in} \methods_\Delta(R)
        \and
        \eta = (\ov{\beta \by S})
        \\
        ( \neg \uinterface(V[\eta]) \lor \utyvar_\Delta(T[\eta]))
        \and
        \Delta\stoup\Gamma \vdash e \mapsto e'
        \and
        \Delta\stoup\Gamma \vdash \ov{d : U}
        \and
        \Delta \stoup\Gamma \vdash \ov{ d  \mathbin{:} U  \mapsto_{\Phi[\eta]}  d'}
      }{
        \Delta\stoup\Gamma \vdash e . m[\ov{S}](\ov{d})
        \mapsto
        e'\smsep \tycomp{R}{\Delta} .  m[\ov{\tycomp{S}{\Delta}}]( \ov{d'})
      }

      \vspace{-3pt}
      \inferrule[Call${}_I$]{
      \Delta\stoup\Gamma \vdash e \mathbin{:} R
      \and
      \uinterface(R)
      \\
      m[\ov{\beta~S''}](\ov{y\:\Phi})\: {\color{black} \typair{T}{V} } \mathbin{\in} \methods_\Delta(R)
      \and
      \neg\utyvar_\Delta(T[\eta])
      \and
      \\
      \eta = (\ov{\beta \by S})
      \and
      \Delta\stoup\Gamma \vdash e \mapsto e'
      \and
      \Delta\stoup\Gamma \vdash \ov{d : U}
      \and
      \Delta \stoup\Gamma \vdash \ov{ d  \mathbin{:} U  \Mapsto_{\Phi[\eta]}  d'}
    }{
      \Delta\stoup\Gamma \vdash e . m[\ov{S}](\ov{d})
      \mapsto
      e' .  m_D[\ov{{\tycomp{S}{\Delta}}}](\ov{d'}) .
      {\color{black} (\typair{\tycomp{T[\eta]}{\Delta}}{\tycomp{T[\eta]}{\Delta}}) }
    }
    
\vspace{-3pt}
      \inferrule[Call${}_{I\alpha}$]{
        \Delta\stoup\Gamma \vdash e \mathbin{:} R
        \and
        \uinterface(R)
        \\
        m[\ov{\beta~S''}](\ov{y\:\Phi})\: {\color{black} \typair{T}{V} } \mathbin{\in} \methods_\Delta(R)
        \and
        \utyvar_\Delta(T[\eta])
        \\
        \eta = (\ov{\beta \by S})
        \and
        \Delta\stoup\Gamma \vdash e \mapsto e'
        \and
        \Delta\stoup\Gamma \vdash \ov{d : U}
        \and
        \Delta \stoup\Gamma \vdash \ov{ d  \mathbin{:} U  \Mapsto_{\Phi[\eta]}  d'}
      }{
        \Delta\stoup\Gamma \vdash e . m[\ov{S}](\ov{d})
        \mapsto
        \CHANGETYPE{e' .  m_D[\ov{\tycomp{S}{\Delta}}](\ov{d'}) }{
          \tycomp{\ghostytypair{T[\eta]}{V[\eta]}}{}        
        }{{\tycomp{T[\eta]}{\Delta}}}
      }
    \end{mathpar}
  \end{ruled}
  \caption{Compilation rules: expressions (part 2)}
  \label{fig:compmcalls}
\end{Figure}

This example reveals a fundamental challenge in the uniform
compilation of methods in the presence of structural typing: When
compiling a method body we cannot \emph{a priori} predict which
interfaces the type will implement and so we cannot locally determine
what arguments must be boxed or not (e.g.~method \goinl{process} in
\goinl{IntProcessor} and in \goinl{GenericProcessor}). Dually, when
considering the call site to an interface method (e.g.,~the call to
\goinl{process} in \goinl{useProcessor}) we also cannot statically
determine the type of the method receiver and so we cannot predict
whether arguments and return values need to be boxed by inspecting the
signature of the interface.
We can manifest the issue further by considering an additional
interface \goinl{ type SimpleProcessor interface \{process(x int)
  int\}}.  We have that both \goinl{IntProcessor} and
\goinl{GenericProcessor[int]} implement \goinl{SimpleProcessor} and
any context performing a call to \goinl{process} on a value of type
\goinl{SimpleProcessor} has no way of determining if the receiver
expects boxing (i.e.,~\goinl{GenericProcessor}) or not
(i.e.~\goinl{IntProcessor}).

\paragraph{Adaptor Methods}
Our solution to this issue is to compile \emph{two} versions of each
method $m$: One whose signature (and name) is identical to the source
WG signature and whose body is compiled accordingly; and a
\emph{dynamic adaptor} version $m_D$ of the method, that expects \emph{all}
arguments and return value to be boxed. The body of $m_D$
redirects to $m$ with the appropriate unboxing operations.

We first address the compilation of method calls. The rules are given
in Figure~\ref{fig:compmcalls}.
The rules are split into two categories, depending on whether the call
is on an interface (i.e.,~a \emph{dynamic} call) or on a non-interface
(i.e.,~a \emph{static} call): Rule {\sc Call}$_I$ and Rule {\sc
  Call}$_{I\alpha}$ deal with the former, and the two {\sc Call}$_S$
rules deal with the latter.
Two rules per category are required to deal with potential unboxing of
the value returned by the call.

In the static call cases, the arguments are compiled according to
their WG signature, via the synthetic cast judgement. If the return
type of $m$ is boxed then we unbox it ({\sc Call}$_S${\sc -boxed})
since the calling context expects an unboxed result; otherwise, we
leave it unchanged ({\sc Call}$_S${\sc -unboxed}).

The compilation of dynamic calls boxes all arguments, as needed by the
adaptor (the $\Mapsto$ judgement boxes all values, see
Figure~\ref{fig:compmcalls}).
As for the return type, we check whether it is a type variable in
WG. If so, then the method call happens in a context where the return
type is still parametric and must therefore be boxed according to the
bound of the type variable.
We will see that the adaptor method $m_D$ returns a boxed value with
an \emph{empty} method table (the bound information is only determined
at the call site). Thus, we {\sc change} the type accordingly, so that
the appropriate method table is generated (rule {\sc
  Call}$_{I\alpha}$).
If the return type is \emph{not} a type variable, then we type assert
the result to the return type. This has the effect of unboxing the
result if needed by the calling context.

Revisiting our initial example, the call to \goinl{process} in the
body of \goinl{useProcessor} compiles to a dynamic call to the adaptor
for \goinl{process}, boxing the integer argument and unboxing the
return value, since the return of \goinl{process} in
\goinl{Processor[int]} is not a type variable, Rule \textsc{Call$_I$} applies, we obtain:
$
\texttt{processor.process}_D(
\CHANGETYPE{\texttt{int}(42)}{\ghostytypair{\texttt{int}}{\texttt{any}}}{\texttt{int}}).(\typair{\texttt{int}}{\texttt{int}})
$.

The compilation of method declarations is given in
Figure~\ref{fig:methdecl}. Rule {\sc d-func} generates two methods per
WG method $m$.
Method $m$ is the compiled version of the original method, with an identical signature (modulo
type compilation) and with method body obtained by compiling the
source body.
Adaptor method $m_D$ bridges dynamic method call destined to $m$.
Its signature is
obtained by forcing the types of arguments and the return type to be
boxed---achieved by the auxiliary function $\tybox{\Phi}{}$.

The adaptor method body consists of a (static) call to method $m$,
generated by an auxiliary $\mathit{Adapt}$ function that inspects the WG
signature of $m$ to determine which arguments need to be unboxed and
whether the return value of $m$ needs to be boxed.
To ensure the adaptor returns a boxed result, we insert a make
instruction with an empty method table via the synthetic cast
$({\tycomp{V}{\Delta}} \rightarrow_\emptyset
\ghost{{\tycomp{\typair{V}{\texttt{Any}}}{\Delta}}})$ when required.
Finally, Rule \textsc{d-type} compiles type declarations. Its sole
purpose is to apply the $\tycomp{\Phi}{}$ transformation on all
(anonymous) interfaces it may contain.

\begin{Figure}
  \begin{ruled}
    Compilation: declarations \hfill \fbox{$\Delta \stoup \Gamma \vdash D \mapsto \ov{D}$}
    \begin{mathpar}
      \inferrule[d-func]
      {
        mM = m[\ov{\beta~S}](\ov{y~\Phi})~\Psi
        \and
        \ov{\Phi} = \ov{\typair{T}{T'}}
        \and
        \Delta = \ov{\alpha : R},\ov{\beta : S}
        \and
        \Gamma = x: t[\ov{\alpha}], \ov{y: T}
        \\\\
        \ov{R \mapsto R'}
        \and
        \ov{S \mapsto S'}
        \and
        \Delta \stoup \Gamma  \vdash e \mathbin{:} V
        \and
        \Delta \stoup \Gamma \vdash e  \mathbin{:} V \mapsto_\Psi e'
        \\\\
        D
        =
        {\func~(x~t[\ov{\alpha~R'}])~m_D[\ov{\beta~S'}](\ov{y~\tycomp{\tybox{\Phi}{\Delta}}{\Delta}})~\tycomp{\tybox{\Psi}{\Delta}}{\Delta}~\br{\, \return~\textit{Adapt}((x~t[\ov{\alpha~R}])~mM) \, }
        }
      }
      {
        {\func~(x~t[\ov{\alpha~R}])~mM ~\br{\return~e}}
        \mapsto
        \{D, \ \func~(x~t[\ov{\alpha~{R'}}])~m[\ov{\beta~S'}](\ov{y~\tycomp{\Phi}{\Delta}})~\tycomp{\Psi}{\Delta}~\br{\return~e'}\}
      }

      \inferrule
      {
        \uinterface(V')
        \and
        \neg\uinterface(T'_i) \implies d_i = y_i . (\tycomp{\typair{T_i}{T'_i}}{\Delta})
        \and
        \uinterface(T'_i) \implies d_i =  \CHANGETYPE{y_i}{\typair{T_i}{T'_i}}{T_i}
      }
      {
        \textit{Adapt}(
        (x~t[\ov{\alpha~R}])~m[\ov{\beta~S}](\ov{y~\typair{T}{T'}})~\typair{V}{V'})
        =
        x. t[\ov{\alpha}] \smsep m[\ov{\beta}](\ov{d})
      }

      \inferrule
      {
        \neg \uinterface(V)
        \and
        \neg\uinterface(T'_i) \implies d_i = y_i . (\tycomp{\typair{T_i}{T'_i}}{\Delta})
        \and
        \uinterface(T'_i) \implies d_i = \CHANGETYPE{y_i}{\typair{T_i}{T'_i}}{T_i}
      }
      {
        \textit{Adapt}(
        (x~t[\ov{\alpha~R}])~m[\ov{\beta~S}](\ov{y~{\typair{T}{T'}}})~\typair{V}{V}
        )
        =
        \MAKETYPE{x. t[\ov{\alpha}] \smsep m[\ov{\beta}](\ov{d})}{{\tycomp{\typair{V}{\texttt{Any}}}{\Delta}}}{\tycomp{V}{\Delta}}{\emptyset}
      }

      \inferrule[d-type]
      {
        \ov{R \mapsto R'}
        \and
        T \mapsto T'
      }
      {
        \type~t{[\ov{\alpha~R}]}~T \mapsto \type~t{[\ov{\alpha~R'}]}~T'
      }
    \end{mathpar}
  \end{ruled}
  \caption{Compilation rules: method declarations}\label{fig:methdecl}
\end{Figure}

\subsection{Properties of Compilation}

Compilation of a well-typed WG program results in a well-typed LWG
program that preserves the behaviour of the original WG program.
The former property is embodied in Theorem~\ref{thm:comptyp}, which
states that a well-typed WG expression of type $T$ compiles to a
well-typed LWG expression of type $\lfloor T \rfloor @ \lfloor
\bounds_\Delta(T)\rfloor$, presupposing all declarations in the ambient
program are compiled. The resulting type captures the essence of our
typing scheme for LWG, where expressions are assigned a (logical) type
but also track the type bound. %

\begin{theorem}\label{thm:comptyp}
  ~
  If $\Delta ; \Gamma \vdash e : T$ and $\Delta ;\Gamma \vdash e \mapsto
e'$ then $\lfloor \Delta \rfloor ; \lfloor \Gamma \rfloor_\Delta
\vdash e' : \lfloor T \rfloor @ \lfloor \bounds_\Delta(T)\rfloor$.
\end{theorem}

Before stating our behavioural correspondence result, we define how WG
and LWG programs should be related. All definitions introduced here
are given formally in \appendixref{app:beh-eq}{Section}.

There is a lowering simulation between a WG program $\ov{D} \prog e$ and
an LWG program $\ov{D'} \prog e'$, if whenever $e \becomes d$
then $e' \becomesast d'$, and $e'$ and $d'$ are in a lowering
simulation. Also, if $e$ is a value, then $d$ is an \emph{equivalent}
value; and if $e$ is stuck, then so is $d$.
A lifting simulation is the reverse, e.g., if $d$ can make a move,
then $e$ should be able to match it, etc.
The theorem below relies on relation $\ceqcast{}{}$ (a weak
bisimulation) on LWG expressions and which essentially adds or removes
synthetic casts.

\begin{theorem}\label{thm:compbisim}
  Suppose $\ov{D} \prog e \ok$, then:
    $
      \mathcal{R} \eqdef
      \{ \,
      (e,d') \st
      \exists d . \
      \emptyset \stoup \emptyset \vdash
      e
      \mapsto
      d
      \text{ and }
      \ceqcast{d}{d'}
      \,
      \}
      $
    is a lowering simulation; and
    $
      \mathcal{R} \eqdef
      \{ \,
      (d,e) \st
      \exists d' . \
      \emptyset \stoup \emptyset \vdash
      e
      \mapsto
      d'
      \text{ and }
      \ceqcast{d}{d'}
      \,
      \}
      $
    is a lifting simulation.
\end{theorem}

\section{Related Work}
\label{lab:related}
\label{sec:related}

This work fits within the established tradition of programming
language research that studies formal models of real-world languages,
e.g.,~\cite{DBLP:journals/pacmpl/GriesemerHKLTTW20,
  Igarashi-et-al-2001,DBLP:conf/popl/Grigore17,DBLP:conf/popl/AminR17,Kennedy-and-Syme-2001,DBLP:journals/pacmpl/0002JKD18,DBLP:conf/pldi/ParkSR15}.

\textit{Formal investigations of Go.}
\citet{DBLP:journals/pacmpl/GriesemerHKLTTW20}~presented the first
core formalism of Go in FGG (Featherweight Generic Go), modelling
generic types and structural subtyping.
WG was initially motivated by an investigation into the type safety of the
features of type unions and type sets.
This quickly led to the observation that FGG's core subset fails to
characterise significant semantic features of Go.
By relying on nominal typing of structs and omitting underlying and
anonymous types, FGG does not accurately model Go's structural typing
that allows methods to accept differently-named structs with the same
underlying structure.
FGG relies solely on interface-driven subtyping, which does not
capture Go's full structural typing nor expose the subtleties in
different type coercions: assignability (a non-transitive superset of
interface subtyping), static conversions (safe no-ops), and dynamic
assertions (requiring RTTI).

Various approaches to implementing Go have been studied based on FGG.
\citet{DBLP:journals/pacmpl/GriesemerHKLTTW20} formalised monomorphisation as a
translation from FGG to non-generic FG.
\citet{DBLP:journals/pacmpl/EllisZYS22}~presented an approach based on
translating FGG to non-generic FG using structs to implement RTTI dictionaries
for generic type arguments.
\citet{DBLP:journals/jfp/SulzmannW23}~presented an approach of translating FGG
sans dynamic type assertions (and RTTI) to an \emph{untyped}
$\lambda$-calculus similar to that known from Haskell type classes.
All these remain within FGG, and thus do not consider the features that we
have introduced in WG.
None consider the underlying mechanics of the Go runtime, such as interface
values, RTTI, runtime type conversions and static versus dynamic resolution of
operations, as in LWG.
In LWG, we distinguish concrete runtime items (e.g., RTTI)
from ghost entities used by our metatheory to establish our correctness
properties.
Note that LWG actions and rules that operate solely on ghost entities
correspond to no-ops in an actual implementation.

Regarding correctness, \citet{DBLP:journals/pacmpl/GriesemerHKLTTW20} shows FGG
programs and their FG monomorphisations are bisimilar.
\citet{DBLP:journals/pacmpl/EllisZYS22} show their FGG to FG translation is a
bisimulation up to dictionary resolution.
\citet{DBLP:journals/jfp/SulzmannW23} show their (untyped) $\lambda$-calculus
translation is value and divergence preserving.
By contrast, we establish a behavioural equivalence
between high level WG and compilations into low level LWG.
The challenges we address in this paper relate to establishing
safety for LWG and the correspondence between high and low level
behaviours.

\textit{Implementations of generics.}
Various approaches have been used to implement generics in OO languages with
differing implications and trade-offs.
Java employs an erasure approach~\cite{Bracha-et-al-1998} where generic type
information is discarded by the compiler, which restricts programs from
performing operations that may dynamically depend on such.
C\# employs a mixed approach~\cite{DBLP:conf/popl/YuKS04,Kennedy-and-Syme-2001} where generic code
for object types is generated once statically, supported by RTTI at runtime,
but generic code for primitives types is specialised dynamically by JIT
compilation.
The type theoretic study of RTTI for polymorphism
dates back to~\citet{DBLP:conf/popl/HarperM95} and~\citet{DBLP:conf/icfp/CraryWM98}.
Rust employs static
monomorphisation~\cite{Turon-2015,rust-book-generics}, enabling
zero-cost abstractions at the cost of output code size and compilation time.
All use nominal subtyping, and do not tackle the interplay between structural
subtyping, generics and non-uniform representations (interfaces vs.\
constants/structs) as in Go.

As discussed earlier, Go employs static
monomorphisation~\cite{DBLP:journals/pacmpl/GriesemerHKLTTW20} and
runtime dictionary passing~\cite{godictionaries}. The Go compiler
monomorphises code up to GC shapes with RTTI for lost type information
to balance specialisation against code size.
This incurs compilation costs and
restrictions such as disallowing polymorphic recursion.
This paper instead combines Go's runtime infrastructure for type
conversions with statically generated adapter code, incurring more
runtime conversions but supporting all of WG, including proposed
features like generic methods, while lifting restrictions and
maintaining separate compilation compatibility.

\citet{DBLP:conf/popl/Leroy92} presented an approach of boxing and unboxing
generic function arguments and results for ML, with careful treatment of
higher order functions.
It does not consider subtyping and the language disallows polymorphic
recursion.
By contrast, our work tackles the problem of compiling generic code in the
presence of structural subtyping between non-uniform representations.

In future work, we will consider techniques to optimise our approach, such
as combining our adaptors with localised monomorphisation of package private
code, and statically safe elimination of type
conversions~\cite{DBLP:conf/popl/Leroy92}.
The overhead of runtime assembly of runtime type-reps can be optimised by
pre-computing and caching maps from open types to their reps when a generic
type or method is instantiated~\cite{DBLP:conf/oopsla/ViroliN00}.

\begin{acks}
  We would like to thank Ian Lance Taylor for valuable discussions,
  and the anonymous reviewers for their feedback.
  This work was supported by national funds through Fundação para a Ciência
  e a Tecnologia, I.P. (FCT) under projects UID/50021/2025 and
  UID/PRR/50021/2025.
\end{acks}

\section*{Data-Availability Statement}

We have implemented a minimal prototype of WG and LWG in Go as interpreters,
and our compilation approach as a translation from WG to LWG.  The
interpreters perform type checking of programs and reduction of their main
expressions.  Our software artifact is available on Zenodo~\cite{wg-artifact}.
It also includes the examples from this paper as tests.

\bibliography{main}

\iflong
\section{Examples}

\subsection{Desugared version of Fig.~\ref{fig:WGexample}}
\label{app:mynum}
\begin{golst}
package main
type MyNum interface { MyInt | MyFloat }
type MyInt int
func (x MyInt) String() string { return strconv.Itoa(int(x)) }
type MyFloat float64
func (x MyFloat) String() string { return strconv.FormatFloat(float64(x), 'E', -1, 64) }
type BiFunc[a Any, b Any, c Any] interface { apply(x a, y b) c }
type List[a any] interface { FoldL[b any](f BiFunc[b, a, b], z b) b }
type Nil[a any] struct {}
type Cons[a any] struct { head a; tail List[a] }
func (x Nil[a]) FoldL[b any](f BiFunc[b, a, b], z b) b { return z }
func (x Cons[a]) FoldL[b any](f BiFunc[b, a, b], z b) b { return x.tail.FoldL[b](f, f(z, x.head)) }
type Add[a MyNum] struct {}
func (x Add[a]) apply(y1 string, y2 a) string { return y1 + ", " + y2.String() }
type Dummy struct {}
func (x Dummy) join[a MyNum](x List[a]) String { return x.FoldL[string](Add[a]{}, "") }
func main() {
    fmt.Println(Dummy{}.join[MyFloat](
        Cons[MyFloat]{MyFloat(1.0), Cons[MyFloat]{MyFloat(2.0), Nil[MyFloat]{}}}
    ))
}
\end{golst}
 
\subsection{Adaptors for Example~\ref{ex:duck-typing-adapters}}
\label{app:adaptors}

Revisiting Example~\ref{ex:duck-typing-adapters}, the adaptors for
\goinl{process} of \goinl{IntProcessor} and \goinl{GenericProcessor}
are:
\begin{golst}
func (p IntProcessor) process$_D$(input int@Any) int@Any {
        return p#IntProcessor.process(input.(int@int)).(int ->$_\emptyset$ int@any) }
func (p GenericProcessor[T any]) process$_D$(input T@Any) T@Any {
        return p#GenericProcessor.process(input.(T -> T@any)) }
\end{golst}
The \goinl{IntProcessor} adaptor calls \goinl{process} by first unboxing the argument
and boxing the result as needed. The \goinl{GenericProcessor} adaptor passes the
already boxed argument to \goinl{process} and returns the already boxed
result.

\section{Additional Definitions and Rules for WG}

\begin{Figure}
  \begin{ruled}
    \begin{mathpar}
      \inferrule
      {\,}
      {\ubase(B)}

      \inferrule
      {(\type~t[\ov{\alpha~S}]~T) \in \ov{D}
        \and
        \ubase(T)
      }
      {\ubase(t[\ov{R}])}
      
      \inferrule
      {\,}
      {\ustruct(\struct~\br{\ov{f~{\Phi}}})}      

      \inferrule
      {(\type~t[\ov{\alpha~S}]~T) \in \ov{D}
        \and
        \ustruct(T)
      }
      {\ustruct(t[\ov{R}])}      

      \inferrule
      {\,}
      {\uinterface(\alpha)}

      \inferrule
      {\,}
      {\uinterface(\interface~\br{\ov{F}})}

      \inferrule
      {(\type~t[\ov{\alpha~S}]~T) \in \ov{D}
        \and
        \uinterface(T)
      }
      {\uinterface(t[\ov{R}])}

      \inferrule
      {c \in \mathbb{Z}}
      {\ctype(c) = \mathtt{int}}
      
      \inferrule
      {c \in \{ \mathtt{true},  \mathtt{false}\}}
      {\ctype(c) = \mathtt{bool}}
      
      \inferrule
      {\ustruct(T)}
      {\vtype(T\br{\ov{v}}) = T}     
      
      \inferrule
      {\ubase(T)}
      {\vtype(T(  c)) = T}
      
      \inferrule
      {\text{$mM_1, mM_2 \in \ov{mM}$ implies $M_1 = M_2$}}
      {\unique(\ov{mM})}

      \inferrule
      {}
      {\tdecls(\ov{D}) = \lst{t \mid (\type~t[\ov{\alpha~S}]~T) \in \ov{D}}}
      
      \inferrule
      {}
      {\mdecls(\ov{D}) = \lst{t.m \mid (\func~(x~t[\ov{\alpha~S}])~mM~\br{\return~e}) \in \ov{D}}}           
      
      \inferrule
      {}
      {\sigdecls(\ov{F}) = \lst{m \mid mM  \in \ov{F}}}
      
      \inferrule
      { \, }
      {\maintype{B} = B}

      \inferrule
      { \, }
      {\maintype{\alpha} = \alpha}
      
      \inferrule
      { \, }
      {\maintype{t[\ov{R}]} = t[\ov{\maintype{R}}]}
        
      \inferrule
      { \, }
      {\maintype{\struct~\br{\ov{f~{\typair{T}{V}}}}} =
        \struct~\br{\ov{f~\maintype{T}}}
      }
      
      \inferrule
      { \, }
      {\maintype{\interface~\br{\ov{F}}} =
        \interface~\br{\ov{\maintype{F}}}
      }
    \end{mathpar}
  \end{ruled}
  \caption{WG: additional auxiliary definitions}\label{fig:wg-aux-defs-app}
\end{Figure}

\begin{Figure}
  \begin{ruled}
    \begin{mathpar}      
      \inferrule[t-b]
      {\, }
      { \Delta \vdash B \ok }

      \inferrule[t-param]
      { (\alpha : T) \in \Delta }
      { \Delta \vdash \alpha \ok }
      
      \inferrule[t-struct]
      {
        \distinct(\ov{f}) \\
        \Delta  \vdash \ov{\Phi} \supernice
      }
      {
        \Delta \vdash \struct~\br{\ov{f~\Phi}} \ok
      }
      
      \inferrule[t-interface]
      {
        \Delta \vdash    \ov{F} \okelmt 
        \and
        \distinct(\sigdecls(\ov{F}))
        \and
        \unique(\methods_\Delta(\interface~\br{ \ov{F} }))
      }
      { \Delta \vdash \interface~\br{ \ov{F} } \ok }       
    \end{mathpar}
    
    Well-formed F
    \hfill \fbox{$\Delta \vdash F \okelmt$}
    \begin{mathpar}
      \inferrule[t-union]
      {\Delta \vdash C_i \okelmt \and C_i \in E}
      { \Delta \vdash E \okelmt}      
      
      \inferrule[t-method]
      {\Delta \vdash mM\ok}
      {\Delta \vdash mM\okelmt}
    \end{mathpar}
    
    Well-formed type element   \hfill \fbox{$\Delta \vdash C \okelmt$}
    \begin{mathpar}      
      \inferrule[t-element]
      {
        \Delta \vdash T \ok
        \\
        \neg\utyvar_\Delta(T)
        \\
        \uinterface(T) \implies \methods_\Delta(T) = \emptyset}
      {\Delta \vdash T \okelmt}

      \inferrule[t-approx]
      {\Delta \vdash T \ok
        \\
        \underlying_\Delta(T)=T
        \\
        \neg\uinterface(T) 
      }
      {\Delta \vdash \ttilde T \okelmt}
    \end{mathpar}

    Well-formed type formals
    \hfill \fbox{$\Delta \vdash \ov{\beta~T} \ok$}
    \begin{mathpar}
      \inferrule[t-formal]
      {
        \uinterface(\ov{S})
        \\
        \neg\utyvar_\Delta(\ov{S})
        \\
        \Delta = \ov{\alpha : T}
        \\
        \distinct(\ov{\alpha},\ov{\beta}) \\
        \Delta , \ov{\beta:S} \vdash \ov{S} \ok
      }
      { \Delta \vdash \ov{\beta~S} \ok}
    \end{mathpar}

    Well-formed method specifications
    \hfill \fbox{$\Delta  \vdash mM \ok$} \qquad \fbox{$\Delta  \vdash T \ok$}
    \begin{mathpar}
      \inferrule[t-specification]
      {
        \Delta \vdash \ov{\beta~S} \ok\\
        \distinct(\ov{x}) \\
        \Delta' = \Delta, \ov{\beta: S}\\
        \Delta' \vdash \ov{ \Phi} \supernice \\
        \Delta' \vdash \Psi \supernice  
      }
      { \Delta \vdash m[\ov{\beta~S}](\ov{x~\Phi})~\Psi \ok }
    \end{mathpar}
  \end{ruled}
  \caption{WG: well-formedness}\label{fig:wg-types-wf-app}\label{fig:wg-iface-types-wf}
\end{Figure}

\begin{Figure}
  \begin{ruled}
    \begin{mathpar}
      \\
      \inferrule[t-assert$_I$]
      {
        \uinterface(T)
        \\
        \Delta \vdash T \nice
        \\
        \Delta \stoup \Gamma \vdash e : S
        \\
        \uinterface(S)
      }
      { \Delta \stoup \Gamma \vdash e.(T) : T }

      \squeeze
      \inferrule[t-assert$_S$]
      {
        \neg\uinterface(T)
        \\
        \Delta \vdash T \nice
        \\
        \Delta \stoup \Gamma \vdash e : S
        \\
        \uinterface(S)
        \\
        T \imp_\Delta \bounds_\Delta(S)
      }
      { \Delta \stoup \Gamma \vdash e.(T) : T }
      \squeeze
      
      \fbox{
        \inferrule[t-stupid]
        {
          \Delta \vdash T \nice
          \\
          \Delta \stoup \Gamma \vdash e : S
          \\
          \neg\uinterface(S)
        }
        { \Delta \stoup \Gamma \vdash e.(T) : T }
      }
      \\
      \inferrule[t-const]
      {
        \Delta  \vdash T \nice
        \\
        \underlying_\Delta(T) = 
        \ctype(c)
      }       
      {\Delta \stoup  \Gamma \vdash T(c) : T}    
      
      \inferrule[t-conv-a]
      {  \Delta  \vdash T \nice
        \and
        \Delta \stoup \Gamma \vdash e : U
        \and
        U \assg_\Delta T
      }       
      {\Delta \stoup \Gamma \vdash T(e) : T}

      \inferrule[t-conv-u]
      { \Delta  \vdash T \nice
        \and
        \Delta \stoup \Gamma \vdash e : U
        \and
        \neg\utyvar_\Delta(U)
        \and
        \neg\utyvar_\Delta(T) 
        \and
        \underlying_\Delta(U) = \underlying_\Delta(T)
      }       
      {\Delta \stoup  \Gamma \vdash T(e) : T}
    \end{mathpar}
  \end{ruled}
  \caption{WG: additional typing rules for expressions}\label{fig:wg-typing-expr-app}
\end{Figure}

Figures~\ref{fig:wg-aux-defs-app} ~\ref{fig:wg-iface-types-wf},
and~\ref{fig:wg-typing-expr-app} give the omitted definitions and
rules of WG.

\section{Proofs}\label{app:proofs}

\subsection{WG Type Safety}
\label{app:wg}

\begin{lemma}[Weakening]
  \label{lem:wgg_weak}
Let $\alpha \not\in \mathit{dom}(\Delta)$ and $x \not\in \mathit{dom}(\Gamma)$:
  \begin{enumerate}
  \item If $\Delta \vdash T \ok$ then $\Delta , \alpha : S \vdash T \ok$
   \item If $\Delta \vdash T \nice$ then $\Delta , \alpha : S \vdash
     T \nice$
  \item If $\Delta \vdash \ov{\beta \, T} \ok$ then $\Delta  , \alpha
    : S \vdash \ov{\beta \,
      T} \ok$
   \item If $\Delta \vdash m[\ov{\beta~T}](\ov{x~U})~U \ok$ then
     $\Delta , \alpha : S \vdash m[\ov{\beta~T}](\ov{x~U})~U \ok$
    \item If $\Delta ; \Gamma \vdash e : T$ then $\Delta , \alpha : S
      ; \Gamma \vdash e : T$ and $\Delta ; \Gamma , x : S \vdash e :
      T$
     \item If $\Delta \vdash T \ok$ and $\Delta \vdash U \ok$ and $T \imp_{\Delta} U$ implies $T \imp_{\Delta, \alpha : S}
       T'$
      \item  If $\Delta \vdash T \ok$ and $\Delta \vdash U \ok$ and $T \imp_{\Delta} U$ implies $T \prec_{\Delta, \alpha : S}
       T'$
  \end{enumerate}
  
\end{lemma}

\begin{proof}
Each statement above is proved by a straightforward induction on the
given derivation. 
\end{proof}

\begin{lemma}\label{lem:subt_refltrans}
  If $\Delta \vdash T \ok$ and $\Delta \vdash S \ok$ then:
  \begin{enumerate}
  \item $T \imp_\Delta T$
  \item $T \imp_\Delta S$ and $S \imp_\Delta U$ implies $T \imp_\Delta
    U$
  \item $T \prec_\Delta T$

  \end{enumerate}

\end{lemma}

\begin{proof}
Statements (1) and (3) above follow straightforwardly from the
definitions of implements and assignable, respectively, with
reflexivity of assignability relying on reflexivity of implements.

Statement (2) follows by induction on the given derivations:

\begin{description}
\item[Case:] $\imp_\alpha$ and $\imp_T$ -- Impossible due to
  $\neg\uinterface_\Delta(\alpha)$.
\item[Case:] $\imp_\alpha$ and $\imp_I$ -- Immediate by $\imp_I$.

\item[Case:] $\imp_T$ and $\imp_\alpha$ -- Impossible.
 
\item[Case:] $\imp_T$ and $\imp_I$ -- Immediate by $\imp_I$.

\item[Case:] $\imp_I$ and $\imp_\alpha$ -- Impossible.

\item[Case:] $\imp_I$ and $\imp_T$ -- Impossible.

\item[Case:] $\imp_I$ and $\imp_I$

  \begin{tabbing}
           $\types_\Delta(S) \subseteq \types_\Delta(T)$\\
           $\methods_\Delta(S) \supseteq \methods_\Delta(T)$\\
          
          $\uinterface_\Delta(T)$\\
          
          $\neg\utyvar_\Delta(T) $ \` by inversion \\
           $\types_\Delta(T) \subseteq \types_\Delta(U)$\\
           $\methods_\Delta(T) \supseteq \methods_\Delta(U)$\\
          
          $\uinterface_\Delta(U)$\\
          
          $\neg\utyvar_\Delta(U) $ \` by inversion \\
          $\types_\Delta(S) \subseteq \types_\Delta (U)$ \` by
          transitivity of $\subseteq$\\
          $\methods_\Delta(S) \supseteq \methods_\Delta(U)$ \` by
          transitivity of $\supseteq$\\
          $S \imp_\Delta U$ \` by $\imp_I$
          
  \end{tabbing}
\end{description}
\end{proof}

\begin{lemma}[Quasi-transitivity]\label{lem:quasitrans}
  Let $\Delta \vdash T, T',S\ok$. If $T \imp_\Delta T'$ and $T' \assg_\Delta S$ then $T
  \assg_\Delta S$. 
\end{lemma}
\begin{proof}
By case analysis on the given derivations.

\begin{description}

\item[Case:]$ alpha \imp alpha$
  \begin{tabbing}
    {\bf Subcase:} $\alpha \assg S$ from $\alpha \imp S$\\
    Follows by transitivity of $\imp$\\
    No other possible derivation of $\assg$\\
  \end{tabbing}

  \item[Case:]$ T \imp T$ and $\neg\uinterface(T)$
  \begin{tabbing}
    Immediate since $T = T'$
  \end{tabbing}

  \item[Case:]$ T \imp T'$ from $\uinterface(T')$, $\types(T)
    \subseteq \types(T')$, $\methods(T) \supseteq \methods(T')$, $\neg\utyvar(T')$
  \begin{tabbing}
    {\bf Subcase:} $T' \assg S$ from $T \imp S$\\
    Follows by transitivity of $\imp$\\
       {\bf Subcase:} $T' \assg S$ from $\underlying(T') = S$\\
       Since $\uinterface(T')$ then $S = \uinterface{\dots}$\\
       $\methods(S) = \methods(T')$ and $\types(S) = \types(T')$
       therefore\\
       $T\imp S$ and so $T\assg S$.\\
       {\bf Subcase:} $T' = L$ and $L \assg t[\dots]$ with $L =
       \underlying(t[\dots])$\\
       Since $T' = L$ then $\uinterface(L)$ and so $\uinterface(t[\dots])$\\ and 
so $methods(L) = methods(t[\dots])$ and $types(L) = types(t[\dots])$ and\\
therefore $T \imp t[\dots]$ and so $T \assg t[\dots]$
  \end{tabbing}

\end{description}

 \end{proof}

 \begin{lemma}
If $T \imp S$ then $\types(T) \subseteq \types(S)$ and $\methods(T)
\supseteq \methods(S)$.
\end{lemma}
\begin{proof}
Straightforward by definition.
\end{proof}

\begin{lemma}[Type Substitution preserves Subtyping]\label{lem:wgsubstsub}
If $T \imp_{\Delta_0, \ov{\alpha : U} , \Delta_1 } S$ , $\Delta_0
\vdash \ov{U} \ok$ and $\ov{U_0} \imp_{\Delta_0} \ov{U}[\ov{\alpha} :=
\ov{U_0}]$ then $T [\ov{\alpha} :=
\ov{U_0}] \imp_{\Delta_0,\Delta_1 [\ov{\alpha} :=
\ov{U_0}]} S [\ov{\alpha} :=
\ov{U_0}]$.
\end{lemma}
\begin{proof}
By induction on the derivation of $T \imp_{\Delta_0, \ov{\alpha : U} ,
  \Delta_1 } S$.

When the derivation is by $\imp_\alpha$ or $\imp_{T}$ the result
follows immediately. When the derivation is by $\imp_I$, we have that
$S$ is a non-type variable interface type, $\types(T) \subseteq
\types(S)$ and $\methods(T) \supseteq \methods(S)$ and both properties
are preserved under substitution.
\end{proof}

\begin{lemma}[Type Substitution preserves Assignability]
\label{lem:wgsubstassign}
If $T \prec_{\Delta_0, \ov{\alpha : U} , \Delta_1 }: S$ , $\Delta_0
\vdash \ov{U} \ok$ and $\ov{U_0} \prec_{\Delta_0} \ov{U}[\ov{\alpha} :=
\ov{U_0}]$ then $T [\ov{\alpha} :=
\ov{U_0}] \prec_{\Delta_0,\Delta_1: [\ov{\alpha} :=
\ov{U_0}]} S [\ov{\alpha} :=
\ov{U_0}]$.
  \end{lemma}
  \begin{proof}
If assignability is derived by rule $\prec_{\imp}$ then the result
follows by Lemma~\ref{lem:wgsubstsub}; if derived by rule $\prec_{tL}$
or  $\prec_{Lt}$ then the result is straightforward.
\end{proof}

\begin{lemma}[Bounds commutes with type substitution]\label{lem:wgbounds}
If $\Delta_0 , \ov{\alpha : T} , \Delta_1 \vdash S \ok$ and $\ov{U}
  \imp_{\Delta_0} \ov{T}[\ov{\alpha} := \ov{U}]$ with $\Delta_0 \vdash
  \ov{U} \ok$ then $\bounds_{\Delta_0, \Delta_1 [\ov{\alpha} :=
    \ov{U}]}(S [\ov{\alpha} := \ov{U}] ) \imp_{\Delta_0,\Delta_1
    [\ov{\alpha} := \ov{U}]} ( \bounds_{\Delta_0 , \ov{\alpha : T} , \Delta_1}(S)) [\ov{\alpha} := \ov{U}]$
\end{lemma}
\begin{proof}
  By case analysis on the implements relation and the derivation of $S \ok$.
\end{proof}

\begin{lemma}[Type substitution preserves well-formedness (1)]\label{lem:wgsubstwf}
If $\Delta_0 , \ov{\alpha : T} , \Delta_1 \vdash S \ok$ and $\ov{U}
\imp_{\Delta_0} T[\ov{\alpha} := \ov{U}] $ with $\Delta_0 \vdash
\ov{U}\nice$ then $\Delta_0, \Delta_1 [\ov{\alpha} := \ov{U}] \vdash S[\ov{\alpha} := \ov{U}]\ok$
\end{lemma}
\begin{proof}
Straightforward induction on the derivation of $\Delta_0 , \ov{\alpha
  : T} , \Delta_1 \vdash S \ok$, appealing to
Lemmas~\ref{lem:wgsubstsub} and~\ref{lem:wgbounds}.
\end{proof}
\begin{lemma}[Type substitution preserves well-formedness (2)]\label{lem:wgsubstnice}
If $\Delta_0 , \ov{\alpha : T} , \Delta_1 \vdash S \nice$ and $\ov{U}
\imp_{\Delta_0} T[\ov{\alpha} := \ov{U}] $ with $\Delta_0 \vdash
\ov{U}\nice$ then $\Delta_0, \Delta_1 [\ov{\alpha} := \ov{U}] \vdash S[\ov{\alpha} := \ov{U}]\nice$
\end{lemma}
\begin{proof}
By Lemma~\ref{lem:wgsubstwf}
\end{proof}

\begin{lemma}[Type Substitution preserves typing]\label{lem:wgsubsttype}
  If $\Delta_0 , \ov{\alpha : T} , \Delta_1 ; \Gamma \vdash e : S$ and
  $\ov{U} \imp_{\Delta_0} T[\ov{\alpha} := \ov{U}]$ with $\Delta_0
  \vdash \ov{U} \nice$ then $\Delta_0, \Delta_1[\ov{\alpha} := \ov{U}]  ; \Gamma[\ov{\alpha} := \ov{U}]
  \vdash e[\ov{\alpha} := \ov{U}] : S[\ov{\alpha} := \ov{U}]$
 \end{lemma}
\begin{proof}
  By induction on typing. Most cases follow immediately by the
  induction hypothesis and the corresponding typing rule.
Let $\eta = (\ov{\alpha} := \ov{U})$ and $\Delta_0 , \ov{\alpha : T} , \Delta_1 =\Delta$:
  \begin{description}
  \item[Case:] {\sc t-field} ($e = e.f_i$)

    \begin{tabbing}
$\Delta_0 , \ov{\alpha : T} , \Delta_1 ; \Gamma \vdash e : S'$ \quad $\ustruct(S')$\quad
$(f\,\ov{\Phi}) = \fields_{\Delta_0 , \ov{\alpha : T} , \Delta_1}(S')$
\` by inversion\\
$(f\,\ov{\Phi}[\eta]) = \fields_{\Delta_0 ,\Delta_1[\eta]}(S'[\eta])$
\` by definition\\
$\Delta_0 , \Delta_1[\eta] ; \Gamma[\eta] \vdash e[\eta] : S'[\eta]$
\` by i.h.\\
$\Delta_0 , \Delta_1[\eta] ; \Gamma[\eta] \vdash e[\eta].f_i : \Phi_i [\eta]\downharpoonright_1$
\` by {\sc t-field}
\end{tabbing}

\item[Case:] {\sc t-call} ($e = e.m[S'](\ov{e})$\quad $\Delta_0 ,
  \ov{\alpha : T} , \Delta_1 = \Delta$)

  \begin{tabbing}
    $\Delta_0 , \ov{\alpha : T} , \Delta_1 ; \Gamma \vdash e : R$ \quad
    $\Delta \vdash \ov{S'}\nice$ \quad
    $m[\ov{\alpha'\,S}][\ov{x\,\Phi}]\,\Psi \in \methods(R)$\\
    $\Delta ; \Gamma \vdash \ov{e : T'}$ \quad $\eta' = (\ov{\alpha'
      := S'})$ \quad $\ov{\alpha'[\eta'] \imp_\Delta S[\eta']}$ \quad
    $\ov{T'[\eta'] \assg_\Delta \Phi\downharpoonright_1 [\eta']}$ \`
    by inversion\\
    $\Delta_0,\Delta_1[\eta] ; \Gamma[\eta] \vdash e[\eta] : R[\eta]$ \` by
    i.h.\\
    $\Delta_0,\Delta_1[\eta] ; \Gamma[\eta] \vdash \ov{e[\eta] :
      T'[\eta]}$ \` by i.h.\\
        $m[\ov{\alpha'\,S[\eta]}][\ov{x\,\Phi[\eta]}]\,\Psi[\eta] \in
        \methods(R[\eta])$ \` by def.\\
 $\Delta_0 , \Delta_1[\eta] \vdash \ov{S'[\eta]}\nice$ \` by
 Lemma~\ref{lem:wgsubstnice}\\
 $\ov{\alpha'[\eta'][\eta] \imp_{\Delta_0,\Delta_1[\eta]}
   S[\eta'][\eta]}$ \` by Lemma~\ref{lem:wgsubstsub}\\
 $\ov{T'[\eta'][\eta] \assg_{\Delta_0,\Delta_1[\eta]}
   \Phi\downharpoonright_1 [\eta'][\eta]}$
 \` by Lemma~\ref{lem:wgsubstassign}\\
 $\Delta_0, \Delta_1[\eta] ; \Gamma[\eta] \vdash
 e[\eta].m[\ov{S'[\eta]}](\ov{e[\eta]}) : \Psi\downharpoonright_1
 [\eta'][\eta]$ \` by {\sc t-call}
  \end{tabbing}

\item[Case:] {\sc t-op-param} $\Delta_0 , \ov{\alpha : T} , \Delta_1 =\Delta$

  \begin{tabbing}
$\Delta_0 , \ov{\alpha : T} , \Delta_1 ; \Gamma \vdash \ov{e} :
\beta$ \quad $\beta : T' \in \Delta_0 , \ov{\alpha : T} , \Delta_1$
\quad
$\{ \underlying_\Delta(V) \mid V \in \types(T')\} \subseteq
\opdom(\circ)$\\ \` by inversion\\
$\Delta_0 , \Delta_1 [\eta] ; \Gamma[\eta] \vdash \ov{e[\eta]} :
\beta[\eta]$ \` by i.h.\\
if $\beta \not\in \ov{\alpha}$ then $\Delta_0,\Delta_1[\eta] \vdash
\circ(\ov{e[\eta]}) : \beta$ \` by {\sc t-op-param}\\
if $\beta \in \ov{\alpha}$ then $\Delta_0 , \Delta_1 [\eta] ; \Gamma[\eta] \vdash \ov{e[\eta]} :
U_i$ with $U_i \imp T'$\\
$\Delta_0,\Delta_1[\eta] \vdash \circ(\ov{e[\eta]}) : U_i$ \` by {\sc t-op}
  \end{tabbing}

\item[Case:] {\sc t-assert}$_I$

  \begin{tabbing}
$\uinterface(T)$ \quad $\Delta \vdash T\nice$ \quad $\Delta ; \Gamma
\vdash e : S$ \quad $\uinterface(S)$ \` by inversion\\
$\Delta_0 ; \Delta_1[\eta] ; \Gamma[\eta] \vdash e[\eta] : S[\eta]$ \` by
i.h.\\
$\Delta_0 , \Delta_1[\eta] \vdash T[\eta]\nice$ \` by
Lemma~\ref{lem:wgsubstnice}\\
if $T \neq \alpha_i$ and $S = \alpha_j$ and $\neg\uinterface(S[\eta])$
$\Delta_0,\Delta_1[\eta];\Gamma[\eta] \vdash e[\eta].(T[\eta]) :
T[\eta]$ \\\` by {\sc t-stupid}\\
if $T \neq \alpha_i$ and $S = \alpha_j$ and $\uinterface(S[\eta])$
then
$\Delta_0,\Delta_1[\eta];\Gamma[\eta] \vdash e[\eta].(T[\eta]) :
T[\eta]$ \\\` by {\sc t-assert}$_I$\\
if $T \neq \alpha_i$ and $S \neq \alpha_j$ 
then
$\Delta_0,\Delta_1[\eta];\Gamma[\eta] \vdash e[\eta].(T[\eta]) :
T[\eta]$ \\\` by {\sc t-assert}$_I$\\

if $T = \alpha_i$ and $S \neq \alpha_j$ and $\uinterface(T[\eta])$
$\Delta_0,\Delta_1[\eta];\Gamma[\eta] \vdash e[\eta].(T[\eta]) :
T[\eta]$ \\\` by {\sc t-assert}$_I$\\

if $T = \alpha_i$ and $S \neq \alpha_j$ and $\neg\uinterface(T[\eta])$
$\Delta_0,\Delta_1[\eta];\Gamma[\eta] \vdash e[\eta].(T[\eta]) :
T[\eta]$ \\\` by {\sc t-stupid} or {\sc t-assert-s}\\

if $T = \alpha_i$ and $S = \alpha_j$ and $\uinterface(T[\eta])$
and $\uinterface(S[\eta])$ then
$\Delta_0,\Delta_1[\eta];\Gamma[\eta] \vdash e[\eta].(T[\eta]) :
T[\eta]$ \\\` by {\sc t-assert-i}\\

if $T = \alpha_i$ and $S = \alpha_j$ and $\neg\uinterface(T[\eta])$
and $\uinterface(S[\eta])$ then
$\Delta_0,\Delta_1[\eta];\Gamma[\eta] \vdash e[\eta].(T[\eta]) :
T[\eta]$ \\\` by {\sc t-assert-s}\\

if $T = \alpha_i$ and $S = \alpha_j$ and $\uinterface(T[\eta])$
and $\neg\uinterface(S[\eta])$ then
$\Delta_0,\Delta_1[\eta];\Gamma[\eta] \vdash e[\eta].(T[\eta]) :
T[\eta]$ \\\` by {\sc t-stupid}\\

if $T = \alpha_i$ and $S = \alpha_j$ and $\neg\uinterface(T[\eta])$
and $\neg\uinterface(S[\eta])$ then
$\Delta_0,\Delta_1[\eta];\Gamma[\eta] \vdash e[\eta].(T[\eta]) :
T[\eta]$ \\\` by {\sc t-stupid}\\

  \end{tabbing}

\item[Case:] {\sc t-assert}$_S$
  \begin{tabbing}
$\neg\uinterface(T)$ \quad $\Delta \vdash T\nice$ \quad $\Delta ;
\Gamma \vdash e : S$ \quad $\uinterface(S)$ \quad
$T\imp_\Delta\bounds_\Delta(S)$ \` by inversion\\
$\Delta_0,\Delta_1 [\eta];\Gamma[\eta]\vdash e :S[\eta]$ \` by i.h.\\
if $S \neq \alpha_i$ then $\Delta_0,\Delta_1 [\eta];\Gamma[\eta]\vdash
e[\eta].(T[\eta]) : T[\eta]$ \` by  {\sc t-assert}$_S$ and
Lemma~\ref{lem:wgbounds}\\
if $S = \alpha_i$ then $\Delta_0,\Delta_1 [\eta];\Gamma[\eta]\vdash
e[\eta].(T[\eta]) : T[\eta]$ \` by  {\sc t-assert}$_S$ or {\sc t-stupid}\\
\end{tabbing}

  \end{description}
\end{proof}

\begin{lemma}[Substitution preserves Typing]\label{lem:wgsubst}
  If $\Delta ; \Gamma , \ov{x : T} \vdash e : S$ and $\Delta ;
  \Gamma \vdash \ov{e : T}$ then $\Delta ; \Gamma \vdash e[\ov{x} :=
  \ov{e}] : S$.
\end{lemma}
\begin{proof}
By induction on typing of $\Delta ; \Gamma , \ov{x : T} \vdash e :
S$. All cases are standard.
\end{proof}

\begin{lemma}\label{lem:wgbody}
If  $(m[\ov{\alpha~S}](\ov{x~T})~T) \in \methods_\Delta(U)$ and
$\mbody(U.m[\ov{S'}]) = (x , \ov{x}).e_0$ with
$\Delta \vdash U \ok$,
$\Delta \vdash \ov{S'} \nice$,
$\ov{S'}\imp_\Delta \ov{S[\ov{\alpha} := \ov{S'}]}$ then there exists
$T'$ such that $\Delta \vdash T' \ok$, $T'\imp_\Delta T[\ov{\alpha} :=
\ov{S'}]$ and 
$\Delta ; x : U , \ov{x : T[\ov{\alpha} := \ov{S'}]} \vdash e_0 : T'$
\end{lemma}
\begin{proof}
By inversion on $\mbody$ and well-formedness, followed by
Lemmas~\ref{lem:wgsubstsub} and~\ref{lem:wgsubsttype}.
\end{proof}

\begin{theorem}[Type Preservation]\label{thm:wgtpres}
If $\ROUNDTWO{\Delta ; \Gamma} \vdash e : T$ and $e \longrightarrow e'$
then $\ROUNDTWO{\Delta ; \Gamma}  \vdash e' : S$, for some $S$ such that $S
\imp_\ROUNDTWO{\Delta} T$
\end{theorem}
\begin{proof}
By induction on the reduction relation $e \longrightarrow e'$.
We show illustrative cases.

\begin{description}
\item[Case:] {\sc r-assert}
  \begin{tabbing}
    $v.(T) \becomes v$ with $\vtype(v) \imp T$ \` this case\\
    {\bf Subcase:} $\vdash v.(T) : T$ by {\sc t-assert}$_S$\\
    $\ROUNDTWO{\Delta ; \Gamma}  \vdash v : T'$ \quad $\uinterface_\ROUNDTWO{\Delta}  (T')$ \quad $T\imp_\ROUNDTWO{\Delta }  \bounds(T')$
    \` by inversion\\
    $\vtype(v) = T' \imp T$ \\
    {\bf Subcase:} $\ROUNDTWO{\Delta ; \Gamma}  \vdash v.(T) : T$ by {\sc t-assert}$_I$\\
    $\ROUNDTWO{\Delta ; \Gamma} \vdash v : T'$ \quad $\uinterface(T')$ 
    \` by inversion\\
    $\vtype(v) = T' \imp T$ \\
    {\bf Subcase:} $\ROUNDTWO{\Delta ; \Gamma} \vdash v.(T) : T$ by {\sc t-stupid}\\
    $\ROUNDTWO{\Delta ; \Gamma} \vdash v : T'$ 
    \` by inversion\\
    $\vtype(v) = T' \imp_\ROUNDTWO{\Delta }  T$ \\
  \end{tabbing}

 \item[Case:] $e.(T) \becomes e'.(T)$

   \begin{tabbing}
     By inversion on typing we have 3 subcases.\\
     {\bf Subcase:} $\ROUNDTWO{\Delta ; \Gamma} \vdash e.(T) : T$ by {\sc t-assert}$_I$\\
     $\ROUNDTWO{\Delta ; \Gamma} \vdash e : T'$ \` by inversion\\
     $\ROUNDTWO{\Delta ; \Gamma} \vdash e' : S'$ with $S'\imp_\ROUNDTWO{\Delta }  T'$ \` by i.h.\\
     if $\uinterface_\ROUNDTWO{\Delta }  (S')$ then $\ROUNDTWO{\Delta ; \Gamma}  \vdash e'.(T) : T$ by {\sc t-assert}$_I$\\
     if $\neg\uinterface_\ROUNDTWO{\Delta }  (S')$ then $\ROUNDTWO{\Delta ; \Gamma}  \vdash e'.(T) : T$ by {\sc t-stupid}\\
     {\bf Subcase:} $\ROUNDTWO{\Delta ; \Gamma} \vdash e.(T) : T$ by {\sc t-assert}$_S$\\
     $\ROUNDTWO{\Delta ; \Gamma} \vdash e : T'$ \quad $T\imp_\ROUNDTWO{\Delta }  T'$  \` by inversion\\
     $\ROUNDTWO{\Delta ; \Gamma} \vdash e' : S'$ with $S'\imp_\ROUNDTWO{\Delta } T'$ \` by i.h.\\
     if $T \imp_\ROUNDTWO{\Delta}  S'$ then $e'.(T) : T$ by {\sc t-assert}$_S$\\
     if $T \notimp_\ROUNDTWO{\Delta }  S'$ then $e'.(T) : T$ by {\sc t-stupid}\\
     if $\uinterface_\ROUNDTWO{\Delta }  (S')$ then $S' = T'$ by value restriction and so
     $T\imp_\ROUNDTWO{\Delta }  S'$\\
     thus, $e'.(T) : T$ by {\sc t-assert}$_S$\\
     if $\neg\uinterface_\ROUNDTWO{\Delta }  (S')$ then
     $\ROUNDTWO{\Delta ; \Gamma}  \vdash e'.(T) : T$ by {\sc t-stupid}\\
     {\bf Subcase:} $\ROUNDTWO{\Delta ; \Gamma} \vdash e.(T) : T$ by {\sc t-stupid}\\
     $\ROUNDTWO{\Delta ; \Gamma} \vdash e' : S'$ with $S'\imp_\ROUNDTWO{\Delta }  T'$ \` by i.h.\\
     $\ROUNDTWO{\Delta ; \Gamma} \vdash e'.(T) : T$ by {\sc t-stupid}\\
   \end{tabbing}

 \item[Case:] {\sc r-convert-b}
   \begin{tabbing}
     $T(U(c)) \becomes T(c)$ with $\ubase(T)$ \` this case\\
     $\ROUNDTWO{\Delta ; \Gamma} \vdash T(U(c)) : T$ \` by inversion on {\sc t-conv-a} or {\sc t-conv-u}\\
     $\ROUNDTWO{\Delta ; \Gamma} \vdash U(c) : U$ \` by {\sc t-const} with $\underlying_\ROUNDTWO{\Delta }  (U) =
     \ctype(c)$\\
     $\ROUNDTWO{\Delta ; \Gamma} \vdash T(c) : T$ \\\` by {\sc t-const} since
     $\underlying_\ROUNDTWO{\Delta }  (T)=\ctype(c)$ by {\sc t-conv-u}
     \\ \` and $U \assg_\ROUNDTWO{\Delta }  T$
     implies $\underlying_\ROUNDTWO{\Delta }  (T) = \underlying_\ROUNDTWO{\Delta }  (U)$
   \end{tabbing}

 \item[Case:] {\sc r-convert-s}
   \begin{tabbing}
     $T(U\{\ov{v}\}) \becomes T\{\ov{v}\}$ with $\ustruct_\ROUNDTWO{\Delta } (T)$\\
     $\ROUNDTWO{\Delta ; \Gamma} \vdash T(U\{\ov{v}\}) : T$ \` by {\sc t-conv-a} or {\sc
       t-conv-u}\\
     {\bf Subcase:} $\ROUNDTWO{\Delta ; \Gamma} \vdash T(U\{\ov{v}\}) : T$ \` by {\sc t-conv-a}
     \\
     $\ROUNDTWO{\Delta ; \Gamma} \vdash U\{\ov{v}\} : U$ and $U \assg_\ROUNDTWO{\Delta }  T$ \` by inversion\\
     $\ROUNDTWO{\Delta ; \Gamma} \vdash \ov{v:S}$ and $(\ov{f\,\Phi}) = \fields_\ROUNDTWO{\Delta } (U)$ and $\ov{S
       \assg_\ROUNDTWO{\Delta }  \Phi\downharpoonright_1}$ \` by inversion\\
     $\fields_\ROUNDTWO{\Delta } (U) = \fields_\ROUNDTWO{\Delta } (T)$ since $U \assg_\ROUNDTWO{\Delta }  T$\\
     $\ROUNDTWO{\Delta ; \Gamma} \vdash T\{\ov{v}\} : T$ \` by {\sc t-literal}\\
            {\bf Subcase:} $\ROUNDTWO{\Delta ; \Gamma} \vdash T(U\{\ov{v}\}) : T$ \` by {\sc
              t-conv-u}\\
     $\ROUNDTWO{\Delta ; \Gamma} \vdash U\{\ov{v}\} : U$ and $\underlying_\ROUNDTWO{\Delta } (U) = \underlying_\ROUNDTWO{\Delta } (T)$ \`
     by inversion\\
     $\ROUNDTWO{\Delta ; \Gamma} \vdash \ov{v:S}$ and $(\ov{f\,\Phi}) = \fields_\ROUNDTWO{\Delta } (U)$ and  $\ov{S
       \assg_\ROUNDTWO{\Delta }  \Phi\downharpoonright_1}$ \` by inversion\\
     $\ROUNDTWO{\Delta ; \Gamma} \vdash T\{\ov{v}\} : T$ \` by {\sc t-literal}
   \end{tabbing}

 \item[Case:] $T(e) \becomes T(e')$ by {\sc r-context}

   \begin{tabbing}
     {\bf Subcase:} $\ROUNDTWO{\Delta;\Gamma } \vdash T(e) : T$ by {\sc t-conv-a}\\
     $\ROUNDTWO{\Delta;\Gamma }\vdash e:U$ and $U \assg_\ROUNDTWO{\Delta }   T$ \` by inversion\\
     $\ROUNDTWO{\Delta;\Gamma }\vdash e':U'$ with $U' \imp_\ROUNDTWO{\Delta }   U$ \` by i.h.\\
     $U' \assg_\ROUNDTWO{\Delta }   T$ \` by Lemma~\ref{lem:quasitrans}\\
     $\ROUNDTWO{\Delta;\Gamma }\vdash T(e') : T$ \` by {\sc t-conv-a}\\
     {\bf Subcase:} $\ROUNDTWO{\Delta;\Gamma }\vdash T(e) : T$ by {\sc t-conv-u}\\
     $\ROUNDTWO{\Delta;\Gamma }\vdash e:U$ and $\underlying_\ROUNDTWO{\Delta }  (U) = \underlying_\ROUNDTWO{\Delta }  (T)$ \` by inversion\\
     $\ROUNDTWO{\Delta;\Gamma }\vdash e':U'$ with $U' \imp_\ROUNDTWO{\Delta }   U$ \` by i.h.\\
     if $\uinterface_\ROUNDTWO{\Delta }  (U)$ then\\
     $U' \imp_\ROUNDTWO{\Delta }   T$ \` since $\underlying_\ROUNDTWO{\Delta }  (U) = \underlying_\ROUNDTWO{\Delta }  (T)$\\
     $U' \assg_\ROUNDTWO{\Delta }   T$ \` by $\assg$ rule \\
     $\ROUNDTWO{\Delta;\Gamma }\vdash T(e) : T$ \` by  {\sc t-conv-a}\\
     if $\neg\uinterface_\ROUNDTWO{\Delta }  (U)$ then $U' = U$ \` by inversion on $\imp$\\
     $\underlying_\ROUNDTWO{\Delta }  (U') = \underlying_\ROUNDTWO{\Delta }  (T)$ \` by assumption\\
     $\ROUNDTWO{\Delta;\Gamma }\vdash T(e) : T$ \` by  {\sc t-conv-u}\\
     \end{tabbing}
   \item[Case:] {\sc r-convert-i}

     \begin{tabbing}
       $\uinterface_\ROUNDTWO{\Delta }  (T)$ \` by inversion \\
       $T(w) \becomes w$ \` this case\\
       {\bf Subcase:} $\ROUNDTWO{\Delta;\Gamma }\vdash T(w) : T$ by {\sc t-conv-a}\\
       $\ROUNDTWO{\Delta;\Gamma }\vdash w : U$ with $U \assg_\ROUNDTWO{\Delta }   T$ \` by inversion\\
       Since either $\ubase_\ROUNDTWO{\Delta }  (U)$ or $\ustruct_\ROUNDTWO{\Delta }  (U)$ then\\
       $U \assg_\ROUNDTWO{\Delta }   T$ must follow from $U \imp_\ROUNDTWO{\Delta }   T$ and so\\
       $\ROUNDTWO{\Delta;\Gamma }\vdash w : U$ with $U \imp_\ROUNDTWO{\Delta }   T$ as needed.\\
        {\bf Subcase:} $\ROUNDTWO{\Delta;\Gamma }\vdash T(w) : T$ by {\sc t-conv-u}\\
       $\ROUNDTWO{\Delta;\Gamma }\vdash w : U$ with $\underlying_\ROUNDTWO{\Delta }  (U) = \underlying_\ROUNDTWO{\Delta }  (T)$ \` by inversion\\
       Impossible since $w$ is a value and $\uinterface_\ROUNDTWO{\Delta }  (T)$.
       \end{tabbing}
     \item[Case:] {\sc r-field}
       \begin{tabbing}
         $(\ov{f\,\Phi}) = \fields_\ROUNDTWO{\Delta } (T)$ \` this case\\
         $\ROUNDTWO{\Delta;\Gamma }\vdash T\{\ov{v}\}.f_i : \Phi_i$ \` by {\sc t-field}\\
         $\ROUNDTWO{\Delta;\Gamma }\vdash T\{\ov{v}\} : T$ \` by  inversion\\
         $\ROUNDTWO{\Delta;\Gamma }\vdash \ov{v : S}$ and $\ov{S \assg_\ROUNDTWO{\Delta }  \Phi\downharpoonright_1}$
         \` by inversion\\
         $\ROUNDTWO{\Delta;\Gamma }\vdash \Phi_i\downharpoonright_1(v_i) :
         \Phi_i\downharpoonright_1$ \` by {\sc t-conv-a}
       \end{tabbing}
       
       \item[Case:] $T\{\ov{v}, e , \ov{e}\} \becomes T\{\ov{v}, e' ,
         \ov{e}\}$
         \begin{tabbing}
           $\ROUNDTWO{\Delta;\Gamma }\vdash T\{\ov{v}, e , \ov{e}\} $ \` by {\sc t-literal}\\
           $\ROUNDTWO{\Delta;\Gamma }\vdash \ov{v},e,\ov{e} : \ov{S}$ with $(\ov{f\,\Phi}) =
           \fields_\ROUNDTWO{\Delta } (T)$ and $\ov{S\assg_\ROUNDTWO{\Delta } \Phi\downharpoonright_1}$ \` by
           inversion\\
           $\ROUNDTWO{\Delta;\Gamma }\vdash e' : S_i'$ with $S_i' \imp_\ROUNDTWO{\Delta }  S_i$ \` by i.h.\\
           $S_i' \assg_\ROUNDTWO{\Delta }  \Phi_i\downharpoonright_1$ \` by
           Lemma~\ref{lem:quasitrans}\\
           $\ROUNDTWO{\Delta;\Gamma }\vdash T\{\ov{v}, e' ,
         \ov{e}\} : T$ \` by {\sc t-literal}

         \end{tabbing}

       \item[Case:] $v'.m[\ov{S}](\ov{v}) \becomes U(e[x := v',\ov{y
           := \Phi(v)}])$ with $(x,\ov{y : \Phi}).e : U = \mbody(\vtype(v').m[\ov{S}])$

         \begin{tabbing}
           $\ROUNDTWO{\Delta;\Gamma }\vdash v'.m[\ov{S}](\ov{v}) : \Psi\downharpoonright
           [\eta]$ \` by {\sc t-call}
           \\
           $\ROUNDTWO{\Delta;\Gamma }\vdash v' : t[\ov{T}]$ and $m[\ov{\beta \, S'}](\ov{x \,
             \Phi})\,\Psi \in \methods_\ROUNDTWO{\Delta }(t[\ov{T}])$ \` by inversion\\
           $\ROUNDTWO{\Delta;\Gamma }\vdash \ov{v : T'}$ and $\eta = (\ov{\beta := S})$ and
           $\ov{\beta[\eta] \imp_\ROUNDTWO{\Delta } S'[\eta]}$ and $\ov{T'[\eta] \assg_\ROUNDTWO{\Delta }
             \Phi\downharpoonright_1[\eta]}$ \` by inversion\\
           $\ROUNDTWO{\Delta;\Gamma }\vdash \ov{\Phi(\ov{v})} : \ov{\Phi}$ \` by {\sc
             t-conv-a}\\
           $\ROUNDTWO{\Delta}; x:R , \ov{y : \Phi} \vdash e : U'$ with $U' \assg_\ROUNDTWO{\Delta }
           \Psi\downharpoonright_1[\eta]$ \` by
           Lemma~\ref{lem:wgbody}\\
           $\ROUNDTWO{\Delta ; \emptyset} \vdash e[x := v',\ov{y
             := \Phi(v)}] : U'$ \` by Lemma~\ref{lem:wgsubst}\\
           $\ROUNDTWO{\Delta ; \emptyset} \vdash U(e[x := v',\ov{y
             := \Phi(v)}]) : U$ \` by {\sc t-conv-a}\\
           \ROUNDTWO{$\Delta;\Gamma \vdash U(e[x := v',\ov{y
             := \Phi(v)}]) : U$} \` \ROUNDTWO{by weakening}
         \end{tabbing}

       \item[Case:] {\sc r-context} for method call

         Identical to contextual case for struct literals, appealing
         the induction hypothesis and Lemma~\ref{lem:quasitrans}.

       \item[Case:] {\sc r-op}

         \begin{tabbing}
           $\circ(\ov{v}) \becomes v'$ with $\delta(\circ,\ov{v}) =
           v'$ \` this case\\
           $\circ(\ov{v}) :T$ \` by {\sc t-op}\\
           $\ROUNDTWO{\Delta;\Gamma }\vdash \ov{v} : T$ and $\underlying_\ROUNDTWO{\Delta }(T) = B$ and $B \in
           \textit{base}(\circ)$ \` by inversion\\
           $\ROUNDTWO{\Delta;\Gamma }\vdash v' : T$ \` by definition of $\delta$\\
         
         \end{tabbing}

         \item[Case:] $\circ(\ov{v}, e, \ov{e}) \becomes \circ(\ov{v},
           e', \ov{e})$

           \begin{tabbing}
             $\circ(\ov{v}, e, \ov{e}) :T$ \` by {\sc t-op}\\
           $\ROUNDTWO{\Delta;\Gamma }\vdash \ov{v} , e , \ov{e} : T$ and $\underlying_\ROUNDTWO{\Delta }(T) = B$ and $B \in
           \textit{base}(\circ)$ \` by inversion\\
           $\ROUNDTWO{\Delta;\Gamma }\vdash e' : T'$ with $T' \imp_\ROUNDTWO{\Delta } T$ \` by i.h.\\
           $T' = T$ \` since $\underlying_\ROUNDTWO{\Delta }(T) = B$ \\
           $\ROUNDTWO{\Delta;\Gamma }\vdash \circ(\ov{v}, e', \ov{e}) : T$ \` by {\sc t-op}
           \end{tabbing}
         
  \end{description}
\end{proof}

\begin{lemma}[WG Canonical Forms]
  \label{lem:wgcanon}

  If $e$ is a value and $\emptyset;\emptyset \vdash e : T$ then either
  $T$ is a struct type and $e = T\br{\ov{v}}$, for some $\ov{v}$ or
  $T$ is a base type and $e = T(c)$, for some $c$.

\end{lemma}
\begin{proof}
Straightforward induction on typing.
 \end{proof}

\begin{theorem}[WG Progress]
If $\emptyset ; \emptyset \vdash e : T$ then either $e$ is a value, $e
\, \mathsf{panic}$ or $e \rightarrow e'$.
  \end{theorem}
  \begin{proof}
By induction on typing. The proof is essentially identical to that of
FGG~\cite{DBLP:journals/pacmpl/GriesemerHKLTTW20}.
We detail the cases pertaining to constants and type conversions.

\begin{description}

\item[Case:] {\sc t-assert}$_I$

  \begin{tabbing}
    $\vdash e.(T) : T$ \` this case\\
    $\vdash e :S$ and $\uinterface(S)$ \` by inversion\\
    $e \becomes e'$ or $e$ is a value or $e\,\mathsf{panic}$ \` by
    i.h.\\
    If $e$ reduces or panics the result is trivial\\
    {\bf Subcase:} $e$ is a value\\
    $e = w$ for some $w$ \` by Lemma~\ref{lem:wgcanon}\\
    if $\vtype(w) \imp T$ then we conclude by {\sc r-assert}\\
    otherwise then $e.(T)\,\mathsf{panic}$.
   \end{tabbing}

   Other type assert cases follow identical reasoning.

 \item[Case:] {\sc t-const}

   Trivial since constants are values.

 \item[Case:] {\sc t-conv-a}
   \begin{tabbing}
     $\vdash T(e) : T$ with $\vdash e : U$ and $U \assg T$ \` by
     inversion\\
     $e \becomes e'$, $e$ is a value or $e\,\mathsf{panic}$ \` by
     i.h.\\
     $e$ reducing or panic are trivial.\\
     {\bf Subcase:} $e$ is a value\\
     if $\mathit{base}(T)$ then we conclude by {\sc r-convert-b}\\
     if $\mathit{struct}(T)$ then we conclude by {\sc r-convert-s}\\
     
     if $\mathit{iface}(T)$ then we conclude by {\sc r-convert-i}\\
     \end{tabbing}
   
\end{description}

  \end{proof}

  \subsection{LWG Type Safety}
\label{app:lwg}
  
  \begin{lemma}[Quasi-transitivity]\label{lem:lwgquasitrans}
  Let $\Delta \vdash T, T',S\ok$. If $T \imp_\Delta T'$ and $T' \assg_\Delta S$ then $T
  \assg_\Delta S$.
\end{lemma}
\begin{proof}
Identical to Lemma~\ref{lem:quasitrans} since the relations are the same.
\end{proof}

\begin{lemma}[Type Substitution preserves Subtyping]\label{lem:lwgsubstsub}
If $T \imp_{\Delta_0, \ov{\alpha : U} , \Delta_1 } S$ , $\Delta_0
\vdash \ov{U} \ok$ and $\ov{U_0} \imp_{\Delta_0} \ov{U}[\ov{\alpha} :=
\ov{U_0}]$ then $T [\ov{\alpha} :=
\ov{U_0}] \imp_{\Delta_0,\Delta_1 [\ov{\alpha} :=
\ov{U_0}]} S [\ov{\alpha} :=
\ov{U_0}]$.
\end{lemma}
\begin{proof}
Identical to Lemma~\ref{lem:wgsubstsub} since the relations are the same.
\end{proof}

\begin{lemma}[Type Substitution preserves Assignability]
\label{lem:lwgsubstassign}
If $T \prec_{\Delta_0, \ov{\alpha : U} , \Delta_1 }: S$ , $\Delta_0
\vdash \ov{U} \ok$ and $\ov{U_0} \prec_{\Delta_0} \ov{U}[\ov{\alpha} :=
\ov{U_0}]$ then $T [\ov{\alpha} :=
\ov{U_0}] \prec_{\Delta_0,\Delta_1: [\ov{\alpha} :=
\ov{U_0}]} S [\ov{\alpha} :=
\ov{U_0}]$.
\end{lemma}
\begin{proof}
Identical to Lemma~\ref{lem:lwgsubstassign} since the relations are the same.
\end{proof}

\begin{lemma}[Type substitution preserves well-formedness (1)]\label{lem:lwgsubstwf}
If $\Delta_0 , \ov{\alpha : T} , \Delta_1 \vdash S \ok$ and $\ov{U}
\imp_{\Delta_0} T[\ov{\alpha} := \ov{U}] $ with $\Delta_0 \vdash
\ov{U}\nice$ then $\Delta_0, \Delta_1 [\ov{\alpha} := \ov{U}] \vdash S[\ov{\alpha} := \ov{U}]\ok$
\end{lemma}
\begin{proof}
Identical to Lemma~\ref{lem:wgsubstwf} since the relations are the same.
\end{proof}
\begin{lemma}[Type substitution preserves well-formedness (2)]\label{lem:lwgsubstnice}
If $\Delta_0 , \ov{\alpha : T} , \Delta_1 \vdash S \nice$ and $\ov{U}
\imp_{\Delta_0} T[\ov{\alpha} := \ov{U}] $ with $\Delta_0 \vdash
\ov{U}\nice$ then $\Delta_0, \Delta_1 [\ov{\alpha} := \ov{U}] \vdash S[\ov{\alpha} := \ov{U}]\nice$
\end{lemma}
\begin{proof}
By Lemma~\ref{lem:lwgsubstwf}
\end{proof}

\begin{lemma}[Type Substitution preserves Typing]
  \label{lem:lwgsubsttype}
    If $\Delta_0 , \ov{\alpha : T} , \Delta_1 ; \Gamma \vdash e : S@V$ and
  $\ov{U} \imp_{\Delta_0} T[\ov{\alpha} := \ov{U}]$ with $\Delta_0
  \vdash \ov{U} \nice$ then $\Delta_0, \Delta_1[\ov{\alpha} := \ov{U}]
  ; \Gamma[\ov{\alpha} := \ov{U}]
  \vdash e[\ov{\alpha} := \ov{U}] : (S@V)[\ov{\alpha} := \ov{U}]$
\end{lemma}
\begin{proof}
By induction on the given typing derivation. Similar to the analogue
property for WG.
\end{proof}

\begin{lemma}[Substitution preserves Typing]
  \label{lem:lwgsubst}
    If $\Delta ; \Gamma , \ov{x : T@V} \vdash e : S@U$ and $\Delta ;
  \Gamma \vdash \ov{e : T@V}$ then $\Delta ; \Gamma \vdash e[\ov{x} :=
  \ov{e}] : S@U$.
\end{lemma}
\begin{proof}
Straightforward induction, similar to the analogue
property for WG.
\end{proof}

\begin{lemma}
  \label{lem:lwgbody}
  If  $(m[\ov{\beta~S}](\ov{y~\Phi})~\Psi) \in \methods_\Delta(t[\ov{R}])$ and
$\mjump(t.m) = [\ov{\alpha};\ov{\beta}].(x , \ov{x}).e_0$ with
$\Delta \vdash \ov{S'} \nice$,
$\ov{S'}\imp_\Delta \ov{S[\ov{\beta} := \ov{S'}]}$ then 
$\Delta ; x : t[\ov{R}]@t[\ov{R}] , \ov{x : \Phi[\ov{\beta := S'}]} \vdash e_0[\ov{\alpha := R}][\ov{\beta := S'}] : \Psi[\ov{\beta := S'}]$
\end{lemma}
\begin{proof}
By inversion on $\mjump(t.m)$, Lemmas~\ref{lem:lwgsubsttype} and~\ref{lem:lwgsubstsub}.
 \end{proof}
  
  \begin{theorem}[Type Preservation]
    If $\ROUNDTWO{\Delta ; \Gamma} \vdash e : T@T'$ and $e \becomes e'$
    then $\ROUNDTWO{\Delta ; \Gamma}\vdash e : T@T'$.
  \end{theorem}
  \begin{proof}
    By induction on the reduction relation.

    \begin{description}
    \item[Case:] {\sc select}
      \begin{tabbing}
        $\ROUNDTWO{\Delta ; \Gamma}\vdash S\{\ov{v}\}.i : S_i @ V_i$ \` by inversion\\
        $\ROUNDTWO{\Delta ; \Gamma}\vdash S\{\ov{v}\} : T@T$ and $\ustruct_\ROUNDTWO{\Delta } (T)$ and $\fields_\ROUNDTWO{\Delta } (T) =
        (\ov{\_ \, S@V})$ \` by inversion\\
        $\ROUNDTWO{\Delta ; \Gamma}\vdash \ov{v : S@V}$ \` by inversion\\
        $\ROUNDTWO{\Delta ; \Gamma}\vdash v_i : S_i @ V_i$ as needed.
      \end{tabbing}

    \item[Case:] {\sc make}
      \begin{tabbing}
        {\bf Subcase:} $w.(S \rightarrow_\rho U@V) \becomes U@V(S,\rho,w)$ and $\uinterface_\ROUNDTWO{\Delta } (U)$ \\
        $\ROUNDTWO{\Delta ; \Gamma}\vdash w.(S \rightarrow_\rho U@V) : U@V$ \` by inversion on
        typing\\
        $\ROUNDTWO{\Delta ; \Gamma}\vdash w : S@S$ and $S\assg_\ROUNDTWO{\Delta }  U \imp_\ROUNDTWO{\Delta }  V$ and $\rho =
        \MKTABLE{S}{V}{\ROUNDTWO{\Delta } } $ \` by inversion\\
        $\ROUNDTWO{\Delta ; \Gamma}\vdash U@V(S,\rho,w) : U@V$ \` by
        {\sc t-interface} \ROUNDTWO{and weakening}\\
         {\bf Subcase:} $w.(S \rightarrow_\rho U@V) \becomes U@V(S,\rho,w)$ and $\neg\uinterface_\ROUNDTWO{\Delta } (U)$ \\
        $\ROUNDTWO{\Delta ; \Gamma}\vdash w.(S \rightarrow_\rho U@V) : U@V$ \` by inversion on
        typing\\
        $\ROUNDTWO{\Delta ; \Gamma}\vdash w : S@S$ and $S\assg_\ROUNDTWO{\Delta } U \imp_\ROUNDTWO{\Delta } V$ and $\rho =
        \MKTABLE{S}{V}{\ROUNDTWO{\Delta }} $ and $U=S$ \` by inversion\\
        $\ROUNDTWO{\Delta ; \Gamma}\vdash U@V(S,\rho,w) : U@V$ \` by
        {\sc t-interface} \ROUNDTWO{and weakening}
      \end{tabbing}

    \item[Case:] {\sc change-type}
      \begin{tabbing}
      $\Phi(S,\rho,w).(T\rightarrow T'@U) \becomes T'@U(S,\rho',w)$
      with $\rho' = \MKTABLE{S}{V}{\emptyset} $\\
      $\ROUNDTWO{\Delta ; \Gamma}\vdash \Phi(S,\rho,w).(T\rightarrow T'@U) : T'@U$ \` by
      inversion on typing\\
      $\ROUNDTWO{\Delta ; \Gamma}\vdash \Phi(S,\rho,w) : T@V$ and $T \assg_\ROUNDTWO{\Delta }  T' \imp_\ROUNDTWO{\Delta }  U$ and
      $\neg\uinterface_\ROUNDTWO{\Delta } (T)\Rightarrow S=T$ \` by
      inversion\\
      $\ROUNDTWO{\Delta ; \Gamma}\vdash w : S@S$ and $S \assg_\ROUNDTWO{\Delta }  T \imp_\ROUNDTWO{\Delta }  V$ \` by inversion\\
      if $\uinterface_\ROUNDTWO{\Delta } (T)$ then $S\assg_\ROUNDTWO{\Delta }  T$ gives us $S \imp T$\\
      $S \assg_\ROUNDTWO{\Delta }  T'$ \` by Lemma~\ref{lem:lwgquasitrans}\\
      $\ROUNDTWO{\Delta ; \Gamma} \vdash T'@U(S,\rho',w)$ \` by {\sc t-interface}\\
      if $\neg\uinterface(T)$ then $S=T$ and so $S \assg_\ROUNDTWO{\Delta }  T'$\\
      $\ROUNDTWO{\Delta ; \Gamma} \vdash T'@U(S,\rho',w)$ \` by {\sc t-interface}\\
    \end{tabbing}

  \item[Case:] {\sc static-change-s}

    \begin{tabbing}
      $\STATICCHANGE{S\{\ov{v}\}}{\tycomp{T}{}}{\tycomp{S}{}} \becomes
      T\{\ov{v}\}$\\
      $\ROUNDTWO{\Delta ; \Gamma} \vdash \STATICCHANGE{S\{\ov{v}\}}{\tycomp{T}{}}{\tycomp{S}{}} \becomes
      T\{\ov{v}\} : T@T$ \` this case\\
      $\ROUNDTWO{\Delta ; \Gamma} \vdash S\{\ov{v}\} : S@S$ and $\underlying_\ROUNDTWO{\Delta } (S) = \underlying_\ROUNDTWO{\Delta } (T)$ \` by inversion\\
      $\ROUNDTWO{\Delta ; \Gamma} \vdash \ov{v : \Phi}$ and $\ov{\Phi} = \fields_\ROUNDTWO{\Delta } (S)$ \` by
      inversion\\
      $\ov{\Phi} = \fields_\ROUNDTWO{\Delta } (T)$ \` by $\underlying_\ROUNDTWO{\Delta } (S) = \underlying_\ROUNDTWO{\Delta } (T)$ and
      $\neg\uinterface_\ROUNDTWO{\Delta } (S,T)$\\
      $\ROUNDTWO{\Delta ; \Gamma} \vdash T\{\ov{v}\} :T$ \` by {\sc t-literal}\\

     \end{tabbing}

   \item[Case:] {\sc static-change-c}
     \begin{tabbing}
 $\STATICCHANGE{S(c)}{\tycomp{T}{}}{\tycomp{S}{}} \becomes
 T(c)$ \` this case\\
  $\STATICCHANGE{S(c)}{\tycomp{T}{}}{\tycomp{S}{}} : T@T$ \` this case\\
$\ROUNDTWO{\Delta ; \Gamma} \vdash S(c) : S@S   $and $\underlying_\ROUNDTWO{\Delta }(S) = \underlying_\ROUNDTWO{\Delta }(T)$ \` by
inversion\\
$\underlying_\ROUNDTWO{\Delta }(S) = \mathit{ctype}(c)$ \` by inversion\\
$= \underlying_\ROUNDTWO{\Delta }(T)$ \\
$\ROUNDTWO{\Delta ; \Gamma} \vdash T(c) : T@T$ \` by {\sc t-const}

    \end{tabbing}

  \item[Case:] {\sc assert-ok-s}

    \begin{tabbing}
      $T@V(S,\rho,w).(S@S)\becomes w$ \` this case\\
      $\ROUNDTWO{\Delta ; \Gamma} \vdash T@V(S,\rho,w).(S@S) :S@S$ and $S \imp_\ROUNDTWO{\Delta } V$ \` by
      inversion\\
      $\ROUNDTWO{\Delta ; \Gamma} \vdash w : S@S$ \` by inversion\\
    \end{tabbing}

      \item[Case:] {\sc assert-ok-i}

    \begin{tabbing}
      $T@V(S,\rho',w).(U@V)\becomes U@V(S,\rho,w)$ with $S\imp_\ROUNDTWO{\Delta } U$ and
      $\rho = \MKTABLE{S}{V}{\emptyset} $ \` this case\\
      $\ROUNDTWO{\Delta ; \Gamma} \vdash T@V(S,\rho,w).(U@V) :U@V$ \` by
      inversion\\
      $\ROUNDTWO{\Delta ; \Gamma} \vdash T@V(S,\rho,w) : T@V$ and $U\imp_\ROUNDTWO{\Delta } V$ and $\uinterface_\ROUNDTWO{\Delta }(U)$  \` by inversion\\
      $\ROUNDTWO{\Delta ; \Gamma} \vdash w : S@S$ and $S\assg_\ROUNDTWO{\Delta } T \imp_\ROUNDTWO{\Delta } V$ \` by inversion\\
      To show: $S \assg_\ROUNDTWO{\Delta } U \imp_\ROUNDTWO{\Delta } V$\\
      $S \assg_\ROUNDTWO{\Delta } U$ \` by $S\imp_\ROUNDTWO{\Delta } U$\\
      $U \imp_\ROUNDTWO{\Delta } V$ \` assumption\\
      $\ROUNDTWO{\Delta ; \Gamma} \vdash  U@V(S,\rho,w)$ \` by {\sc t-interface} 
    \end{tabbing}

  \item[Case:] {\sc call-dyn}
    \begin{tabbing}
      $T_1@T_2(U,\rho,w).m[\ov{S}](\ov{v}) \becomes e[\ov{\alpha :=
        R}][\ov{\beta := S}][x := w][\ov{y := v}]$ \` this case\\
      $U = t[\ov{R}]$ and $[\ov{\alpha} ; \ov{\beta}](x,\ov{y}).e =
      \mjump(\rho(m))$ \` inversion\\
      $\ROUNDTWO{\Delta ; \Gamma}\vdash \ov{ v : \Phi[\eta]}$ and
      $T_1@T_2(U,\rho,w).m[\ov{S}](\ov{v}) : \Psi[\eta]$ \\
      $T_1@T_2(U,\rho,w) : T_1 @T_2$ and $\uinterface_\ROUNDTWO{\Delta }(T_2)$ and
      $\ov{\beta[\eta] \imp S'[\eta]}$ \\
      $m[\ov{\beta \,S'}][\ov{y\,\Phi}]\,\Psi \in \methods_\ROUNDTWO{\Delta }(T_2)$ and
      $\eta = (\ov{\beta := S})$ \` by inversion\\
      $\ROUNDTWO{\Delta ; \Gamma}\vdash w : t[\ov{R}]@t[\ov{R}]$ and $t[\ov{R}] \assg_\ROUNDTWO{\Delta } T_1 \imp_\ROUNDTWO{\Delta }
      T_2$ and $\rho = \MKTABLE{t[\ov{R}]}{T_2}{\ROUNDTWO{\Delta }}$ \` by
      inversion\\
      $t[\ov{R}] \imp_\ROUNDTWO{\Delta } T_2$ and $\rho(m) = t.m$\` by inversion\\
      $m[\ov{\beta \,S'}][\ov{y\,\Phi}]\,\Psi \in \methods_\ROUNDTWO{\Delta }(t[\ov{R}])$
      \` since $t[\ov{R}] \imp T_2$\\

      $\ROUNDTWO{\Delta } ; x:t[\ov{R}]@t[\ov{R}] , \ov{ y : \Phi[\eta]} \vdash e[\ov{\alpha
          := R}][\ov{\beta := S}] : \Psi[\eta]$ \` by
        Lemma~\ref{lem:lwgbody}\\
      $\ROUNDTWO{\Delta ; \Gamma}\vdash e[\ov{\alpha :=
        R}][\ov{\beta := S}][x := w][\ov{y := v}] : \Psi[\eta]$ \` by 
      Lemma~\ref{lem:lwgsubst}\\
    \end{tabbing}

  \item[Case:] {\sc call-static}

    \begin{tabbing}
      $w \smsep t[\ov{R}] .m[\ov{S}](\ov{v}) : \Psi[\eta]$ and $[\ov{\alpha} ; \ov{\beta}](x,\ov{y}).e =
      \mjump(t.m)$ and
      $m[\ov{\beta\,S'}](\ov{y \,\Phi})\,\Psi \in \methods(t[\ov{R}])$\\
      $\ROUNDTWO{\Delta ; \Gamma}\vdash w : t[\ov{R}]@t[\ov{R}]$ and $\eta = (\ov{\beta := S})$
      \` by inversion\\
      $\ROUNDTWO{\Delta ; \Gamma}\vdash \ov{v : \Phi[\eta]}$ and $\ov{\beta[\eta] \imp S'[\eta]}$
        \` by inversion\\
        $\ROUNDTWO{\Delta } ; x:t[\ov{R}]@t[\ov{R}] , \ov{y : \Phi[\eta]}
        \vdash e[\ov{\alpha := R}][\eta] : \Psi[\eta]$ \` by
        Lemma~\ref{lem:lwgbody}\\
        $\ROUNDTWO{\Delta ; \Gamma}\vdash e[\ov{\alpha :=
        R}][\ov{\beta := S}][x := w][\ov{y := v}] : \Psi[\eta]$ \` by 
      Lemma~\ref{lem:lwgsubst} \ROUNDTWO{and weakening}\\
     \end{tabbing}

   \item[Case:] {\sc op-dyn}
     \begin{tabbing}
$\circ(\ov{v}) \becomes v''$ with $\forall i. v_i = T@V(t[\ov{S}],\rho,v_i')$
and $\delta(\rho(\circ),\ov{v'}) = v''$ \` this case\\
$\ROUNDTWO{\Delta ; \Gamma}\vdash \circ(\ov{v}) : T@V$ \` by inversion\\
$\ROUNDTWO{\Delta ; \Gamma}\vdash \ov{v} : T@V$ and $\underlying_\ROUNDTWO{\Delta}(T) = B$ and $B \in \dots$ \` by
inversion\\
$\ROUNDTWO{\Delta ; \Gamma}\vdash \ov{v} : t[\ov{S}]@t[\ov{S}]$ and $t[\ov{S}]\assg_\ROUNDTWO{\Delta} T \imp_\ROUNDTWO{\Delta} V$ and $\rho =
\MKTABLE{S}{V}{\ROUNDTWO{\Delta }} $ \` by inversion\\
$\rho(\circ) = \circ_t$ and $t[\ov{S}] \imp_\ROUNDTWO{\Delta} V$ \` by inversion\\
$\ROUNDTWO{\Delta ; \Gamma}\vdash v'' : T@V$ \` by definition of $\delta$
     \end{tabbing}

     Remaining cases are standard.
     
  \end{description}

  \end{proof}

\begin{lemma}[LWG Canonical Forms]
\label{lem:lwgcanon}
  If $e$ is a value and $\emptyset;\emptyset \vdash e : T@V$ then either:

  \begin{itemize}
    \item $\uinterface(V)$ and $e = T@V(S,\rho,w)$ for some $S, \rho,
      w$.
    \item $\neg\uinterface(V)$ and $T=V$ with $\ustruct(T)$ and $e =
      T\{\ov{v}\}$, for some $\ov{v}$ or $\ubase(T)$ and $e = T(c)$
      for some $c$.
   \end{itemize}
   
 \end{lemma}
\begin{proof}
Straightforward induction on typing.
\end{proof}

\begin{theorem}[LWG Progress]
If $\vdash e :T@V$ then $e\becomes e'$, $e$ is a value or $e\,\mathsf{panic}$.
\end{theorem}
\begin{proof}
  By induction on typing. We show a few cases:

  \begin{description}
  \item[Case:] {\sc t-make}
    \begin{tabbing}
      $\vdash \MAKETYPE{e}{\typair{U}{V}}{T}{\rho} : U@V$ \` this case\\
      $\vdash e : T@T$ and $\neg\uinterface(T)$ \` by inversion\\
       $e\becomes e'$, $e$ is a value or $e\,\mathsf{panic}$ \` by
       i.h.\\
       Reduction and panic cases are trivial.\\
       $e$ is a value:\\
       $e = w$ for some $w$ \` by Lemma~\ref{lem:lwgcanon}\\
       We conclude by reduction {\sc make} accordingly.
     \end{tabbing}

   \item[Case:] {\sc t-change}
     \begin{tabbing}
       $\vdash \CHANGETYPE{e}{\ghostytypair{U}{V'}}{T} : U@V'$ \` by
       inversion\\
       $\vdash e : T@V$ \` by inversion\\
       $e\becomes e'$, $e$ is a value or $e\,\mathsf{panic}$ \` by
       i.h.\\
       Reduction and panic cases are trivial.\\
       $e$ is a value:\\
       Since $\uinterface(V)$, $e = T@V(S,\rho,w)$ \` by
       Lemma~\ref{lem:lwgcanon}\\
       We conclude by reduction {\sc change-type} accordingly.
     \end{tabbing}

   \item[Case:] {\sc t-static-change}

     \begin{tabbing}
       $\vdash \STATICCHANGE{e}{U}{T} : U@U$ \` this case\\
       $\vdash e : T@T$ \` by inversion\\
       $\neg\uinterface(T,U)$ and $\underlying(T) = \underlying(U)$ \`
       by inversion\\
       $e\becomes e'$, $e$ is a value or $e\,\mathsf{panic}$ \` by
       i.h.\\
       Reduction and panic cases are trivial.\\
       $e$ is a value:\\
       Since $\neg\uinterface(T,U)$, $e = w$ \` by
       Lemma~\ref{lem:lwgcanon}\\
       We conclude by reductions {\sc static-change-s} or {\sc
         static-change-c} accordingly.
     \end{tabbing}

   \item[Case:] {\sc t-dyn-call}
     
      Standard, noting that $\uinterface(V)$ provides us with the
      appropriate information for Lemma~\ref{lem:lwgcanon} in the
      appropriate subcase, where $e$ equals some $R@V(t[\ov{U}],\rho,w)$, allowing
      us to conclude $\rho(m) = t.m$ and $\ov\alpha \stoup \ov\beta](x,  \ov{y}).e =
      \mjump(\rho(m))$.

      Other cases are straightforward.
  \end{description}
 \end{proof}

 \subsection{Compilation}
 \label{app:comp}

 We use the following auxiliary definitions below:
 \begin{mathpar}
   \inferrule{ }{
\lfloor \Delta ; \Gamma \rfloor \triangleq \lfloor \Delta \rfloor ;
\lfloor \Gamma \rfloor_\Delta}

\inferrule{ }{
\lfloor \Delta , \alpha : T \rfloor \triangleq \lfloor \Delta \rfloor
, \alpha : \lfloor T\rfloor}

\inferrule{ }{
\lfloor \Gamma , x:T \rfloor_\Delta \triangleq \lfloor \Gamma \rfloor_\Delta
, x : \lfloor T\rfloor @ \lfloor \bounds_\Delta(T)\rfloor}
\end{mathpar}

\begin{lemma}[Compilation Preserves Well-formedness]
  \label{lem:compwf}
  ~
  \begin{itemize}
\item $\Delta \vdash T \ok$ implies $\lfloor \Delta \rfloor \vdash \lfloor T
\rfloor\ok$
\item $\Delta \vdash T \nice$ implies $\lfloor \Delta \rfloor \vdash \lfloor T
\rfloor\nice$
\end{itemize}
\end{lemma}
\begin{proof}
Straightforward induction on definition of compilation of types.
 \end{proof}

 \begin{lemma}[Compilation Preserves Implements and Assignability]
   \label{lem:compsubs}
   ~
     \begin{itemize}
       \item If $T \imp_\Delta U$ then $\lfloor T\rfloor \imp_{\lfloor
           \Delta\rfloor} \lfloor U \rfloor$
                  \item If $T \assg_\Delta U$ then $\lfloor T\rfloor \assg_{\lfloor
                      \Delta\rfloor} \lfloor U \rfloor$
                  \item if $\underlying(U)=\underlying(T)$ then
                    $\underlying(\lfloor T\rfloor) = \underlying(\lfloor U\rfloor$
        \end{itemize}
      \end{lemma}
\begin{proof}
Straightforward induction on definition of compilation of types.
\end{proof}

 \begin{theorem}[Compilation Preserves Types]\label{thm:compilation-preserves-typing}
   If $\Delta ; \Gamma \vdash e : T$:
   \begin{itemize}

     \item $\Delta ; \Gamma 
       \vdash e \mapsto e'$ implies $\lfloor\Delta ; \Gamma \rfloor
       \vdash e' : \lfloor T\rfloor @ \lfloor \bounds_\Delta(T)\rfloor$.
     \item $\Delta ; \Gamma \vdash e : T \mapsto_{U@V} e'$ and $T\assg
       U$ and ($\uinterface(V)$ and $\neg\uinterface(U)$ implies $U=T$) or ($\neg\uinterface(T)$ and $\neg\uinterface(U)$ and
       $\underlying(T) = \underlying(U)$) then
       $\lfloor\Delta ; \Gamma\rfloor \vdash e' : U@V$ .

      \item $\Delta \stoup\Gamma \vdash d  \mathbin{:} U
        \Mapsto_{\typair{T}{V}}  d'$ and $T\assg
       U$ and ($\uinterface(V)$ and $\neg\uinterface(U)$ implies $U=T$) or ($\neg\uinterface(T)$ and $\neg\uinterface(U)$ and
       $\underlying(T) = \underlying(U)$) and $\neg\uinterface(V)$ implies
        $\lfloor\Delta ; \Gamma\rfloor \vdash e' :
        \tycomp{T@\mathsf{Any}}{}$
\item $\Delta \stoup\Gamma \vdash d  \mathbin{:} U
        \Mapsto_{\typair{T}{V}}  d'$ and $T\assg
       U$ and ($\uinterface(V)$ and $\neg\uinterface(U)$ implies $U=T$) or ($\neg\uinterface(T)$ and $\neg\uinterface(U)$ and
       $\underlying(T) = \underlying(U)$) and $\uinterface(V)$ implies
        $\lfloor\Delta ; \Gamma\rfloor \vdash e' :
        \tycomp{T@V}{}$
        
      \end{itemize}
    \end{theorem}

    \begin{proof}
      By simultaneous induction on the definition of compilation. Statements (3)
      and (4) above follow immediately by i.h.(2). We often
      make implicit use of Lemmas~\ref{lem:compwf} and~\ref{lem:compsubs} above
      for conciseness.

      \begin{description}
      \item[Case:] {\sc var} and {\sc const}

        Trivial.

      \item[Case:] {\sc select-boxed}

        \begin{tabbing}
$\vdash e.f \mapsto e'.i.(\lfloor S_i\rfloor@\lfloor V_i\rfloor)$ and $\neg\utyvar(S_i)$ \`
this case\\
$\vdash e.f : S_i$ \` by inversion on typing\\
$\vdash e : T$ and $\ustruct(T)$ \` by inversion on typing\\
$\vdash e' : \lfloor T@T \rfloor$ \` by i.h., since $\ustruct(T)$\\
$\vdash e'.i : \lfloor S_i\rfloor@\lfloor V_i\rfloor$ \` by {\sc
  t-field}\\
if $\uinterface(S_i)$ then $\vdash e'.i. (\lfloor S_i\rfloor@\lfloor
V_i\rfloor) : \lfloor S_i\rfloor@\lfloor V_i\rfloor$ \` by {\sc
  t-assert-i}\\
if $\neg\uinterface(S_i)$ then $V_i = S_i$ and $\vdash e'.i. (\lfloor S_i\rfloor@\lfloor
V_i\rfloor) : \lfloor S_i\rfloor@\lfloor S_i\rfloor$ \` by {\sc t-assert-s}\\
      \end{tabbing}
  
\item[Case:] {\sc select-unboxed}

  \begin{tabbing}
    $\vdash e.f \mapsto e'.i$ and $\ov{f\,S@V} = \fields(T)$ and $\uinterface(V_i) \vee \utyvar(S_i)$ \\
    $\vdash e.f : S_i$ \` this case\\
    $\vdash e : T$ \` by inversion\\
    $\vdash e' : \lfloor T@T \rfloor$ \` by i.h since $\ustruct(T)$\\
    $\vdash e'.i : \lfloor S_i\rfloor@\lfloor V_i\rfloor$ \` by {\sc
  t-field}\\
    
   \end{tabbing}
  
 \item[Case:] {\sc assert}

   \begin{tabbing}
$\vdash e.(T) \mapsto e'.(\lfloor T\rfloor @ \lfloor
\bounds(T)\rfloor)$ \` this case\\
{\bf Subcase:} $\vdash e.(T) : T$ by {\sc t-assert-i}\\
$\vdash e: S$ and $\uinterface(S)$ and $\uinterface(T)$ \` by
inversion\\
$\vdash e' : \lfloor S\rfloor @ \lfloor
\bounds(S)\rfloor$ \` by i.h.\\
$\vdash e'.(\lfloor T\rfloor @ \lfloor
\bounds(T)\rfloor) : \lfloor T\rfloor @ \lfloor
\bounds(T)\rfloor$ \` by {\sc t-assert-i}\\

   {\bf Subcase:} $\vdash e.(T) : T$ by {\sc t-assert-s}\\

$\vdash e: S$ and $\uinterface(S)$ and $\neg\uinterface(T)$ and $T\imp
\bounds(S)$ \` by
inversion\\
$\vdash e' : \lfloor S\rfloor @ \lfloor
\bounds(S)\rfloor$ \` by i.h.\\
 $\neg\uinterface(T) \Rightarrow \bounds(T)=T$ \\
$\vdash e'.(\lfloor T\rfloor @ \lfloor
T\rfloor) : \lfloor T\rfloor @ \lfloor
\bounds(T)\rfloor$ \` by {\sc t-assert-s}\\
   \end{tabbing}
       
 \item[Case:] {\sc conversion}

   \begin{tabbing}
$\vdash T(e) \mapsto e'$ with $\vdash e:S$ and $e:S \mapsto
T@\bounds(T) e'$ \` this case\\
{\bf Subcase:} typing by {\sc t-conv-a}\\
$\vdash e : S$ and $S\assg T$ \` by inversion\\
$\vdash e' : \lfloor T@\bounds(T)\rfloor$ \` by i.h.(2)\\
{\bf Subcase:} typing by {\sc t-conv-u}\\
$\vdash e : S$ and $\neg\utyvar(S,T)$ $\underlying(S)=\underlying(T)$ \` by inversion\\
$\vdash e' : \lfloor T@\bounds(T)\rfloor$ \` by i.h.(2)\\
\end{tabbing}

\item[Case:] {\sc struct}

  \begin{tabbing}
    $\vdash T\{\ov{e}\} \mapsto \lfloor T \rfloor \{\ov{e'}\}$ with
    $\vdash \ov{e:U}$ $(\ov{f\,\Phi}) = \fields(T)$ and $\vdash e : U
    \mapsto_\Phi e'$ \` this case\\
    $\vdash \ov{e' : \lfloor \Phi\rfloor}$ \` by i.h.(2)\\
      $\vdash \lfloor T\rfloor\{\ov{e'}\} : \lfloor T\rfloor@\lfloor
      T\rfloor$ \` by {\sc t-struct}\\
    \end{tabbing}

  \item[Case:] {\sc op-dyn}

    \begin{tabbing}
$\vdash \circ(\ov{e}) \mapsto \circ(\ov{e'})$ with $\vdash \ov{e} :
\alpha$ and $e \mapsto e'$\` this case\\
$\vdash \circ(\ov{e}) : \alpha$ and $\vdash \ov{e} : \alpha$
        and
        $\{ \underlying(V) \mid V \in \types(\Delta(\alpha)) \}$ \` by inversion\\
$\vdash \ov{e'} : \lfloor \alpha\rfloor @ \lfloor \bounds(\alpha) \rfloor$
\` by i.h.\\
$\vdash \circ(\ov{e'}) : \lfloor \alpha\rfloor @ \lfloor
\bounds(\alpha) \rfloor$ \` by {\sc t-op-dyn}
    \end{tabbing}

  \item[Case:] {\sc op-static}

    \begin{tabbing}
$\vdash \circ(\ov{e}) \mapsto \circ\smsep \lfloor T \rfloor (\ov{e'})$
with $\vdash \ov{e} : T$ and $\neg\utyvar(T)$ and $e \mapsto e'$ \`
this case\\
$\vdash \circ(\ov{e}) : T$ and $\underlying(T) = B$ and
$B\in\dom(\circ)$ \` by inversion on typing\\
$\vdash \ov{e'} : \lfloor T \rfloor @ \lfloor T\rfloor$ \` by i.h.\\
$\vdash \circ\smsep \lfloor T \rfloor (\ov{e'}) : \lfloor T \rfloor @
\lfloor T\rfloor$ \` by {\sc t-op-static}\\
    \end{tabbing}

  \item[Case:] {\sc change}

    \begin{tabbing}
$\vdash e : U \mapsto_{T@V}
\CHANGETYPE{e'}{\tycomp{\ghostytypair{T}{V}}{\Delta}}{\tycomp{U}{\Delta}}$
and $\uinterface(U,V)$ with $U \assg T$ \` this case\\
$\vdash e' : \lfloor U\rfloor @ \lfloor \bounds(U)\rfloor$ \` by
i.h.(1)\\
$\lfloor U \rfloor \assg \lfloor T \rfloor \imp \lfloor V \rfloor$ \`
by Lemma~\ref{lem:compsubs}\\
$\vdash
\CHANGETYPE{e'}{\tycomp{\ghostytypair{T}{V}}{\Delta}}{\tycomp{U}{\Delta}}
: T@V$ \` by {\sc t-change}

     \end{tabbing}

   \item[Case:] {\sc static-change}

     \begin{tabbing}
$\vdash e  \mathbin{:} U \mapsto_{\typair{T}{T}}
\STATICCHANGE{e'}{\tycomp{T}{\Delta}}{\tycomp{U}{\Delta}}$ with $\neg\uinterface(U)$
  and
      $\neg\uinterface(T)$ and
      $\vdash e \mapsto e'$ \` this case\\
      $\vdash e' : U$ \` by inversion on typing\\
      $\vdash e' : \lfloor U \rfloor @ \lfloor U \rfloor$ \` by i.h.\\
      $\underlying(\lfloor U \rfloor) = \underlying(\lfloor T\rfloor)$
      \` by Lemma~\ref{lem:compsubs}\\
      $\vdash
      \STATICCHANGE{e'}{\tycomp{T}{\Delta}}{\tycomp{U}{\Delta}} :
      \lfloor T \rfloor @ \lfloor T \rfloor$ \` by {\sc t-static-change}
     \end{tabbing}

   \item[Case:] {\sc make-iface}

     \begin{tabbing}
$e \mathbin{:} S \mapsto_{\typair{T}{V}}
\MAKETYPE{e'}{{\tycomp{\typair{T}{V}}{\Delta}}}{\tycomp{S}{\Delta}}{\rho}$
with $\uinterface(V)$
      and
      $\neg\uinterface(S)$
      and
      $\uinterface(T)$\\
      and
      $\vdash e \mapsto e'$ and $\rho
      =\MKTABLE{\tycomp{S}{}}{\tycomp{V}{}}{}$ \` this case\\
      $\vdash e' : \lfloor S\rfloor @ \lfloor S \rfloor$ \` by i.h.\\
      $\lfloor S \rfloor \assg \lfloor T\rfloor$ \` assumption and
      Lemma~\ref{lem:compsubs}\\
      $\lfloor T\rfloor \imp V$ \` from $T@V$ and
      Lemma~\ref{lem:compsubs}\\
      $\vdash
      \MAKETYPE{e'}{{\tycomp{\typair{T}{V}}{\Delta}}}{\tycomp{S}{\Delta}}{\rho}
      : T@V$ \` by {\sc t-make}
    \end{tabbing}

  \item[Case:] {\sc make-bs}

    \begin{tabbing}
$\vdash e \mathbin{:}S \mapsto_{\typair{T}{V}}
      \MAKETYPE{
        \STATICCHANGE{e'}{\tycomp{T}{}}{\tycomp{S}{}}
      }{{\tycomp{\typair{T}{V}}{\Delta}}}{\tycomp{T}{\Delta}}{\rho}$
      with
      $\uinterface(V)$
      and
      $\neg\uinterface(S)$\\
      and
      $\neg\uinterface(T)$
      and
      $\vdash e \mapsto e' $ and$ \rho =\MKTABLE{\tycomp{T}{}}{\tycomp{V}{}}{\Delta}$
      \\
      \` this case\\

      $\vdash e' : \tycomp{S}{}@ \tycomp{S}{}$ \` by i.h.\\
      $\underlying(\tycomp{S}{}) = \underlying(\tycomp{T}{})$ \` since
      $S\assg T$ and Lemma~\ref{lem:compsubs}\\
      $\vdash \STATICCHANGE{e'}{\tycomp{T}{}}{\tycomp{S}{}} :
      \tycomp{T}{} @\tycomp{T}{}$ \` by {\sc t-static-change}\\
      $\tycomp{T}{} \assg \tycomp{T}{}$ \` by  reflexivity\\
      $\tycomp{T}{} \imp \tycomp{V}{}$ \` by inversion on $\rho$\\
      $\vdash \MAKETYPE{
        \STATICCHANGE{e'}{\tycomp{T}{}}{\tycomp{S}{}}
      }{{\tycomp{\typair{T}{V}}{\Delta}}}{\tycomp{T}{\Delta}}{\rho} :
      \tycomp{T@V}{}$ \` by {\sc t-make}
    \end{tabbing}

  \item[Case:] {\sc call-s-boxed}

    \begin{tabbing}
$\vdash e . m[\ov{S}](\ov{d})
      \mapsto
      e'\smsep \tycomp{R}{\Delta} .  m[\ov{\tycomp{S}{\Delta}}]( \ov{d'}) .
      {\color{black}
        (\typair{\tycomp{T[\eta]}{\Delta}}{\tycomp{V[\eta]}{\Delta}})
      }$\\
      with
          $\vdash e \mathbin{:}  R$
      and
      $\neg\uinterface(R)$
      and
      $m[\ov{\beta~S''}](\ov{y\:\Phi})\: {\color{black} \typair{T}{V}
      } \mathbin{\in} \methods(R)$\\
      
      and
      $\neg\utyvar_\Delta(T[\eta])$
      and
      $\uinterface(V[\eta])$
      and
      $\eta = (\ov{\beta \by S})$
      and
      $\vdash e \mapsto e'$\\
      and
      $ \vdash \ov{d : U}$
      and
      $ \vdash \ov{ d  \mathbin{:} U  \mapsto_{\Phi[\eta]}  d'}$ \`
      this case\\
      $\vdash  e . m[\ov{S}](\ov{d}) : \Psi\downharpoonright_1[\eta] =
      T[\eta]$ \` typing\\
      $\vdash e' : \tycomp{R}{}@\tycomp{R}{}$ \` by i.h.\\
      $\ov{U[\eta] \imp \Phi\downharpoonright_1[\eta]}$ \` by typing\\
        $\ov{\beta[\eta] \imp S'[\eta]}$ \` by typing\\
          $\vdash d' : \tycomp{\Phi[\eta]}{}$ \` by i.h.(2)\\
      $m[\ov{\beta~\tycomp{S''}{}}](\ov{y\:\tycomp{\Phi}{}})\: {\color{black} \tycomp{\typair{T}{V}}{}
      } \mathbin{\in} \methods(\tycomp{R}{})$ \` since methods are
      preserved by compilation\\
      $\vdash e'\smsep \tycomp{R}{\Delta} .
      m[\ov{\tycomp{S}{\Delta}}]( \ov{d'}) : \tycomp{T@V}{}$ \` by
      {\sc t-static-call}\\
      if $\uinterface(\tycomp{T}{})$ then  $\vdash e'\smsep \tycomp{R}{\Delta} .  m[\ov{\tycomp{S}{\Delta}}]( \ov{d'}) .
      {\color{black}
        (\typair{\tycomp{T[\eta]}{\Delta}}{\tycomp{V[\eta]}{\Delta}})}
        : \tycomp{T@T}{}$ \` by {\sc t-assert-i}\\
        if  $\neg\uinterface(\tycomp{T}{})$ then $V \imp V'$ and so
        we conclude by {\sc t-assert-s}\\
\end{tabbing}
     
      \item[Case:] {\sc call-s-unboxed}
        \begin{tabbing}
$\vdash e . m[\ov{S}](\ov{d})
        \mapsto
        e'\smsep \tycomp{R}{\Delta} .  m[\ov{\tycomp{S}{\Delta}}](
        \ov{d'})$ with \\
          $\vdash e \mathbin{:}  R$
        and
        $\neg\uinterface(R)$
        and
        $m[\ov{\beta~S''}](\ov{y\:\Phi})\: {\color{black} \typair{T}{V} } \mathbin{\in} \methods(R)$
        and
        $\eta = (\ov{\beta \by S})$
        \\
        $( \neg \uinterface(V[\eta]) \lor \utyvar(T[\eta]))$
        and
        $\vdash e \mapsto e'$
        and
        $\vdash \ov{d : U}$
        and
        $\vdash \ov{ d  \mathbin{:} U  \mapsto_{\Phi[\eta]}  d'}$ \`
        this case\\
        $\vdash e' : \tycomp{R}{}@\tycomp{R}{}$ \` by i.h.\\
        $\vdash \ov{d' : \tycomp{\Phi[\eta]}{}}$ \` by i.h.(2)\\
          $m[\ov{\beta~\tycomp{S''}{}}](\ov{y\:\tycomp{\Phi}{}})\: {\color{black} \tycomp{\typair{T}{V}}{}
      } \mathbin{\in} \methods(\tycomp{R}{})$ \` since methods are
      preserved by compilation\\
      $\vdash e'\smsep \tycomp{R}{\Delta} .
      m[\ov{\tycomp{S}{\Delta}}]( \ov{d'}) : \tycomp{T@V}{}[\eta]$ \` by
      {\sc t-static-call}\\
      $\vdash e . m[\ov{S}](\ov{d}) : V[\eta]$ \` by WG typing\\
      if $\utyvar(T[\eta])$ then $\bounds(T) = V$ \` by construction\\
      if $\neg\uinterface(V[\eta])$ then $\neg\uinterface(T[\eta])$
      and so $T=V$\\
    \end{tabbing}

  \item[Case:] {\sc call-i}

    \begin{tabbing}
$\vdash e . m[\ov{S}](\ov{d})
      \mapsto
      e' .  m_D[\ov{{\tycomp{S}{\Delta}}}](\ov{d'}) .
      {\color{black}
        (\typair{\tycomp{T[\eta]}{\Delta}}{\tycomp{T[\eta]}{\Delta}})
      }$ with\\
      $\vdash e \mathbin{:} R$
      and
      $\uinterface(R)$
      and
      $m[\ov{\beta~S''}](\ov{y\:\Phi})\: {\color{black} \typair{T}{V}
      } \mathbin{\in} \methods(R)$\\
      and
      $\neg\utyvar(T[\eta])$
      and
      $\eta = (\ov{\beta \by S})$
      and
      $\vdash e \mapsto e'$
      \\
      $\vdash \ov{d : U}$
      and
      $ \vdash \ov{ d  \mathbin{:} U  \Mapsto_{\Phi[\eta]}  d'}$ \`
      this case\\
      $\vdash e . m[\ov{S}](\ov{d}) : T$ \` by typing\\
      $\vdash e' : \tycomp{R@\bounds(R)}{}$ \` by i.h.\\
      $\vdash \ov{d' : \tycomp{\langle\Phi\rangle}{}}$ by i.h.(3) and (4) \\
        
        $m_D[\ov{\beta~\tycomp{S''}{}}](\ov{y\:\tycomp{\langle\Phi\rangle}{}})\: {\color{black} \tycomp{\langle\typair{T}{V}\rangle}{}
        } \mathbin{\in} \methods(\tycomp{\bounds(R)}{})$ \` by
        construction\\
        $\vdash e' .  m_D[\ov{{\tycomp{S}{\Delta}}}](\ov{d'}) :
        \tycomp{\langle\typair{T}{V}\rangle}{}$ \` by {\sc t-call-i}\\
        if $\uinterface(T)$ then $\vdash e' .  m_D[\ov{{\tycomp{S}{\Delta}}}](\ov{d'}) .
      {\color{black}
        (\typair{\tycomp{T[\eta]}{\Delta}}{\tycomp{T[\eta]}{\Delta}})
      } : T@T$ \` by {\sc t-assert-i}\\
        if $\neg\uinterface(T)$ then $\vdash e' .  m_D[\ov{{\tycomp{S}{\Delta}}}](\ov{d'}) .
      {\color{black}
        (\typair{\tycomp{T[\eta]}{\Delta}}{\tycomp{T[\eta]}{\Delta}})
      } : T@T$ \` by {\sc t-assert-s}\\
\end{tabbing}
    \item[Case:] {\sc call-i-}$\alpha$

      Follows a similar scheme as the previous case but where we
      conclude by {\sc t-change}, noting that
      $\tycomp{T[\eta]@V[\eta]}{} = \tycomp{T[\eta]@\bounds(T[\eta])}{} $.

     \end{description}
     \end{proof}

\subsection{Behavioural equivalence}\label{app:beh-eq}
\newcommand{\synthcast}[1]{#1^{\circ}}

\begin{definition}[Lowering Simulation]\label{def:lowering-sim}
  Given a WG program $\ov{D} \prog e$ and a LWG program $\ov{D^\dagger} \prog d$.
  A relation $\mathcal{R}$ is a \emph{lowering simulation} when for all $(e,d) \in \mathcal{R}$
  \begin{itemize}
  \item if $e$ is a value then $\eqval{e}{d}$,
  \item if $e \not\becomes  $ then $d \not\becomes$, and
  \item if $e \becomes e'$ then there is $d'$ such that $d \becomes^\ast d'$ and $(e',d') \in \mathcal{R}$.
  \end{itemize}
\end{definition}

\begin{Figure}
  \begin{ruled}
    Equivalence for values  \hfill \fbox{$\eqval{e}{d}$}
    \begin{mathpar}
      \inferrule
      {\eqval{v}{v'}
        \and
        S \assg_\emptyset T \imp_\emptyset V
      }
      {\eqval{v}{\ghost{\typair{T}{V}}(S,\rho, v')}}

      \inferrule
      {
        \ov{\eqval{v}{v'}}
        \and
        T \assg_\emptyset U 
      }
      {\eqval{T\br{\ov{v}}}{\ghost{\tycomp{U}{}}\br{\ov{v}}}}

      \inferrule            
      {
        T \assg_\emptyset U
      }
      {\eqval{T(c)}{U(c)}}
    \end{mathpar}
  \end{ruled}
  \caption{Correspondence between WG and LWG values.}   \label{fig:eqval-relation}
\end{Figure}

\begin{Figure}
  \begin{ruled}
    Equivalence modulo synthetic casts  \hfill \fbox{$\eqcast{d_1}{d_2}$ with $\Delta \stoup \Gamma \vdash  d_i : \Phi_i$}
    \begin{mathpar}
      \inferrule[refl]
      {\,}
      {\eqcast{d}{d}}

      \inferrule[sym]
      {\eqcast{d'}{d}}
      {\eqcast{d}{d'}}

      \inferrule[base]
      {
        T \assg_\Delta U
      }
      {\eqcast{T(c)}{U(c)}}

      \inferrule[stru]
      {
        T \assg_\Delta U
        \and
        \ov{\eqcast{d}{d'}}
      }
      {\eqcast{\ghost{T}\br{\ov{d}}}{\ghost{{U}}\br{\ov{d'}}}}

      \inferrule[ifaceval]
      { \, }
      {\eqcast{w}{\ghost{\Phi}(S,\rho, w)}}

      \inferrule[ifacevalchange]
      {\eqcast{d}{w}
      }
      {\eqcast{\CHANGETYPE{d}{\ghostytypair{T}{V}}{U}}{\ghost{\typair{T}{V}}(S,\rho, w)}}

      \inferrule[ifacevalassert]
      {\eqcast{d}{w}
      }
      {\eqcast{d. (\typair{T}{V})}{\ghost{\typair{T}{V}}(S,\rho, w)}}      
      
      \inferrule[ifacevalmake]
      {\eqcast{d}{w}}
      {\eqcast{\MAKETYPE{d}{\Phi}{S}{\rho}}{\ghost{\typair{T}{S}}(S,\rho, w)}}

        \inferrule[dassert]
        {%
          \eqcast{d}{d'}
        }
        {\eqcast{d}{d'.\synthcast{(\typair{V}{T})}}}
      
      \inferrule[dmakeS]
      {\ %
        \eqcast{d}{d'}
      }
      {\eqcast{d}{\MAKETYPE{d'}{\typair{T}{V}}{S}{\rho}}}

      \inferrule[dchange]
      {%
        \eqcast{d}{d'}
      }
      {\eqcast{d}{\CHANGETYPE{d'}{\ghostytypair{T}{S}}{T}}}
      
      \inferrule[dstaticchange]
      {%
        \eqcast{d}{d'}
      }
      {\eqcast{d}{\STATICCHANGE{d'}{T}{U}}}

      \inferrule[mcallbridge]
      {\eqcast{v}{v'}
        \and
        {\Delta \stoup \Gamma \vdash \ov{\ceqcast{f}{f'}}}
      }
      {\eqcast{v \smsep T . m[\ov{S}](\ov{f})}{v' . m_D[\ov{S}](\ov{f'})}}      
      
      \inferrule[dmcall]
      {\eqcast{d}{d'}
        \and
        {\Delta \stoup \Gamma \vdash \ov{\ceqcast{f}{f'}}}
      }
      {\eqcast{d.m[\ov{S}](\ov{f})}{d'.m[\ov{S}](\ov{f'})}}

      \inferrule[smcall]
      {\eqcast{d}{d'}
        \and
        {\Delta \stoup \Gamma \vdash \ov{\ceqcast{f}{f'}}}
      }
      {\eqcast{d \smsep T . m[\ov{S}](\ov{f})}{d' \smsep T . m[\ov{S}](\ov{f'})}}

      \inferrule[field]
      {\eqcast{d}{d'}}
      {\eqcast{d.i}{d'.i}}

      \inferrule[assert]
      {\eqcast{d}{d'}}
      {\eqcast{d.(\typair{V}{T})}{d'.(\typair{V}{T})}}

      \inferrule[op-d]
      {\Delta \stoup \Gamma \vdash \ov{\ceqcast{d}{d'}}}
      {\eqcast{\circ(\ov{d})}{\circ(\ov{d'})}}

      \inferrule[op-s]
      {\Delta \stoup \Gamma \vdash \ov{\ceqcast{d}{d'}}}
      {\eqcast{\circ\smsep T(\ov{d})}{\circ\smsep T(\ov{d'})}}

      \inferrule[changetype]
      {\eqcast{d}{d'}}
      {\eqcast{\CHANGETYPE{d}{\ghostytypair{T}{S}}{T}}{\CHANGETYPE{d'}{\ghostytypair{T}{S}}{T}}}

      \inferrule[make]
      {\eqcast{d}{d'}}
      {\eqcast{\MAKETYPE{d}{\typair{T}{V}}{T}{\rho}}{\MAKETYPE{d'}{\typair{T}{V}}{T}{\rho}}}

      \inferrule[staticchange]
      {\eqcast{d}{d'}}
      {\eqcast{\STATICCHANGE{d}{T}{U}}{\STATICCHANGE{d'}{T}{U}}}
    \end{mathpar}
  \end{ruled}
  \caption{Equivalence between LWG \emph{well-typed} expressions modulo synthetic casts.
    We write $\ceqcast{d}{d'}$ for $\eqcast[\emptyset \stoup \emptyset]{d}{d'}$ and
    we write ${e.\synthcast{(\typair{V}{T})}}$ to indicate a type-assert added by the compiler.
  }
  \label{fig:eqcast-relation}
\end{Figure}

\begin{lemma}\label{lem:tycomp-preserve-subtyping}
  Given a WG program $\ov{D} \prog e$ and its compilation $\ov{D^\dagger} \prog d$.
  For all types $S$ and $T$ defined in $\ov{D}$, the following holds:
  \begin{itemize}
  \item If $ S \imp_\Delta T \iff \tycomp{S}{} \imp_\Delta \tycomp{T}{}$.
  \item If $ S \assg_\Delta T \iff \tycomp{S}{} \assg_\Delta \tycomp{T}{}$.
  \end{itemize}
\end{lemma}
\begin{proof}
  Straightforward from the definitions of $\imp$, $\assg$, and
  $\tycomp{\_}{}$. Note that $\tycomp{\_}{}$ renames consistently each
  method $m$ in interfaces to $m_D$ (incl.\ anonymous types) and, for
  each method $m$ defined on non-interface types, an additional
  adaptor method $m_D$ is defined.
\end{proof}

\begin{lemma}[Type pairs]\label{lem:type-pairs-properties}
  For any $e$ such that
  $\Delta \stoup \Gamma \vdash e : \typair{S}{T}$,
  \begin{itemize}
  \item if $\uinterface(S) \implies \uinterface(T)$, and
  \item if $\neg\uinterface(T) \implies S = T$.
  \end{itemize}
\end{lemma}
\begin{proof}
  Straightforward from the judgement $\Delta \stoup \Gamma \vdash e : \typair{S}{T}$.
\end{proof}

\begin{lemma}\label{lem:lwg-val-form}
  Suppose $\ov{D} \ok$.
  Let $v$ be a value such that $\emptyset \stoup \emptyset \vdash v : \typair{T}{V}$
  \begin{itemize}
  \item $\uinterface(V)$ then $v =    \ghost{\typair{T}{V}}(S, \rho, e)$, $\emptyset \stoup \Gamma \vdash e \mathbin{:}  \typair{S}{S}$,
    $S \assg_\emptyset T  \imp_\emptyset V$
  \item $\neg \uinterface(V)$ then $v = T\{ w \}$ or $v = T(c)$ and $V=T$
  \end{itemize}
\end{lemma}
\begin{proof}
  Immediate by inversion from $\emptyset \stoup \emptyset \vdash v : \typair{T}{V}$ (Figure~\ref{fig:lwg:type-expr1}).
\end{proof}

\begin{lemma}\label{lem:compil-judgements}
  Suppose  $\ov{D} \ok$ and $\Delta \stoup \Gamma \vdash e : U$.
  If  $\Delta \stoup \Gamma \vdash e    \mapsto d$
  and
  $\Delta \stoup \Gamma \vdash e  : U \mapsto_\Psi d'$
  then $\eqcast{d}{d'}$.
\end{lemma}
\begin{proof}
  Straightforward case analysis on the judgement of the form
  $\Delta \stoup \Gamma \vdash e : U  \mapsto_\Psi d'$.
\end{proof}

\begin{lemma}\label{lem:compositionality}
  Suppose
  $\ov{D} \ok$ and
  $\Delta \stoup \ov{x: U} \vdash e  \mathbin{:} T $ and
  $\Delta \stoup \emptyset \vdash \ov{v \mathbin{:} U} $.
  Assume $\Delta \stoup \ov{x: U} \vdash e  \mapsto e'$ and
  $\Delta \stoup \emptyset \vdash \ov{v \mapsto d} $.
  Then
  $\Delta \stoup \emptyset \vdash e[\ov{x := v}]  \mapsto e'[\ov{x := d}]$.
\end{lemma}
\begin{proof}
  Mechanical by investigating the type derivation of
  $\Delta \stoup \ov{x: U} \vdash e \mathbin{:} T $ with
  Theorem~\ref{thm:compilation-preserves-typing}.
\end{proof}

\begin{lemma}\label{lem:bodies}
  Suppose $\ov{D} \ok$ and we have
  $\mbody(t[\ov{R'}].m[\ov{S'}]) = (x,\ov{y}).e$ in WG
  such that  
  \[
    \func~(x~t[\ov{\alpha~R}])~m[\ov{\beta~S}](\ov{y~\typair{T}{T'}})~\typair{V}{V'} ~\br{\return~e'} \in \ov{D}
    \;    \text{and} \;
    \emptyset  \stoup  x: t[\ov{\alpha}], \ov{y: T} \vdash e  \mathbin{:} V \mapsto d
  \]

  and
  $\mjump(t.m) = [\ov{\alpha}, \ov{\beta}](x,\ov{y}).d'$
  in LWG.

  Then we have $\eqcast[{\emptyset  \stoup  x: t[\ov{\alpha}], \ov{y: T} }]{d}{d'[\ov{\alpha := \tycomp{R}{}}][\ov{\beta := \tycomp{S'}{}}]}$.
\end{lemma}
\begin{proof}
  Straightforward application of Rule \textsc{d-func} and Lemma~\ref{lem:compil-judgements}.
\end{proof}

\begin{lemma}\label{lem:compilation-context-eqcast}
  Suppose $\ov{D} \ok$.
  If
  $\Delta \stoup \Gamma \vdash e  \mathbin{:} T $,
  and
  $\Delta \stoup \Gamma \vdash e \mapsto d$,
  for any context $E$, such that  $\Delta \stoup \Gamma \vdash E[e] \mathbin{:} V $,
  there exists an LWG
  context $E'$ and term $d'$ such that:
  $\Delta \stoup \Gamma \vdash E[e] \mapsto E'[d']$
  and 
  $\eqcast{d'}{d}$.
\end{lemma}
\begin{proof}
  By structural induction on the derivation of
  $\Delta \stoup \Gamma \vdash e \mapsto d$.
\end{proof}

\begin{lemma}\label{lem:operator-eqval}
  If $\delta(\circ, \ov{v}) = v'$ with
  $\emptyset \stoup \emptyset \vdash \ov{v } : T$ in WG.
  Then
  for any $\ov{w}$ such that $\ov{\eqval{v}{w}}$ and $\delta(\circ_T, \ov{w}) = w'$, we have $\eqval{v'}{w'}$.
\end{lemma}
\begin{proof}
  This follows from the fact that operators only take base types (or
  named based types) as operand.
\end{proof}

\begin{lemma}\label{lem:value-comp}
  Suppose $\ov{D} \ok$.
  Let $v$ be a WG value s.t. $\vdash v  \mathbin{:} T$ . Then following holds:
  \begin{enumerate}
  \item If $ \vdash v \mapsto e$ then $e \becomesast v'$, where $v'$
    is an LWG value s.t. $\eqval{v}{v'}$  and $\ceqcast{e}{v'}$.
  \item If $\vdash v \mathbin{:} T \mapsto_{\typair{U}{V}} e$ and
    $T \assg_\emptyset U \imp_\emptyset V$ then $e \becomesast v'$,
    where $v'$ is an LWG value s.t. $\eqval{v}{v'}$  and
    $\ceqcast{e}{v'}$.
  \item If
    $ \vdash v \mathbin{:} T
    \Mapsto_{\typair{U}{V}} e$ and
    $T \assg_\emptyset U \imp_\emptyset V$ then $e \becomesast v'$,
    where $v'$ is an LWG value s.t. $\eqval{v}{v'}$  and
    $\ceqcast{e}{v'}$.
  \end{enumerate}
\end{lemma}
\begin{proof}
  We prove this by induction on the structure of $v$.
  
  \noindent
  \textbf{Base cases.}

  \begin{itemize}
  \item $v = T(c)$

    For (1) by Rule \textsc{Const}, we have
    $ \vdash T(c) \mapsto T(c)$, hence we have the result with
    $ e = v' = T(c)$, and $\ceqcast{e}{v'}$, trivially.

    For (2) since
    $T$ is not an interface, it can only be derived
    from
    \begin{itemize}
    \item Rule \textsc{static-change}.
      We have $\vdash T(c) \mathbin{:} T \mapsto_{\typair{U}{U}} e$.
      This gives us $e =
      \STATICCHANGE{T(c)}{\tycomp{U}{}}{\tycomp{T}{}}$.
      where we have $\tycomp{U}{} = U$ and ${\tycomp{T}{}} = T$ since
      they cannot be literal/anonymous types. 
      
      Thus we have $ e \becomes U(c) = v'$, an LWG value, and
      $\eqval{v}{v'}$ since $T \assg U$.

      Finally, we have
      $e = \ceqcast{\STATICCHANGE{T(c)}{{U}}{{T}}}{U(c)} = v'$ easily
      by Figure~\ref{fig:eqcast-relation}.
    \item Rule \textsc{make-iface}.
      We have
      $\vdash T(c) \mathbin{:} T \mapsto_{\typair{U}{V}} e$ with
      $\uinterface(V)$ and $\uinterface(U)$.
      Pose $\rho = \MKTABLE{{T}}{V}{}$, we have
      $e = \MAKETYPE{T(c)}{{{\typair{U}{V}}{}}}{{T}{}}{\rho}$.

      Thus we have
      $ e \becomes \ghost{\typair{U}{V}}(T,\rho, T(c)) = v'$, an LWG
      value and $\eqval{v}{v'}$.

      Finally, we have
      $e =
      \ceqcast{\MAKETYPE{T(c)}{{{\typair{U}{V}}{}}}{{T}{}}{\rho}}{\ghost{\typair{U}{V}}(T,\rho,
        T(c))} = v'$
      easily
      by Figure~\ref{fig:eqcast-relation}.
    \item Rule \textsc{make-bs}.
      We have
      $\vdash T(c) \mathbin{:} T \mapsto_{\typair{U}{V}} e$ with
      $\uinterface(V)$ and $\neg\uinterface(U)$.
      Pose $\rho = \MKTABLE{{T}}{V}{}$, we have
      $e = \MAKETYPE{T(c)}{{{\typair{U}{V}}{}}}{{U}{}}{\rho}$.

      Thus we have
      $ e \becomes \ghost{\typair{U}{V}}(U,\rho, T(c)) = v'$, an LWG
      value and $\eqval{v}{v'}$.

      Finally, we have
      $e =
      \ceqcast{\MAKETYPE{T(c)}{{{\typair{U}{V}}{}}}{{U}{}}{\rho}}{\ghost{\typair{U}{V}}(U,\rho, T(c))} = v'$
      easily
      by Figure~\ref{fig:eqcast-relation}.
    \end{itemize}

    For (3), if $\uinterface(V)$, then it reduces to case (2). If
    $\neg\uinterface(V)$ then we have
    $\vdash v \mathbin{:} T \mapsto_{\typair{U}{\textit{Any}}} e$ and
    $T \assg U \imp \textit{Any}$, thus case (2) applies again.

  \item $ v= T\br{} $

    For (1) we have $\vdash T\br{} \mapsto \ghost{\tycomp{T}{}}\br{}$.
    Thus $e = \ghost{\tycomp{T}{}}\br{} = v'$ and it is already a
    value and we have $\eqval{v}{v'}$ and $\ceqcast{e}{v'}$ trivially.

    For (2) since $T$ is not an interface, it can only be derived from
    \begin{itemize}
    \item Rule \textsc{static-change}.
      We have $\vdash T\br{} \mathbin{:} T \mapsto_{\typair{U}{U}} e$.
      This gives us $e =
      \STATICCHANGE{\ghost{\tycomp{T}{}}\br{}}{\tycomp{U}{}}{\tycomp{T}{}}$.

      Thus we have $ e \becomes \ghost{\tycomp{U}{}}\br{} = v'$, an LWG value
      and $\eqval{v}{v'}$ since $T \assg U$.

      Finally, we have
      $e =
      \ceqcast{\STATICCHANGE{\ghost{\tycomp{T}{}}\br{}}{\tycomp{U}{}}{\tycomp{T}{}}}{\ghost{\tycomp{U}{}}\br{}} = v'$
      easily
      by Figure~\ref{fig:eqcast-relation}.

    \item Rule \textsc{make-iface}.
      We have
      $\vdash T\br{} \mathbin{:} T \mapsto_{\typair{U}{V}} e$ with
      $\uinterface(V)$ and $\uinterface(U)$.
      Pose $\rho = \MKTABLE{{T}}{V}{}$, we have
      $e = \MAKETYPE
      {\ghost{\tycomp{T}{}}\br{}}
      {\tycomp{\typair{U}{V}}{}}
      {\ghost{\tycomp{T}{}}}{\rho}$.

      Thus we have
      $ e \becomes \ghost{\tycomp{\typair{U}{V}}{}}(T,\rho, {\ghost{\tycomp{T}{}}\br{}}) = v'$, an
      LWG value and $\eqval{v}{v'}$.

      Finally, we have
      $e =
      \ceqcast{\MAKETYPE
        {\ghost{\tycomp{T}{}}\br{}}
        {\tycomp{\typair{U}{V}}{}}
        {\ghost{\tycomp{T}{}}}{\rho}}{
        \ghost{\tycomp{\typair{U}{V}}{}}(\ghost{\tycomp{T}{}},\rho,
        {\ghost{\tycomp{T}{}}\br{}})} = v'$ easily by
      Figure~\ref{fig:eqcast-relation}.
    \item Rule \textsc{make-bs}.
      We have
      $\vdash T\br{} \mathbin{:} T \mapsto_{\typair{U}{V}} e$ with
      $\uinterface(V)$ and $\neg\uinterface(U)$.
      Pose $\rho = \MKTABLE{{T}}{V}{}$, we have
      $e = \MAKETYPE{\ghost{\tycomp{T}{}}\br{}}{{{\tycomp{\typair{U}{V}}{}}{}}}{{\tycomp{U}{}}{}}{\rho}$.

      Thus we have
      $ e \becomes \ghost{\tycomp{\typair{U}{V}}{}}({\tycomp{U}{}},\rho, {\ghost{\tycomp{T}{}}\br{}}) = v'$, an LWG
      value and $\eqval{v}{v'}$.

      Finally, we have
      $e =
      \ceqcast{\MAKETYPE{\ghost{\tycomp{T}{}}\br{}}{{{\tycomp{\typair{U}{V}}{}}{}}}{{{\tycomp{U}{}}}{}}{\rho}}{ \ghost{\tycomp{\typair{U}{V}}{}}({\tycomp{U}{}},\rho, {\ghost{\tycomp{T}{}}\br{}})} = v'$
      easily
      by Figure~\ref{fig:eqcast-relation}.
    \end{itemize}
    
  \end{itemize}

  \noindent
  \textbf{Inductive cases.}
  Assume the lemma holds for all substructures of a struct value. We
  need to prove it holds for $v = T\br{v_1, \ldots, v_n}$.

  For (1), by Rule \textsc{struct}, we have
  $\vdash T\br{v_1, \ldots, v_n} \mapsto \ghost{\tycomp{T}{}}\br{e_1,
    \ldots , e_n}$, where $ \vdash v_i \mapsto e_i$ for each $i$.
  By the induction hypothesis, each $ e_i \becomesast v''_i$ where
  $v'_i$ is an LWG value and $\eqval{v_i}{v''_i}$.
  Therefore,
  $\ghost{\tycomp{T}{}}\br{e_1, \ldots , e_n} \becomesast
  \ghost{\tycomp{T}{}}\br{v''_1, \ldots , v''_n}$, an LWG value, through several
  applications of the \textsc{context} rule.
  By definition of $\eqval{}{}$, we also have
  $\eqval{T\br{v_1, \ldots, v_n}}{ \ghost{\tycomp{T}{}}\br{v''_1, \ldots
      , v''_n}}$.
  By induction hypothesis, we also have $\ceqcast{e_i}{v''_i}$ for each
  $i$, hence, we have
  $\ceqcast{\ghost{\tycomp{T}{}}\br{e_1, \ldots , e_n}}{ \ghost{\tycomp{T}{}}\br{v''_1, \ldots , v''_n}}$.

  For (2) since $T$ is not an interface, it can only be derived from
  \begin{itemize}
  \item Rule \textsc{static-change}.  We have
    $\vdash T\br{\ov{v}} \mathbin{:} T \mapsto_{\typair{U}{U}} e$.
    This gives us
    $e =
    \STATICCHANGE
    {\ghost{\tycomp{T}{}}\br{\ov{e}}}
    {\tycomp{U}{}}
    {\tycomp{T}{}}$.
    where $v_i : T_i \mapsto_{\Phi_i} e_i$ for each $i$.

    By the induction hypothesis, each $e_i \becomesast v''_i$, where
    $v''_i$ is an LWG value and $\eqval{v_i}{v''_i}$ and $\ceqcast{e_i}{v''_i}$.

    Thus we have $ e \becomesast \STATICCHANGE
    {\ghost{\tycomp{T}{}}\br{\ov{v''}}}
    {\tycomp{U}{}}
    {\tycomp{T}{}}
    \becomes
    {\ghost{\tycomp{U}{}}\br{\ov{v''}}} = v'
    $
    an LWG value
    and $\eqval{v}{v'}$ by IH and the fact that $T \assg U$.

    Finally, we have
    $\ceqcast{\STATICCHANGE
      {\ghost{\tycomp{T}{}}\br{\ov{e}}}
      {\tycomp{U}{}}
      {\tycomp{T}{}}{}}{{\ghost{\tycomp{U}{}}\br{\ov{v''}}}}$ by IH, Figure~\ref{fig:eqcast-relation}, and the fact that $T \assg U$.

  \item Rule \textsc{make-iface}.
    We have
    $\vdash T\br{\ov{v}} \mathbin{:} T \mapsto_{\typair{U}{V}} e$ with
    $\uinterface(V)$ and $\uinterface(U)$.
    Pose $\rho = \MKTABLE{{T}}{V}{}$, we have
    $e = \MAKETYPE{\ghost{\tycomp{T}{}}\br{\ov{e}}}{{{\tycomp{\typair{U}{V}}{}}{}}}{{T}{}}{\rho}$
    where $v_i : T_i \mapsto_{\Phi_i} e_i$ for each $i$.

    By the induction hypothesis, each $e_i \becomesast v''_i$, where
    $v''_i$ is an LWG value and $\eqval{v_i}{v''_i}$ and $\ceqcast{e_i}{v''_i}$.
    
    Thus we have $ e \becomesast
    \MAKETYPE{\ghost{\tycomp{T}{}}\br{\ov{v''}}}{{{\tycomp{\typair{U}{V}}{}}{}}}{{T}{}}{\rho}
    \becomes
    \ghost{\tycomp{\typair{U}{V}}{}}(T,\rho, \ghost{\tycomp{T}{}}\br{\ov{v''}}) = v'$, an LWG value and $\eqval{v}{v'}$ by IH.

    Finally, we have
    $\ceqcast{\MAKETYPE{\ghost{\tycomp{T}{}}\br{\ov{e}}}{{{\tycomp{\typair{U}{V}}{}}{}}}{{T}{}}{\rho}}{\ghost{\tycomp{\typair{U}{V}}{}}(T,\rho,  \ghost{\tycomp{T}{}}\br{\ov{v''}})}$ by IH and Figure~\ref{fig:eqcast-relation}.

  \item Rule \textsc{make-bs}.
    We have
    $\vdash T\br{\ov{v}} \mathbin{:} T \mapsto_{\typair{U}{V}} e$ with
    $\uinterface(V)$ and $\neg\uinterface(U)$.
    Pose $\rho = \MKTABLE{{T}}{V}{}$, we have
    $e = \MAKETYPE{\ghost{\tycomp{T}{}}\br{\ov{e}}}{{{\tycomp{\typair{U}{V}}{}}{}}}{{{\tycomp{U}{}}}{}}{\rho}$
    where $v_i : T_i \mapsto_{\Phi_i} e_i$ for each $i$.

    By the induction hypothesis, each $e_i \becomesast v''_i$, where
    $v'_i$ is an LWG value and $\eqval{v_i}{v''_i}$ and $\ceqcast{e_i}{v''_i}$  and $\ceqcast{e_i}{v''_i}$.

    Thus we have
    $ e
    \becomesast \MAKETYPE{\ghost{\tycomp{T}{}}\br{\ov{v'}}}{{{\tycomp{\typair{U}{V}}{}}{}}}{{{\tycomp{U}{}}}{}}{\rho}
    \becomes \ghost{\tycomp{\typair{U}{V}}{}}({\tycomp{U}{}},\rho, \ghost{\tycomp{T}{}}\br{\ov{v'}}) = v'$, an LWG
    value and $\eqval{v}{v'}$.

    Finally, we have
    $\ceqcast
    {\MAKETYPE{\ghost{\tycomp{T}{}}\br{\ov{e}}}{{{\tycomp{\typair{U}{V}}{}}{}}}{{{\tycomp{U}{}}}{}}{\rho}}
    {\ghost{\tycomp{\typair{U}{V}}{}}({\tycomp{U}{}},\rho, \ghost{\tycomp{T}{}}\br{\ov{v'}})}$ by IH and
    Figure~\ref{fig:eqcast-relation}.
    \qedhere
  \end{itemize}

  For (3), if $\uinterface(V)$, then it reduces to case (2). If
  $\neg\uinterface(V)$ then we have
  $\vdash v \mathbin{:} T \mapsto_{\typair{U}{\textit{Any}}} e$ and
  $T \assg U \imp \textit{Any}$, thus we follow the same argument as in case (2).
\end{proof}

\begin{lemma}\label{lem:simulation-ed}
  Suppose $\ov{D} \ok$ and $\emptyset \stoup \emptyset \vdash e : T$ for some
  $T$.
  Then, the relation
  \[
    \mathcal{R} \eqdef
    \{ \,
    (e,d') \st
    \exists d . \
    \emptyset \stoup \emptyset \vdash
    e
    \mapsto
    d
    \text{ and }
    \ceqcast{d}{d'}
    \,
    \}
  \]
  is a lowering simulation.
\end{lemma}
\begin{proof}
  We show that
  \[
    \mathcal{R} \eqdef
    \{ \,
    (e,d') \st
    \exists d . \
    \emptyset \stoup \emptyset \vdash
    e
    \mapsto
    d
    \text{ and }
    \ceqcast{d}{d'}
    \,
    \}
  \]
  is a lowering simulation. We proceed by induction on the structure of
  the reduction.

  Figure~\ref{fig:lem-simulation-ed} gives an overview of our proof
  strategy.
  We start from $e$, which evaluates (one-step) to $e'$ and compiles to $d$, and assume there is $d'$
  such that $\ceqcast{d}{d'}$,
  then we show that there is $d'''$ such that $d' \becomesast d'''$,
  with $d'''$ a compilation target from $e'$, possibly via a
  $\ceqcast{}{}$-equivalent term ($d''$); hence we have that
  $(e',d''') \in \mathcal{R}$.

  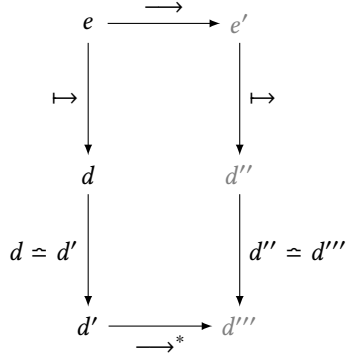
\begin{figure}
    \begin{tikzpicture}[
      auto,
      process/.style={minimum size=0.5cm},
      >=latex
      ]

      \node[process] (P1) at (0,2) {$e$};
      \node[process, gray] (P2) at (2,2) {$e'$};
      \node[process] (Q1) at (0,0) {$d$};
      \node[process, gray] (Q2) at (2,0) {$d''$};
      \node[process] (R1) at (0,-2) {$d'$};
      \node[process,gray] (R2) at (2,-2) {$d'''$};

      \draw[->] (P1) -- (P2) node[midway, above] {$\becomes$};
      \draw[->] (R1) -- (R2) node[midway, below] {$\becomesast$};
      \draw[->] (P1) -- (Q1) node[midway, left] {$\mapsto$};
      \draw[->] (P2) -- (Q2) node[midway, right] {$\mapsto$};
      \draw[->] (Q1) -- (R1) node[midway, left] {$\ceqcast{d}{d'}$};
      \draw[->] (Q2) -- (R2) node[midway, right] {$\ceqcast{d''}{d'''}$};
    \end{tikzpicture}
    \caption{Key diagram for the proof of Lemma~\ref{lem:simulation-ed}.}
    \label{fig:lem-simulation-ed}
  \end{figure}

  \begin{itemize}
  \item if $e \not\becomes $ then, by type safety of WG, we must have
    $e = v.(T)$ with $\neg( \vtype(v) \imp_{\black{\emptyset}} T)$.
    By Rule \textsc{Assert}, we have $d = n . (\tycomp{\typair{T}{T}}{})$, with
    $ \emptyset \stoup \emptyset \vdash v \mapsto n$.
    By Lemma~\ref{lem:value-comp}, we have $n \becomesast v'$ and $\ceqcast{n}{v'}$ and $\eqval{v}{v'}$.
    Assume by contradiction that $d \becomes d'$.
    Then by the reduction rules, we have 
    
    $v' =   {\ghost{\typair{U}{V}}(S,\rho, v'')}$.
    Since we have $\eqval{v}{v'}$, we also have 
    $  S \assg_\emptyset U \imp_\emptyset V$
    and $\eqval{v}{v''}$ hence $\tycomp{\vtype(v)}{} = S$.
    Thus, since we have $ \neg(\vtype(v) \imp_{\black{\emptyset}})T$, we
    cannot have $T = S$ (Rule \textsc{assert-ok$_S$}) nor
    $\rho = \MKTABLE{S}{V}{\emptyset}$.

  \item If $e$ is a value then the result follows from
    Lemma~\ref{lem:value-comp}.
  \item If $ e = T\br{\ov{v}}.f_i \becomes \maintype{\Phi_i}(v_i) = e'$
    with $(\ov{f~\Phi}) = \fields_\black{\emptyset}(T)$.
    Pose  $\Phi_i = \typair{S_i}{V_i}$
    and assume 
    $ \Delta \stoup \Gamma \vdash
    v_i : U_i$ for each $v_i$ in $\ov{v}$ ($U_i \assg_\Delta S_i$ by Rule \textsc{t-literal} and type preservation).

    We first do a case analysis to identify $d$ (the compilation
    target). There are two cases, depending on the boxy-ness of the
    field by rules \textsc{select-unboxed} and \textsc{select-boxed}.

    \begin{itemize}
    \item If $\neg \uinterface(V_i) \lor \utyvar_\emptyset(S_i)$ then we have
      \begin{equation}
        d = \ghost{ {\tycomp{T}{\Delta}}} \{ \ov{n} \} . i
        \qquad \text{with }
        \Delta \stoup \Gamma \vdash
        \ov{v
          \mapsto_\Phi
          n}
      \end{equation}

      Next, we determine the compilation target $d''$ of $e' =  S_i(v_i)$,
      \[
        \Delta \stoup \Gamma \vdash S_i(v_i) \mapsto d''
      \]

      From Rule \textsc{Conversion}, we have
      \begin{equation}\label{eq:sim-conv}
        \Delta \stoup \Gamma \vdash
        v_i : U_i
        \mapsto_{\typair{S_i}{\bounds(S_i)}}
        d''
      \end{equation}

      We $\neg \uinterface(V_i) \lor \utyvar_\emptyset(S_i)$ by
      assumption, and $\utyvar_\emptyset(S_i)$ does not hold, so we
      have $\neg \uinterface(V_i)$ and thus $\neg \uinterface(S_i)$.
      Additionally, 
      $\neg\uinterface(U_i)$ (since $v_i$ is a value).

      Thus Rule \textsc{static-change} must apply
      in which case
      $\neg\uinterface(S_i)$, then
      \[
        d'' =
        \STATICCHANGE{n_i}{\tycomp{S_i}{\Delta}}{\tycomp{U_i}{\Delta}}
      \]

      Now, we have that, following the letters from
      Figure~\ref{fig:lem-simulation-ed}:
      \[
        d = d'=  \ghost{ {\tycomp{T}{\Delta}}} \{ \ov{n} \} . i
        \becomesast
        \ghost{ {\tycomp{T}{\Delta}}} \{ \ov{v'} \} . i
        \becomes
        d''' = 
        \ceqcast{v'_i}{ \STATICCHANGE{n'_i}{\tycomp{S_i}{\Delta}}{\tycomp{U_i}{\Delta}}} = d''
      \]
      since the first step ($\becomesast$) follows from
      Lemma~\ref{lem:value-comp} (via rule \textsc{context}), and the
      second step ($\becomes$) follows from rule \textsc{select}.
      The $\ceqcast{}{}$-equality follows from
      Lemma~\ref{lem:value-comp} as each $n_i$ is also the
      result of compiling a value ($v_i$), we also have that
      $\ceqcast{n_i}{v'_i}$.

    \item If $\neg\utyvar_\emptyset(S_i) \land \uinterface(V_i)$ then we have
      \begin{equation}
        d = \ghost{ {\tycomp{T}{\Delta}}} \{ \ov{n} \} . i . (\typair{\tycomp{S_i}{}}{\tycomp{S_i}{}})
        \qquad \text{with }
        \Delta \stoup \Gamma \vdash
        \ov{v
          \mapsto_\Phi
          n}
      \end{equation}
      Next, we determine the compilation target $d''$ of
      $e' = S_i(v_i)$,
      \[
        \Delta \stoup \Gamma \vdash S_i(v_i) \mapsto d''
      \]
      From Rule \textsc{Conversion}, we have
      \begin{equation}\label{eq:sim-conv}
        \Delta \stoup \Gamma \vdash
        v_i : U_i
        \mapsto_{\typair{S_i}{\bounds(S_i)}}
        d''
      \end{equation}
      Note that we must have  $ v'_i = {\ghost{\typair{S_i}{V_i}}(W_i,\rho_i, v''_i)}$, with $\eqval{v_i}{v''_i}$, where the make-type is $W_i = \tycomp{S_i}{}$ if $\neg\uinterface(S_i)$ (Rule \textsc{make-bs}) and $W_i = \tycomp{U_i}{}$ if $\uinterface(S_i)$ (Rule \textsc{make-iface}).
      The following rules are applicable
      since
      $\neg\utyvar_\Delta(S_i) \land \uinterface(V_i)$ by assumption and 
      $\neg\uinterface(U_i)$ (since $v_i$ is a value):
      \begin{itemize}
      \item Rule \textsc{static-change}, in which case
        $\neg\uinterface(S_i)$, then
        \[
          d'' =
          \STATICCHANGE{n_i}{\tycomp{S_i}{\Delta}}{\tycomp{U_i}{\Delta}}
        \]

        Now, we have that, following the letters from
        Figure~\ref{fig:lem-simulation-ed}:
        \[
          \begin{array}{l}
            d = d'=  \ghost{ {\tycomp{T}{\Delta}}} \{ \ov{n} \} . i . (\typair{\tycomp{S_i}{}}{\tycomp{S_i}{}})
            \becomesast
            \ghost{ {\tycomp{T}{\Delta}}} \{ \ov{v'} \} . i . (\typair{\tycomp{S_i}{}}{\tycomp{S_i}{}})
            \\
            \becomes {v'_i . (\typair{\tycomp{S_i}{}}{\tycomp{S_i}{}}) } \becomes v''_i =  
            \ceqcast{d'''}{ \STATICCHANGE{n'_i}{\tycomp{S_i}{\Delta}}{\tycomp{U_i}{\Delta}}} = d''
          \end{array}
        \]
        since the first step ($\becomesast$) follows from
        Lemma~\ref{lem:value-comp} (via rule \textsc{context}), and
        the second step ($\becomes$) follows from rule
        \textsc{select}, the third from rule \textsc{assert-ok$_S$}.
        The $\ceqcast{}{}$-equality follows from
        Lemma~\ref{lem:value-comp} as each $n_i$ is also the
        result of compiling a value ($v_i$), we also have that
        $\ceqcast{n_i}{v'_i}$.

      \item Rule \textsc{make-iface}, in which case  $\uinterface(S_i)$, then
        \[
          d'' =
          \MAKETYPE{n_i}{{\tycomp{\typair{S_i}{S_i}}{\Delta}}}{\tycomp{U_i}{\Delta}}{\rho}
        \]

        Now, posing $\rho' = \MKTABLE{\tycomp{U_i}{}}{\tycomp{S_i}{}}{\emptyset}$,
        we have that, following the letters from
        Figure~\ref{fig:lem-simulation-ed}:
        \[
          \begin{array}{l}
            d = d'=  \ghost{ {\tycomp{T}{\Delta}}} \{ \ov{n} \} . i . (\typair{\tycomp{S_i}{}}{\tycomp{S_i}{}})
            \becomesast
            \ghost{ {\tycomp{T}{\Delta}}} \{ \ov{v'} \} . i . (\typair{\tycomp{S_i}{}}{\tycomp{S_i}{}})
            \\
            \becomes {v'_i . (\typair{\tycomp{S_i}{}}{\tycomp{S_i}{}}) }
            \becomes \ghost{\typair{\tycomp{S_i}{}}{\tycomp{S_i}{}}}(\tycomp{U_i}{}, \rho', v''_i) =
            \ceqcast{d'''}{\MAKETYPE{n_i}{{\tycomp{\typair{S_i}{S_i}}{\Delta}}}{\tycomp{U_i}{\Delta}}{\rho}} = d''
          \end{array}
        \]
        since the first step ($\becomesast$) follows from
        Lemma~\ref{lem:value-comp} (via rule \textsc{context}), the
        second step ($\becomes$) follows from rule \textsc{select}, and
        the third step ($\becomes$) follows from rule
        \textsc{assert-ok$_I$}.  Indeed, $\uinterface(\tycomp{S_i}{})$
        and, by Rule \textsc{make-iface}, the make-type of $v'_i$ is
        $\tycomp{U_i}{}$ with $\tycomp{U_i}{} \imp_\emptyset
        \tycomp{S_i}{}$, so the assertion \emph{re-boxes} the value
        (rather than unboxing it) into
        $d''' = \ghost{\typair{\tycomp{S_i}{}}{\tycomp{S_i}{}}}(\tycomp{U_i}{}, \rho', v''_i)$.
        The $\ceqcast{}{}$-equality $\ceqcast{d'''}{d''}$ follows from
        Lemma~\ref{lem:value-comp} and Figure~\ref{fig:eqcast-relation}:
        since each $n_i$ is also the result of compiling a value
        ($v_i$), we have $\ceqcast{n_i}{v'_i}$, and $d''$ reduces to a
        value $\ceqcast{}{}$-equal to $d'''$.
      \end{itemize}
    \end{itemize}

  \item If $ e = v . (T) \becomes v  = e'$
    with ${ \vtype(v) \imp_{\black{\emptyset}} T }$.
    The compilation target $d$ of $e$ is unique:
    \[ d = n . ({\tycomp{T}{\Delta}@\tycomp{\bounds_\Delta(T)}{\Delta}}) \qquad\qquad
      \text{where } \Delta \stoup
      \Gamma \vdash v \mapsto n
    \]
    Thus, the compilation target of $e' = v$ is $n$.

    Hence following
    the letters from Figure~\ref{fig:lem-simulation-ed}, we have either

    \begin{itemize}
    \item  If $T$ is $\neg\uinterface({\bounds_\Delta(T)})$ then
      \[
        \ceqcast{d}{d'} =
        n . ({\tycomp{T}{\Delta}@\tycomp{\bounds_\Delta(T)}{\Delta}})
        \becomesast
        \ghost{\Phi}(T,\rho, v')
        . ({\tycomp{T}{\Delta}@\tycomp{\bounds_\Delta(T)}{\Delta}})
        \becomes
        \ceqcast{ \ghost{\Phi}(T,\rho, v')}{n} = d''' = d''     
      \]
      With the first step ($\becomesast$) following from
      Lemma~\ref{lem:value-comp} (via rule \textsc{context}), and the
      second step ($\becomes$) follows from rule \textsc{assert-ok$_S$}.
    \item If $T$ is $\uinterface({\bounds_\Delta(T)})$ then
      \[
        \begin{array}{c}
        \ceqcast{d}{d'} =
        n . ({\tycomp{T}{\Delta}@\tycomp{\bounds_\Delta(T)}{\Delta}})
        \becomesast
        \ghost{\Phi}(T,\rho, v') . ({\tycomp{T}{\Delta}@\tycomp{\bounds_\Delta(T)}{\Delta}})
        \\ \becomes       
          \ceqcast{ \ghost{{\tycomp{T}{\Delta}@\tycomp{\bounds_\Delta(T)}{\Delta}}}(T,\rho, v')}{n} = d''' = d''
        \end{array}
      \]
      With the first step ($\becomesast$) following from
      Lemma~\ref{lem:value-comp} (via rule \textsc{context}), and the
      second step ($\becomes$) follows from rule  \textsc{assert-ok$_I$}.
    \end{itemize}
    Recall that $\ceqcast{v'}{n}$ from Lemma~\ref{lem:value-comp} since
    $n$ is the compilation target of a value and reduces to $v'$.

  \item If $e = E[f] \becomes E[g] = e'$ with $f \becomes g$
    
    By induction hypothesis on $f \becomes g$, we have
    $\Delta; \Gamma \vdash f \mapsto n$ with $\ceqcast{n}{n'}$, and
    $\Delta; \Gamma \vdash g \mapsto n''$ with $\ceqcast{n''}{n'''}$,
    and $n' \becomesast n'''$.

    Since compilation preserves the structure of terms,
    Lemma~\ref{lem:compilation-context-eqcast} yields an LWG context
    $E'$ such that
    $\Delta; \Gamma \vdash E[f] \mapsto E'[n]$
    and
    $\Delta; \Gamma \vdash E[g] \mapsto E'[n'']$;
    so $d = E'[n]$ and $e' = E[g]$ compiles to $d'' = E'[n'']$.
    As $\ceqcast{}{}$ is a congruence and $\ceqcast{d}{d'}$ with
    $d = E'[n]$, it suffices to take $d' = E'[n']$, recalling
    $\ceqcast{n}{n'}$ from the induction hypothesis.

    Thus we have, following the letters from
    Figure~\ref{fig:lem-simulation-ed}:
    \[
      d' = E'[n'] \becomesast E'[n'''] = d'''
      \qquad\text{and}\qquad
      d'' = E'[n''] \;\text{ with }\; \ceqcast{d''}{d'''} ,
    \]
    where $d' \becomesast d'''$ follows from $n' \becomesast n'''$ by
    rule \textsc{context}, and
    $\ceqcast{d''}{d'''}$, i.e.\ $\ceqcast{E'[n'']}{E'[n''']}$, follows
    from $\ceqcast{n''}{n'''}$ by congruence (both from the induction
    hypothesis).
    Hence $(E[g], d''') = (E[g], E'[n''']) \in \mathcal{R}$ as required.
  \item If $e = v.m[S](\ov w) \rightarrow U(e_0[x := v, \ov {y := T(w) }]) = e'$ 
    where $(x, \ov{ y : T} ).e_0 : U = body(\vtype(v).m[S])$.

    Next we find $d$. There are four possibilities depending on the rule being used.
    \begin{itemize}
    \item Rule \textsc{Call$_S$-boxed}
      
      Assuming we have
      \begin{itemize}
      \item $\vtype(v) = t[\ov{R}]$
      \item
        $m[\ov{\beta~S''}](\ov{y\:\Phi})\: {\color{black}
          \typair{U}{V} } \mathbin{\in} \methods_\Delta( t[\ov{R}])$
        with $\uinterface(V)$
      \item  $\mjump(t.m) = [\ov{\alpha}, \ov{\beta}](x,\ov{y}).e'_0$
      \end{itemize}
      
      By Lemma~\ref{lem:value-comp}, we have
      \begin{itemize}
      \item  $\vdash v \mapsto n$ with $n \becomesast v'$ and $\ceqcast{n}{v'}$
      \item  $\vdash \ov {w \mapsto_{\Phi[\ov{\beta := S}]} n'}$ with $\ov{n' \becomesast w'}$ and $\ov{\ceqcast{n'}{w'}}$
      \end{itemize}

      We have
      
      \[
        \begin{array}{lcl}
          d = d' & = &
                       n \smsep t[\ov{ \tycomp{R}{\Delta}}] .  m[\ov{\tycomp{S}{\Delta}}]( \ov{n'}) .  (\typair{\tycomp{U[\eta]}{\Delta}}{\tycomp{U[\eta]}{\Delta}}) 
          \\
                 &    \becomesast &
                                    v' \smsep t [\ov{ \tycomp{R}{\Delta}}] .  m[\ov{\tycomp{S}{\Delta}}]( \ov{w'}) .  (\typair{\tycomp{U[\eta]}{\Delta}}{\tycomp{U[\eta]}{\Delta}}) 
          \\
                 &    \becomes&
                                e'_0 [\ov{\alpha := \tycomp{R}{}}][\ov{\beta := \tycomp{S'}{}}][x := v'][\ov{y:= w'}] .  ({\tycomp{U[\eta]}{\Delta}}@{\tycomp{U[\eta]}{\Delta}})  = d'''
        \end{array}
      \]

      Assume
      $e' \mapsto d''$, 
      we have $\ceqcast{d''}{d'''}$ from Lemmas~\ref{lem:compositionality} and~\ref{lem:bodies}.

    \item The case for Rule \textsc{Call$_S$-unboxed} is the same as above without the cast $({\tycomp{U[\eta]}{\Delta}}@{\tycomp{U[\eta]}{\Delta}}) $
    \item Rule \textsc{Call$_I$}.
      Assuming we have
      \begin{itemize}
      \item $\vtype(v) = t[\ov{R}]$
      \item $v' = \Phi(S, \rho, v'')$
      \item $M = [\ov{\beta~S''}](\ov{y\:\typair{T}{T'}})\: {\color{black} \typair{U}{V} }$
      \item $mM \mathbin{\in} \methods_\Delta( t[\ov{R}])$
      \item  $\mjump(\rho(m_D)) = [\ov{\alpha}, \ov{\beta}](x,\ov{y}).e'_0$
      \item  $\mjump(t.m) = [\ov{\alpha}, \ov{\beta}](x,\ov{y}).e'_1$
      \item $\func~(x~t[\ov{\alpha~R'}])~m_D[\ov{\beta~S'}](\ov{y~\tycomp{\tybox{\typair{T}{T'}}{\Delta}}{\Delta}})~\tycomp{\tybox{ \typair{U}{V} }{\Delta}}{\Delta}~\br{\, \return~\textit{Adapt}((x~t[\ov{\alpha~R}])~mM)} $
      \end{itemize}

      By Lemma~\ref{lem:value-comp}, we have
      \begin{itemize}
      \item  $ \vdash v \mapsto n$ with $n \becomesast v'$ and $\ceqcast{n}{v'}$
      \item  $\vdash \ov {w \Mapsto_{\Phi[\ov{\beta := S}]} n'}$ with $\ov{n' \becomesast w'}$ and $\ov{\ceqcast{n'}{w'}}$.
        Note that as by definition of $\Mapsto$ we have that each $w'_i$ is an interface value.
      \end{itemize}

      \begin{itemize}
      \item If $\uinterface(V)$,
        we have
        \[
          \begin{array}{lcl}
            d = d' & = &
                         n \smsep t[\ov{ \tycomp{R}{\Delta}}] .  m_D[\ov{\tycomp{S}{\Delta}}]( \ov{n'}) .  (\typair{\tycomp{U[\eta]}{\Delta}}{\tycomp{U[\eta]}{\Delta}}) 
            \\
                   &    \becomesast & v' \smsep t[\ov{ \tycomp{R}{\Delta}}] .  m_D[\ov{\tycomp{S}{\Delta}}]( \ov{w'}) .  (\typair{\tycomp{U[\eta]}{\Delta}}{\tycomp{U[\eta]}{\Delta}}) 
            \\
                   &    \becomes &   e'_0 [\ov{\alpha := \tycomp{R}{}}][\ov{\beta := \tycomp{S'}{}}][x := v'][\ov{y:= w'}] .  (\typair{\tycomp{U[\eta]}{\Delta}}{\tycomp{U[\eta]}{\Delta}})
            \\
                   &    = &   \textit{Adapt}((z~t[\ov{\alpha~R}])~mM) [\ov{\alpha := \tycomp{R}{}}][\ov{\beta := \tycomp{S'}{}}][x := v'][\ov{y:= w'}] .  (\typair{\tycomp{U[\eta]}{\Delta}}{\tycomp{U[\eta]}{\Delta}})
            \\
                   &  = & 
                          v' . t[\ov{\tycomp{R}{}}] \smsep m[\ov{\tycomp{S'}{}}](\ov{d})  .  (\typair{\tycomp{U[\eta]}{\Delta}}{\tycomp{U[\eta]}{\Delta}})
            \\
                   && \qquad \qquad \text{with }
                      \neg\uinterface(T'_i) \implies d_i = w'_i . (\typair{\tycomp{T_i}{\Delta}}{\tycomp{T_i}{\Delta}})
                      \text{ and }
                      \uinterface(T'_i) \implies d_i = w'_i
            \\
                   &    \becomesast &  v' . t[\ov{\tycomp{R}{}}] \smsep m[\ov{\tycomp{S'}{}}](\ov{w''})  .  (\typair{\tycomp{U[\eta]}{\Delta}}{\tycomp{U[\eta]}{\Delta}})
            \\ && \qquad \qquad  \text{since}:
                  w'_i . (\typair{\tycomp{T_i}{\Delta}}{\tycomp{T_i}{\Delta}}) = \ghost{\tycomp{\typair{T_i}{\textit{Any}}}{}}(T_i, \emptyset, w''_i) (\typair{\tycomp{T_i}{\Delta}}{\tycomp{T_i}{\Delta}}) \becomes w''_i
                 \\ && \qquad \qquad  \text{ and } \ceqcast{w'_i}{w''_i}
            \\
                   &    \becomes &  e'_1 [\ov{\alpha := \tycomp{R}{}}][\ov{\beta := \tycomp{S'}{}}][x := v'][\ov{y:= w'}] .  ({\tycomp{U[\eta]}{\Delta}}@{\tycomp{U[\eta]}{\Delta}})  = d'''
          \end{array}
        \]

        Assume
        $e' \mapsto d''$, 
        we have $\ceqcast{d''}{d'''}$ from Lemmas~\ref{lem:compositionality} and~\ref{lem:bodies}.

      \item If $\neg \uinterface(V)$, we have
        \[
          \begin{array}{lcl}
            d = d' & = &
                         n \smsep t[\ov{ \tycomp{R}{\Delta}}] .  m_D[\ov{\tycomp{S}{\Delta}}]( \ov{n'}) .  (\typair{\tycomp{U[\eta]}{\Delta}}{\tycomp{U[\eta]}{\Delta}}) 
            \\
                   &    \becomesast & v' \smsep t[\ov{ \tycomp{R}{\Delta}}] .  m_D[\ov{\tycomp{S}{\Delta}}]( \ov{w'}) .  (\typair{\tycomp{U[\eta]}{\Delta}}{\tycomp{U[\eta]}{\Delta}}) 
            \\
                   &    \becomes &   e'_0 [\ov{\alpha := \tycomp{R}{}}][\ov{\beta := \tycomp{S'}{}}][x := v'][\ov{y:= w'}] .  (\typair{\tycomp{U[\eta]}{\Delta}}{\tycomp{U[\eta]}{\Delta}})
            \\
                   &    = &   \textit{Adapt}((z~t[\ov{\alpha~R}])~mM) [\ov{\alpha := \tycomp{R}{}}][\ov{\beta := \tycomp{S'}{}}][x := v'][\ov{y:= w'}]
                            .  (\typair{\tycomp{U[\eta]}{\Delta}}{\tycomp{U[\eta]}{\Delta}})
            \\
                   &    = & 
                            \MAKETYPE{v' . t[\ov{\tycomp{R}{}}] \smsep m[\ov{\tycomp{S'}{}}](\ov{d}) }{{\tycomp{\typair{V}{\texttt{Any}}}{\Delta}}}{\tycomp{V}{\Delta}}{\emptyset} .  (\typair{\tycomp{U[\eta]}{\Delta}}{\tycomp{U[\eta]}{\Delta}})
            \\
                   && \qquad \qquad \text{with }
                      \neg\uinterface(T'_i) \implies d_i = w'_i . (\typair{\tycomp{T_i}{\Delta}}{\tycomp{T_i}{\Delta}})
                      \text{ and }
                      \uinterface(T'_i) \implies d_i = w'_i
            \\
                   &    \becomesast &  \MAKETYPE{v' . t[\ov{\tycomp{R}{}}] \smsep m[\ov{\tycomp{S'}{}}](\ov{w''}) }{{\tycomp{\typair{V}{\texttt{Any}}}{\Delta}}}{\tycomp{V}{\Delta}}{\emptyset} .  (\typair{\tycomp{U[\eta]}{\Delta}}{\tycomp{U[\eta]}{\Delta}}) 
            \\ && \qquad \text{since}:
                  w'_i . (\typair{\tycomp{T_i}{\Delta}}{\tycomp{T_i}{\Delta}}) = \ghost{\tycomp{\typair{T_i}{\textit{Any}}}{}}(T_i, \emptyset, w''_i) (\typair{\tycomp{T_i}{\Delta}}{\tycomp{T_i}{\Delta}}) \becomes w''_i
                  \text{ and } \ceqcast{w'_i}{w''_i}
            \\
                   &    \becomes &
                                   \MAKETYPE{ e'_1 [\ov{\alpha := \tycomp{R}{}}][\ov{\beta := \tycomp{S'}{}}][x := v'][\ov{y:= w'}] }{{\tycomp{\typair{U}{\texttt{Any}}}{\Delta}}}{\tycomp{U}{\Delta}}{\emptyset} .  ({\tycomp{U[\eta]}{\Delta}}@{\tycomp{U[\eta]}{\Delta}})   = d'''
          \end{array}
        \]
        Assume
        $e' \mapsto d''$, 
        we have $\ceqcast{d''}{d'''}$ from Lemmas~\ref{lem:compositionality} and~\ref{lem:bodies}.
      \end{itemize}
    \item The case for Rule \textsc{Call$_{I\alpha}$} is similar to above
      with the synthetic cast
      $\CHANGETYPE{\_}{\ghost{{\tycomp{\Psi[\eta]}{\Delta}}}}{{\tycomp{\maintype{\Psi}[\eta]}{\Delta}}}$
      instead of $\_ .  (\typair{\tycomp{U[\eta]}{\Delta}}{\tycomp{U[\eta]}{\Delta}})$.
    \end{itemize}

  \item If $e = \circ(\ov{v}) \becomes v' = e'$, with $\delta(\circ, \ov{v}) = v'$.
    Assume $\emptyset  \stoup \emptyset \vdash \ov{v } : T$.

    \begin{itemize}
    \item Rule \textsc{t-op-param} cannot have applied here (we cannot have
      $d = \circ(\ov{e'})$ with $e: \alpha$ as terms are closed).w

    \item With Rule \textsc{t-op}, we have $d = \circ\smsep {\tycomp{T}{\Delta}}(\ov{f})$ with
      $\emptyset  \stoup \emptyset \vdash \ov{v \mapsto f}$.

      By Lemma~\ref{lem:value-comp}, we have
      $\emptyset  \stoup \emptyset \vdash \ov{v \mapsto f}$ with $\ov{f \becomesast w}$ and $\ov{\ceqcast{f}{w}}$ and $\ov{\eqval{v}{w}}$.
      Assume  $\delta(\circ_T, \ov{w}) = w'$.
      By Lemma~\ref{lem:operator-eqval}, we have $\eqval{v'}{w'}$.

      Then, we have
      \[
        d =  \circ\smsep {\tycomp{T}{\Delta}}(\ov{f})
        \becomesast
        \circ\smsep {\tycomp{T}{\Delta}}(\ov{w})
        \becomesast
        w'
      \]
      as required.
    \end{itemize}
    
  \item If  $e= T(U(c)) \becomes  T(c) = e'$, with $\ubase(T)$.
    Following Rule \textsc{conversion}, we have
    \begin{itemize}
    \item $\emptyset \stoup\emptyset \vdash T(U(c)) \mathbin{:} T$
    \item
      $\emptyset \stoup\emptyset \vdash U(c) \mathbin{:} U
      \mapsto_{\typair{T}{T}} d$
    \item $d = \STATICCHANGE{U(c)}{\tycomp{T}{\Delta}}{\tycomp{U}{\Delta}}$ by Rule \textsc{static-change}
    \item ${\tycomp{T}{\Delta}} = T$ and ${\tycomp{U}{\Delta}}$ since $T$ and $U$ are base types.
    \end{itemize}
    Hence we have
    \[
      d = d' = \STATICCHANGE{U(c)}{T}{U} \becomes T(c) = d'''
    \]
    and we have $d'' = \eqval{T(c)}{T(c)} = e'$, as required.
  \item If $ e= T(U\br{\ov{v}}) \becomes T\br{ \ov v } = e'$, with $\ustruct(T)$.
    Following Rule \textsc{conversion}, we have
    \begin{itemize}
    \item $\emptyset \stoup\emptyset \vdash T(U\br{\ov{v}}) \mathbin{:} T$
    \item
      $\emptyset \stoup\emptyset \vdash U\br{\ov{v}} \mathbin{:} U
      \mapsto_{\typair{T}{T}} d$
    \item $d = \STATICCHANGE{U\br{\ov{v}}}{\tycomp{T}{\Delta}}{\tycomp{U}{\Delta}}$ by Rule \textsc{static-change}
    \end{itemize}
    Hence we have
    \[
      d = d' = \STATICCHANGE{U\br{\ov{v}}}{\tycomp{T}{\Delta}}{\tycomp{U}{\Delta}}\becomes {\tycomp{T}{\Delta}\br{\ov{v}}}  = d''' 
    \]    
    and we have $d'' = \eqval{ T\br{ \ov v }}{\tycomp{T}{\Delta}\br{\ov{v}}} = e'$, as required.
  \item If $e=T(U\br{\ov{v}}) \becomes U\br{ \ov v } = e'$, with $ \uinterface(T)$.
    Following Rule \textsc{conversion}, we have
    \begin{itemize}
    \item $\emptyset \stoup\emptyset \vdash T(U\br{\ov{v}}) \mathbin{:} U$
    \item
      $\emptyset \stoup\emptyset \vdash U\br{\ov{v}} \mathbin{:} U
      \mapsto_{\typair{T}{T}} d$
    \item $d = \MAKETYPE{U\br{\ov{v}}}{{\tycomp{\typair{T}{T}}{\Delta}}}{\tycomp{U}{\emptyset}}{\rho}$ by Rule \textsc{make-iface}
    \end{itemize}
    Hence we have
    \[
      d = d' =  \MAKETYPE{U\br{\ov{v}}}{{\tycomp{\typair{T}{T}}{\Delta}}}{\tycomp{U}{\emptyset}}{\rho}
      \becomes
      \ghost{ \tycomp{\typair{T}{T}}{}}(U, \rho, U\br{\ov{v}})
      = d'''      
    \]
    and we have $d'' = \eqval{U\br{\ov{v}}}{\ghost{ \tycomp{\typair{T}{T}}{}}(U, \rho, U\br{\ov{v}})} = e'$, as required.
    The case  when $e=T(U(c))$ works similarly.
    \qedhere
  \end{itemize}
\end{proof}

\begin{definition}[Lifting Simulation]\label{def:lifting-sim}
  Given a WG program $\ov{D} \prog e$ and a LWG program $\ov{D^\dagger} \prog d$.
  A relation $\mathcal{R}$ is a \emph{lifting simulation} when for all $(d,e) \in \mathcal{R}$
  \begin{itemize}
  \item if $d$ is a value then $\eqval{e}{d}$,
  \item if $d \not\becomes  $ then $e \not\becomes$, and
  \item if $d \becomes d'$ then there is $e'$ such that $e \becomes^\ast e'$ and $(d',e') \in \mathcal{R}$.
  \end{itemize}
\end{definition}

\begin{figure}
  \begin{tikzpicture}[
    auto,
    process/.style={minimum size=0.5cm},
    >=latex
    ]

    \node[process] (P1) at (0,-2) {$e$};
    \node[process, gray] (P2) at (2,-2) {$e'$};
    \node[process] (Q1) at (0,0) {$d'$};
    \node[process, gray] (Q2) at (2,0) {$d''$};
    \node[process] (R1) at (0,2) {$d$};
    \node[process,gray] (R2) at (2,2) {$d'''$};

    \draw[->] (P1) -- (P2) node[midway, below] {$\becomesast$};
    \draw[->] (R1) -- (R2) node[midway, above] {$\becomes$};
    \draw[->] (P1) -- (Q1) node[midway, left] {$\mapsto$};
    \draw[->] (P2) -- (Q2) node[midway, right] {$\mapsto$};
    \draw[->] (Q1) -- (R1) node[midway, left] {$\ceqcast{}{}$};
    \draw[->] (Q2) -- (R2) node[midway, right] {$\ceqcast{}{}$};
  \end{tikzpicture}
  \caption{Key diagram for the proof of Lemma~\ref{lem:simulation-de-lifting}.}
  \label{fig:lem-simulation-ed}
\end{figure}

\begin{lemma}\label{lem:simulation-de-lifting}
  Suppose $\ov{D} \ok$ and $\emptyset \stoup \emptyset \vdash e : T$ for some
  $T$.
  Then, the relation
  \[
    \mathcal{R} \eqdef
    \{ \,
    (d,e) \st
    \exists d' . \
    \emptyset \stoup \emptyset \vdash
    e
    \mapsto
    d'
    \text{ and }
    \ceqcast{d}{d'}
    \,
    \}
  \]
  is a lifting simulation.
\end{lemma}
\begin{proof}
  We show that
  \[
    \mathcal{R} \eqdef
    \{ \,
    (d,e) \st
    \exists d' . \
    \emptyset \stoup \emptyset \vdash
    e
    \mapsto
    d'
    \text{ and }
    \ceqcast{d}{d'}
    \,
    \}
  \]
  is a lifting simulation. We proceed by induction on the structure of
  the LWG reduction.

  Figure~\ref{lem:simulation-de-lifting} gives an overview of our proof
  strategy.
  We start from $d$, which evaluates (one-step) to $d'''$. We show
  that there must be a WG term $e$ which must compile to an
  $\ceqcast{}{}$-equivalent term, i.e., $d'$.
  We then show that $e$ reduces (in zero or more steps) to $e'$ which
  compiles to $d'''$, possibly via an $\ceqcast{}{}$-equivalent term,
  i.e., $d'''$.
  \begin{itemize}
  \item If
    $ d= \ghost{\tycomp{S}{}}\{\ov{v} \}.i \rightarrow v_i =
    d'''$. Then we must show that $e$ reduces to an equivalent value.
    Then we must have $e = S\br{\ov{w}} . f_i $ with $\ov{\eqval{w}{v}}$.
    By inversion, we must have 
    $ d ' = S\br{\ov{n}} . f_i . (\maintype{\Phi_i})$ (if \textsc{select-boxed} was
    applied) or $ d ' = S\br{\ov{n}} . f_i$ (if
    \textsc{select-unboxed} was applied);
    assuming
    $    \ov{f~\Phi} = \fields_\Delta(S)$
    and 
    $\ov{\Delta\stoup\Gamma \vdash w  \mathbin{:} U \mapsto_\Phi n}$, with $\ov{U \imp_\Delta \maintype{\Phi}}$.
    Via Lemma~\ref{lem:value-comp}, we have
    \[
      S\br{\ov{n}} . f_i  \becomesast \eqval{w_i}{v_i}
    \]
    and since $U_i \imp_\Delta \maintype{\Phi_i}$, we also have
    \[
      S\br{\ov{n}} . f_i . (\maintype{\Phi_i}) \becomesast w_i . (\maintype{\Phi_i}) \becomes \eqval{w_i}{v_i}
    \]
    as required.
  \item $d = \MAKETYPE{w}{\Phi}{\tycomp{S}{}}{\rho}
    \rightarrow
    \ghost{\Phi}(\tycomp{S}{},\rho, w ) = d'''$.
    Here we have that $w = \tycomp{S}{}(c)$ or
    $w = \tycomp{S}{}\br{\ov{v}}$, and $d'$ can be any term
    $\ceqcast{}{}$-equivalent to $w$ (i.e., $w$ with additional
    synthetic casts).
    Hence, from the compilation rules, it must be the case that
    $e = S(c)$ or $e = S\br{\ov{v}}$, $e = e'$, and $d = d''$.
    Thus the case follows trivially with
    $\eqval{e}{\ghost{\Phi}(\tycomp{S}{},\rho, w )}$.
  \item $d = \CHANGETYPE{\ghost\Phi(\tycomp{S}{},\rho , w)}{\ghostytypair{T'}{U}}{T}
    \rightarrow
    \ghost{{\typair{T'}{U}}}( \tycomp{S}{}, \rho' , w) = d'''$
    with $  \rho' = \MKTABLE{\tycomp{S}{}}{U}{\emptyset}$.
    Hence, from the compilation rules, it must be the case that
    $e = S(c)$ or $e = S\br{\ov{v}}$, $e = e'$, and $d = d''$.
    Thus the case follows trivially with
    $\eqval{e}{\ghost{\Phi}(\tycomp{S}{},\rho, w )}$.
  \item Cases for Rules \textsc{static-change-s} and 
    \textsc{static-change-c}  are similar
    to the previous two cases.
  \item $d = \ghost\Phi(\tycomp{S}{}, \rho, w) .(\tycomp{S}{}) \rightarrow w = d'''$, via Rule \textsc{assert-ok${}_S$}.
    We must have $w = \tycomp{S}{}(c)$ or $w = \tycomp{S}{}\br{\ov{v}}$,
    and $d'$ can be any term $\ceqcast{}{}$-equivalent to $w$ (i.e., $w$
    with additional synthetic casts).
    There are two cases:
    \begin{itemize}
    \item   $e = S(c)$ or $e = S\br{\ov{v}}$, $e = e'$, and $d = d''$.
      Thus the case follows trivially with
      $\eqval{e}{\ghost{\Phi}(\tycomp{S}{},\rho, w ) . (\typair{\tycomp{S}{}}{\tycomp{S}{}})}$.
    \item If we have $e = v' . (S)$ with $v' = S(c)$ or
      $v' = S\br{\ov{v}}$., then
      we have $e' = v'$, and thus the case follows
      trivially with $\eqval{v'}{w}$.
    \end{itemize}

  \item $d = \ghost{\typair{U}{V'}}(\tycomp{S}{}, \rho', w)
    .(U')
    \rightarrow
    \ghost{{\typair{U'}{U'}}}(\tycomp{S}{}, \rho, w) = d'''$, with 
    $\tycomp{S}{} \imp_\emptyset U'$ and $\rho = \MKTABLE{\tycomp{S}{}}{U'}{\emptyset} $
    via Rule \textsc{assert-ok$_I$}.
    There are two cases:
    \begin{itemize}
    \item   $e = S(c)$ or $e = S\br{\ov{v}}$, $e = e'$, and $d = d''$.
      Thus the case follows easily with
      $\eqval{e}{\ghost{\Phi}(\tycomp{S}{},\rho, w ) . (U')}$.
    \item If we have $e = v' . (\revtycomp{U'})$ with $v' = S(c)$ or
      $v' = S\br{\ov{v}}$, then the WG transition fires since we have
      $\tycomp{S}{} \imp_\emptyset U' \implies S \imp_\emptyset
      \revtycomp{U'}$ (Lemma~\ref{lem:tycomp-preserve-subtyping}).
      Thus we have $e' = v'$, and thus the case follows
      trivially with $\eqval{v'}{w}$.
    \end{itemize}
  \item $d =  w \smsep t[\ov{\tycomp{R}{}}] . m[\ov{\tycomp{S}{}}](\ov{v})
    \rightarrow
    e[\ov{\alpha := \tycomp{R}{}}][\ov{\beta := \tycomp{S}{}}][x := v][\ov{y := v}]$
    with
    $[\ov\alpha \stoup \ov\beta](x,  \ov{y}).e = \mjump(t.m)$.
    Then we must have
    $ e =  w . m[S](\ov{v'}) $ with $\ov{\eqval{v'}{v}}$.
    Posing
    $(x : \_ \, ,\ov{y: T}).e' : U =
    \mbody(\vtype(w).m\black{[\ov{S}]})$, the result follows via
    Lemmas~\ref{lem:compositionality} and~\ref{lem:bodies}.
  \item $d = \ghost{\Phi}(U, \rho, v). m[\ov{S}](\ov{v})
    \rightarrow
    e[\ov{\alpha := R}][\ov{\beta := S}][x := v][\ov{y := v}] = d'$ with
    $U = t[\ov{R}]$
    and
    $[\ov\alpha \stoup \ov\beta](x,  \ov{y}).e = \mjump(\rho(m))$
    Following Rule \textsc{d-func}, we can simply take $e = e'$ here,
    then the case reduces to a matching static method call.
  \item $d =  \circ(\ov{v}) \becomes v'' = d'$ with
    $\forall i . \,
    v_i = \ghost{\Phi}(S, \rho, v_i')$
    and
    $\delta(\rho(\circ), \ov{v'}) = v''$.
    Then $e = \circ (\ov{v})$ and $\vtype(v_i) = \maintype{\Phi}$ and the result
    follows with Lemma~\ref{lem:operator-eqval}.
    
  \item $d = \circ\smsep T(\ov{v}) \becomes v' = d'$ with $ \delta(\circ_T, \ov{v}) = v'$.
    Then $e = \circ (\ov{v})$ and $\vtype(v_i) = T$ and the result
    follows with Lemma~\ref{lem:operator-eqval}.
  \item $d = E[n] \rightarrow E[m] = d'$ with $n \becomes m$.
    By induction hypothesis on $n \becomes m$, we have        
    $\Delta; \Gamma \vdash n'' \mapsto \ceqcast{n'}{n}$; and
    $\Delta; \Gamma \vdash m'' \mapsto \ceqcast{m'}{m}$;
    and $n'' \becomesast m''$.
    Since compilation preserve the structure of terms, we must have $E'$ such that
    $e= E'[n''] \becomesast E'[m''] = e'$, as required.
  \end{itemize}
\end{proof}

\begin{lemma}\label{lem:synth-type-assert-reduce}
  Let $e$ be a WG expression, with
  $\emptyset ; \emptyset \vdash e : T$, such that $e \mapsto d$, then
  for all $d'$ such that $d \longrightarrow^{\ast} d'$
  either $d'$ is a value, $d' \longrightarrow d''$ for some $d''$, or
  $d'$ panics on a \emph{non-synthetic} type-assert.
\end{lemma}
\begin{proof}
  By induction on the compilation, we show that for any synthetic cast
  of the form $d = E[n.\synthcast{(\tycomp{\typair{T}{V}}{\Delta})}]$
  with $\Delta ; \Gamma \vdash n : \typair{T}{U}$, the assertion
  $n.\synthcast{(\tycomp{\typair{T}{V}}{\Delta})}$ reduces (i.e.\ never
  panics).
  There are three places where such casts are added in the compilation.
  \begin{itemize}
  \item
    $ \Delta\stoup\Gamma \vdash e . f \mapsto e' . i . {\color{black}
      (\typair{\tycomp{S_i}{\Delta}}{\tycomp{S_i}{\Delta}}) }$ with
    $\ov{f\; \typair{S}{V}} = \fields_\Delta(T)$, in the compilation
    of field select (Rule \textsc{Select-boxed}).
  \item
    $ e' .  m_D[\ov{{\tycomp{S}{\Delta}}}](\ov{d'}) .  {\color{black}
      (\typair{\tycomp{T[\eta]}{\Delta}}{\tycomp{T[\eta]}{\Delta}}) }$
    with
    $      m[\ov{\beta~S''}](\ov{y\:\Phi})\: {\color{black} \typair{T}{V} } \mathbin{\in} \methods_\Delta(R)$
    and
    $\eta = (\ov{\beta \by S})$ in the compilation of static and dynamic
    method calls (Rules \textsc{Call$_S$-boxed} and \textsc{Call$_I$}.
  \item
    $\neg\uinterface(T'_i) \implies d_i = y_i
    . (\tycomp{\typair{T_i}{T'_i}}{\Delta})$ in adapter methods (to
    unbox arguments that should not be boxed). Here we have that each
    $y_i$ has type${\tycomp{\tybox{\Phi_i}{\Delta}}{\Delta}}$
    with $\ov{\Phi} = \ov{\typair{T}{T'}}$ (see Figure~\ref{fig:methdecl}).
  \end{itemize}
  In each case the cast's main-type is exactly the type $T$ of $n$,
  so by Lemmas~\ref{lem:type-pairs-properties}
  and~\ref{lem:lwg-val-form} the value $n$ evaluates to
  $\ghost{\typair{T}{U}}(S, \rho, w)$ with $S \imp_\emptyset T$; hence
  Rule \textsc{assert-ok$_S$} (if $\neg\uinterface(T)$) or Rule
  \textsc{assert-ok$_I$} (if $\uinterface(T)$) applies.
  Thus, by LWG reduction, all such casts always reduce successfully.
\end{proof}

\begin{theorem}[$\ceqcast{}{}$ is a Weak Bisimulation]\label{thm:eqcast-bisimulation}
  The relation $$\ceqcast{}{}$$ is a weak bisimulation with respect to
  LWG reduction well-typed terms.
\end{theorem}
\begin{proof}
  We show that the relation
  \[
    \mathcal{B} \eqdef
    \{ \,
    (d_1, d_2) \st
    \ceqcast{d_1}{d_2}
    \,
    \}
  \]
  is a weak bisimulation (with $d_1$ and $d_2$ well-typed). We proceed
  by structural induction on the $\ceqcast{}{}$ derivation.

  \noindent
  \textbf{Case 1: Reflexivity.} 
  If $\ceqcast{d}{d}$, then trivially if $d \becomes d'$, we have $d \becomes d'$ and $\ceqcast{d'}{d'}$.

  \noindent
  \textbf{Case 2: Symmetry.} 
  If $\ceqcast{d_1}{d_2}$ is derived from $\ceqcast{d_2}{d_1}$, then the bisimulation property follows by swapping the roles of $d_1$ and $d_2$.

  \noindent
  \textbf{Case 3: Make operation equivalence.}
  If $\ceqcast{\MAKETYPE{e}{\Phi}{S}{\rho}}{\ghost{\Phi}(S,\rho, w)}$ with $\ceqcast{e}{w}$ where $w$ is a value:
  \begin{itemize}
  \item \textbf{Forward:} If $\MAKETYPE{e}{\Phi}{S}{\rho} \becomes d_1'$, then by LWG reduction rules:
    \begin{itemize}
    \item If $e$ is already the value $w$, then $\MAKETYPE{w}{\Phi}{S}{\rho} \becomes \ghost{\Phi}(S,\rho, w)$.
      The right-hand side is already this value, so we have $\ceqcast{\ghost{\Phi}(S,\rho, w)}{\ghost{\Phi}(S,\rho, w)}$.
    \item If $e \becomes e_1$, then $\MAKETYPE{e}{\Phi}{S}{\rho} \becomes \MAKETYPE{e_1}{\Phi}{S}{\rho}$.
      Since $\ceqcast{e}{w}$ and $w$ is a value, by induction hypothesis $e_1$ gets closer to $w$.
      Eventually $e \becomesast w$, so $\MAKETYPE{e}{\Phi}{S}{\rho} \becomesast \MAKETYPE{w}{\Phi}{S}{\rho} \becomes \ghost{\Phi}(S,\rho, w)$.
    \end{itemize}
  \item \textbf{Backward:} Since $\ghost{\Phi}(S,\rho, w)$ is a value, it cannot reduce further.
    This case is vacuous.
  \end{itemize}

  \noindent
  \textbf{Case 4: Synthetic type assertion equivalence.}
  If $\ceqcast{e}{e'.\synthcast{(\typair{T}{V})}}$ with appropriate typing conditions:
  \begin{itemize}
  \item \textbf{Forward:} If $e \becomes e_1$, we need to show $e'.(\synthcast{\typair{T}{V}}) \becomesast d_2'$ with $\ceqcast{e_1}{d_2'}$.
    By induction hypothesis,
    if $e$ diverges, so does $e'$, so we focus on the case when
    $e$ reduces to a value $w$.
    Then we have to show that the type assertion succeeds which follows from Lemma~\ref{lem:synth-type-assert-reduce}.
  \item \textbf{Backward:} By induction hypothesis, $e$ and $e'$
    either both correspondingly diverge, get stuck, or reduce to a
    value. We focus on the case where they reduce to a value.
    Pose $e \becomesast w$ and $e' \becomesast w'$, with
    $\ceqcast{w}{w'}$ by induction hypothesis.
    By Lemma~\ref{lem:synth-type-assert-reduce}, we have
    $w.(\synthcast{\typair{T}{V}}) \becomes w'$, and the result
    follows since $\ceqcast{w}{w'}$.
  \end{itemize}

  \noindent
  \textbf{Case 5: Synthetic cast.}
  The cases for
  ${\eqcast{e}{\MAKETYPE{e'}{\typair{T}{V}}{S}{\rho}}}$,   
  ${\eqcast{e}{\CHANGETYPE{e'}{\ghostytypair{T}{S}}{T}}}$, and
  ${\eqcast{e}{\STATICCHANGE{e'}{T}{U}}}$ (static-change) are all
  similar and straightforward since they only update the ghost type
  information.
  
  \noindent
  \textbf{Case 7: Congruence.}
  We consider each specific syntactic construct:

  \noindent
  \textbf{Case 7a: Variables.} If $\ceqcast{x}{x}$, this is trivial since variables cannot reduce.

  \noindent
  \textbf{Case 7b: Dynamic method calls.} If $\ceqcast{d.m[\ov{S}](\ov{f})}{d'.m[\ov{S}](\ov{f'})}$ with $\ceqcast{d}{d'}$ and $\ov{\ceqcast{f}{f'}}$:
  \begin{itemize}
  \item \textbf{Forward:} If $d.m[\ov{S}](\ov{f}) \becomes d_1'$:
    \begin{itemize}
    \item If reduction is in receiver context: $d \becomes d_1$, then by IH $d' \becomesast d_1''$ with $\ceqcast{d_1}{d_1''}$, so $d'.m[\ov{S}](\ov{f'}) \becomesast d_1''.m[\ov{S}](\ov{f'})$.
    \item If reduction is in argument context: $f_i \becomes f_i^1$, then by IH $f'_i \becomesast f_i'^1$ with $\ceqcast{f_i^1}{f_i'^1}$, so $d'.m[\ov{S}](\ov{f'}) \becomesast d'.m[\ov{S}](\ov{f'_1, \ldots, f_i'^1, \ldots})$.
    \item If method call reduces to body substitution when receiver and arguments are values, the same reduction applies by equivalence.
    \end{itemize}
  \item \textbf{Backward:} Symmetric argument.
  \end{itemize}

  \noindent
  \textbf{Case 7c: Static method calls.} If $\ceqcast{d \smsep T . m[\ov{S}](\ov{f})}{d' \smsep T . m[\ov{S}](\ov{f'})}$ with $\ceqcast{d}{d'}$ and $\ov{\ceqcast{f}{f'}}$:
  Similar to dynamic method calls but with static dispatch.

  \noindent
  \textbf{Case 7d: Field selection.} If $\ceqcast{d.i}{d'.i}$ with $\ceqcast{d}{d'}$:
  \begin{itemize}
  \item \textbf{Forward:} If $d.i \becomes d_1'$:
    \begin{itemize}
    \item If $d \becomes d_1$, then by IH $d' \becomesast d_1''$ with $\ceqcast{d_1}{d_1''}$, so $d'.i \becomesast d_1''.i$.
    \item If $d$ is a struct value $\ghost{T}\{\ov{v}\}$ and $d.i \becomes \maintype{\Phi_i}(v_i)$, then since $\ceqcast{d}{d'}$, we have $d'$ reduces to an equivalent struct, so the same field selection applies.
    \end{itemize}
  \item \textbf{Backward:} Symmetric argument.
  \end{itemize}

  \noindent
  \textbf{Case 7e: Type assertions.} If $\ceqcast{d.(T)}{d'.(T)}$ with $\ceqcast{d}{d'}$:
  Similar to field selection, but with type assertion reduction rules.

  \noindent
  \textbf{Case 7f: Dynamic operations.} If $\ceqcast{\circ(\ov{d})}{\circ(\ov{d'})}$ with $\ov{\ceqcast{d}{d'}}$:
  \begin{itemize}
  \item \textbf{Forward:} If $\circ(\ov{d}) \becomes d_1'$, either by reducing arguments in context or by applying the operation when all arguments are values. By IH, the same reduction sequence applies to $\circ(\ov{d'})$.
  \item \textbf{Backward:} Symmetric argument.
  \end{itemize}

  \noindent
  \textbf{Case 7g: Static operations.} If $\ceqcast{\circ\smsep T(\ov{d})}{\circ\smsep T(\ov{d'})}$ with $\ov{\ceqcast{d}{d'}}$:
  Similar to dynamic operations but with static type information.
  
  \noindent
  \textbf{Case 8: Structural equivalences.}  For assignability-based
  equivalences like $\ceqcast{T(c)}{U(c)}$ when $T \assg_\Delta U$,
  all are values so the result is immediate.
  \qedhere
\end{proof}

 \end{document}